\documentclass[11pt]{article}
\pdfoutput=1 
\usepackage[sc]{mathpazo}

 \usepackage[margin=1in]{geometry}
\usepackage[english]{babel}
\usepackage[utf8x]{inputenc}
\usepackage[compact]{titlesec}

\usepackage{cmap}
\usepackage[T1]{fontenc}
\usepackage{bm}
\usepackage{amsfonts}
\usepackage{amsmath}
\usepackage{amssymb}
\usepackage{amsthm}
\usepackage{float}
\usepackage{graphics}

\usepackage{hyperref}
\usepackage[svgnames]{xcolor}
\hypersetup{colorlinks={true},urlcolor={blue},linkcolor={DarkBlue},citecolor=[named]{DarkGreen},linktoc=all}
\usepackage{natbib}

\usepackage{tikz}  
\usetikzlibrary{arrows}
\usetikzlibrary{patterns,snakes}
\usetikzlibrary{decorations.shapes}
\usetikzlibrary{positioning,backgrounds}
\tikzstyle{overbrace text style}=[font=\tiny, above, pos=.5, yshift=5pt]
\tikzstyle{overbrace style}=[decorate,decoration={brace,raise=5pt,amplitude=3pt}]
\usetikzlibrary{shapes.geometric}

\usepackage{microtype}
\usepackage[capitalise,nameinlink,noabbrev]{cleveref}

\usepackage{doi}

\usepackage{enumitem}
\usepackage{array}

\theoremstyle{definition}
\newtheorem{definition}{Definition}

\theoremstyle{plain}
\newtheorem{theorem}{Theorem}[section]
\newtheorem{inftheorem}{Informal Theorem}
\newtheorem{lemma}[theorem]{Lemma}

\newtheorem{proposition}[theorem]{Proposition}
\newtheorem{fact}[theorem]{Fact}

\theoremstyle{remark}
\newtheorem{remark}{Remark}
\newtheorem*{remark*}{Remark}
\newtheorem{claim}{Claim}
\newtheorem*{claim*}{Claim}

\DeclareMathOperator*{\argmin}{argmin}
\DeclareMathOperator*{\argmax}{argmax}

\newcommand{\sz}{\textup{size}}

\newcommand{\tfnp}{\textup{TFNP}}
\newcommand{\ppa}{\textup{PPA}}
\newcommand{\ppad}{\textup{PPAD}}
\newcommand{\ppp}{\textup{PPP}}

\newcommand{\degtag}[1]{[\# #1]}
\newcommand{\ppak}[1]{\textup{PPA-$#1$}}
\newcommand{\neck}[1]{{\normalfont \scshape $#1$-Necklace-Splitting}}
\newcommand{\pstartucker}[1]{{\normalfont \scshape $\mathbb{Z}_{#1}$-star-Tucker}}
\newcommand{\imba}[1]{{\normalfont \scshape Imbalance-mod-$#1$}}
\newcommand{\ppakl}[2]{\ppak{#1\degtag{#2}}}
\newcommand{\bipa}[1]{\textsc{Bipartite-mod-$#1$}}
\newcommand{\bipal}[2]{\bipa{#1\degtag{#2}}}
\newcommand{\imbal}[2]{\imba{#1\degtag{#2}}}

\newcommand{\bss}{\textsc{BSS}}
\newcommand{\tucker}{\textsc{BSS-Tucker}}
\newcommand{\polygonTucker}[1]{{\normalfont \scshape $#1$-polygon-Tucker}}
\newcommand{\polygonBorsukUlam}[1]{{\normalfont \scshape $#1$-polygon-Borsuk-Ulam}}

\newcommand{\idx}{\mathsf{index}}
\newcommand{\val}{\mathsf{value}}
\newcommand{\chull}{\mathsf{conv}}

\newcommand{\bbR}{\mathbb{R}}
\newcommand{\bbZ}{\mathbb{Z}}
\newcommand{\bv}{\mathbf{v}}
\newcommand{\bx}{\mathbf{x}}
\newcommand{\bu}{\mathbf{u}}
\newcommand{\by}{\mathbf{y}}
\newcommand{\be}{\mathbf{e}}
\newcommand{\bz}{\mathbf{z}}
\newcommand{\bs}{\mathbf{s}}

\newcommand{\norm}[1]{\left\|#1\right\|}
\newcommand{\abs}[1]{\left|#1\right|}

\renewcommand{\vec}[1]{\boldsymbol{#1}}

\makeatletter
\newcommand{\thickhline}{%
    \noalign {\ifnum 0=`}\fi \hrule height 1.5pt
    \futurelet \reserved@a \@xhline
}
\newcolumntype{"}{@{\hskip\tabcolsep\vrule width 1.5pt\hskip\tabcolsep}}
\makeatother

\title{A Topological Characterization of Modulo-$p$ Arguments and Implications for Necklace Splitting}

\author{
\begin{tabular}{cc}
& \\
\textbf{Aris Filos-Ratsikas} & \textbf{Alexandros Hollender}\\
\small{University of Liverpool, United Kingdom} & \small{University of Oxford, United Kingdom}\\
\href{mailto:Aris.Filos-Ratsikas@liverpool.ac.uk}{\small{\texttt{Aris.Filos-Ratsikas@liverpool.ac.uk}}} & \href{mailto:Alexandros.Hollender@cs.ox.ac.uk}{\small{\texttt{Alexandros.Hollender@cs.ox.ac.uk}}}\\
& \\
\textbf{Katerina Sotiraki} & \textbf{Manolis Zampetakis}\\
\small{University of California Berkeley, USA} & \small{University of California Berkeley, USA}\\
\href{mailto:katesot@berkeley.edu}{\small{\texttt{katesot@berkeley.edu}}} & \href{mailto:mzampet@berkeley.edu}{\small{\texttt{mzampet@berkeley.edu}}}\\
& \\
\end{tabular}
}

\date{}

\begin{document}

\maketitle

\begin{abstract}

The classes \ppak{p} have attracted attention lately, because they are the main candidates for capturing the complexity of \emph{Necklace Splitting} with \emph{$p$ thieves}, for prime $p$. However, these classes were not known to have complete problems of a topological nature, which impedes any progress towards settling the complexity of the Necklace Splitting problem.
On the contrary, topological problems have been pivotal in obtaining completeness results for PPAD and PPA, such as the PPAD-completeness of finding a Nash equilibrium \citep{Daskalakis2009,chen2009settling} and the PPA-completeness of Necklace Splitting with \emph{$2$ thieves} \citep{FRG18-Necklace}.

In this paper, we provide the first \emph{topological characterization} of the classes \ppa-$p$. First, we show that the computational problem associated with a simple generalization of Tucker's Lemma, termed \textsc{$p$-polygon-Tucker}, as well as the associated Borsuk-Ulam-type theorem, \textsc{$p$-polygon-Borsuk-Ulam}, are \ppa-$p$-complete. Then, we show that the computational version of the well-known \emph{BSS Theorem} \citep*{BSS81}, as well as the associated \textsc{BSS-Tucker} problem are \ppa-$p$-complete. Finally, using a different generalization of Tucker's Lemma (termed \textsc{$\mathbb{Z}_p$-star-Tucker}), which we prove to be \ppa-$p$-complete, we prove that $p$-thief Necklace Splitting is in \ppa-$p$. This latter result gives a new combinatorial proof for the Necklace Splitting theorem, the only proof of this nature other than that of \citet{Meunier2014simplotopal}.

All of our containment results are obtained through a new combinatorial proof for $\mathbb{Z}_p$-versions of Tucker's lemma that is a natural generalization of the standard combinatorial proof of Tucker's lemma by \citet{FT81}. We believe that this new proof technique is of independent interest.

\end{abstract}

\clearpage
\tableofcontents

\section{Introduction}
The class \tfnp\ \citep{Megiddo1991} is the class of \emph{Total Search Problems in NP}, 
i.e., problems for which a solution is always guaranteed to exist, and can be verified 
in polynomial time. In a seminal paper, attempting to capture the complexity of numerous
interesting problems, \citet{Papadimitriou94-TFNP-subclasses} defined several subclasses of 
TFNP, such as \ppad, \ppa\ and \ppp, among others, each of which is associated with a different
existence principle. For example, \ppad\ is based on the principle that given a source in 
a directed graph with in-degree and out-degree at most 1, there must exist another vertex of
degree 1; \ppa\ is based on a similar principle on an undirected graph, and \ppp\ is based on
the pigeonhole principle.

Among those, \ppad\ has been largely  successful in capturing the complexity of many well-known 
problems, most prominently that of computing Nash equilibria in games 
\citep{Daskalakis2009,chen2009settling}. \ppa\ and \ppp\ have been more elusive in that regard
for nearly two decades, until the recent results of \citet{FRG18-Consensus,FRG18-Necklace} and
\citet{sotirakiZZ18}, who provided the first ``natural'' complete problems for these classes
respectively. Here, a ``natural'' problem is one that does not explicitly include a polynomial-sized
circuit in its definition (as termed in \citep{Grigni2001}).
Interestingly, the \ppa-complete problems in \citep{FRG18-Consensus,FRG18-Necklace} that solidified the status of \ppa\ as a class containing such natural problems were the well-known \emph{Necklace Splitting} problem for \emph{two thieves} as well as its continuous variant (coined the \emph{Consensus-Halving} problem in \citep{SS03-Consensus}). 

The Necklace Splitting problem is a classical problem in combinatorics, dating back to the mid-1980s and the works of \citet{Goldberg1985}, \citet{Alon1986} and \citet{Alon87-Necklace}, among others. In this problem, $k$ thieves are aiming to split an open necklace containing $n$ beads of $t$ different colors (with exactly $k\cdot a_i$ beads of color $i$, for some $a_i \in \mathbb{N}$), such that each thief receives exactly $a_i$ beads of color $i$. Furthermore, the thieves are allowed to use only $(k-1)t$ cuts to obtain this division.
A solution to the Necklace Splitting problem is guaranteed to exist, as was proven by \citet{Alon87-Necklace}. Earlier on, \citet{Goldberg1985} and \citet{Alon1986} had proven the existence of a solution for the case of $2$ thieves; this result invokes a fundamental tool from mathematics, the \emph{Borsuk-Ulam Theorem} \citep{Borsuk1933}. Alon's proof for the general case proceeds in two steps. First, using a simple argument, he proves that if the theorem holds for any prime number $p$ of thieves, it also holds for any other number of thieves. In the second step, which is significantly more involved, he proves the theorem for any prime number by using the \emph{BSS Theorem}, a generalization of the Borsuk-Ulam Theorem due to \citet*{BSS81}.

Questions about the computational complexity of finding a solution were raised explicitly as early as when the first existence results were proven \citep{Goldberg1985} and then later on by a series of papers \citep{alon1988some,alon1990non,Meunier2008,Meunier2009,Meunier2012}. The first definitive answer was provided by \citet{FRG18-Necklace}, via their \ppa-completeness result, following an initial \ppad-hardness result by \citet{filos2018hardness}. Crucially however, their result only applies to the case of two thieves.\footnote{The authors also extend the \ppa\ membership straightforwardly to numbers of thieves which are powers of $2$.} In fact, the authors observed (also referencing \citet{DeLongueville2006} and \citet{Meunier2014simplotopal} as previously having made similar observations) that the version of the problem with $p \geq 3$ thieves does not seem to boil down to the principle associated with the class \ppa. To this end, they conjectured that $p$-thief Necklace Splitting is complete for the computational class \ppa-$p$, also defined by \citet{Papadimitriou94-TFNP-subclasses}, in which the associated principle is the following; Given a vertex with degree which is not a multiple of $p$ in a bipartite graph, find another such vertex. It follows from the definition that $\ppak{2}=\ppa$.

The classes \ppak{p} have been very recently studied by \citet{Hollender19} and \citet{GoosKSZ2019}. The authors of the latter paper in fact provide the first \ppak{p}-completeness result for a natural problem, a computational version of the Chevalley-Warning theorem. Importantly, they were able to obtain their completeness result via reductions to and from equivalent variants of the canonical problems of the class. To prove any results about $p$-thief Necklace Splitting however, such an approach seems insufficient. 

To see this, note that the results for $p=2$, both for the \ppa-membership \citep{filos2018hardness} and for \ppa-hardness \citep{FRG18-Necklace} of the problem, are obtained via reductions to/from a computational version of \emph{Tucker's Lemma} \citep{tucker1945some}, a discrete analogue to the Borsuk-Ulam theorem, proven to be \ppa-complete by \citet{Papadimitriou94-TFNP-subclasses} and \citet{ABB15-2DTucker}. Tucker's lemma asserts that if we have an antipodally symmetric triangulation of a $d$-dimensional ball $B$ and a labeling function which assigns complementary labels to antipodal points on the boundary, then there is a \emph{complementary edge}, i.e., two adjacent points with equal-and-opposite labels. This connection is not a coincidence; for example, the idea in \citep{filos2018hardness} is in fact a direct adaptation of a combinatorial existence proof of \citet{SS03-Consensus} for the Consensus-Halving problem, which goes via Tucker's lemma.

In order to obtain a \ppak{p} result for $p$-thief Necklace Splitting, it seems rather imperative to develop an ``arsenal'' of computational problems of a related nature, that we could reduce to/from, namely generalizations of the Borsuk-Ulam theorem and Tucker's lemma. Such ``topological characterizations'' did not only enable researchers in settling the complexity of the problem (and some other related problems) for $\ppa = \ppak{2}$, but also facilitated the success of \ppad\ to the utmost extent, seeing as virtually all the related important results go via the computational versions of the associated topological theorems (specifically \emph{Brouwer's Fixed-Point Theorem} \citep{brouwer1911abbildung} and \emph{Sperner's Lemma} \citep{sperner1928neuer}).\footnote{To be more precise, several important \ppad-hardness results were obtained via reductions from the \emph{Generalized Circuit problem} \citep{Daskalakis2009,chen2009settling,rubinstein2018inapproximability}, which was however proven to be \ppad-complete via the aforementioned topological problems.} In a very related manner, \citet{SS03-Consensus} stressed the importance of obtaining such a generalization, in the quest for obtaining a combinatorial proof of \emph{Consensus-$1/k$-Division} (the generalization of Consensus-Halving) and therefore, for Necklace Splitting with $k$ thieves, for $k>2$. Lastly, \citet{GoosKSZ2019} explicitly raised the complexity of one of these generalizations, the BSS Theorem, as an open problem.

\subsection{Our Results}
In this paper, we obtain such a \emph{topological characterization} of the classes \ppak{p}, for all primes $p \geq 3$. Namely, we provide generalizations of the computational versions of the Borsuk-Ulam theorem and Tucker's lemma, parameterized by $p$, which are complete for \ppak{p}. A highlight of our generalizations is the \ppa-$p$ completeness of the computational version of the BSS Theorem. Finally, we use a further generalization to prove that Necklace Splitting with $p$ thieves lies in \ppak{p}.\\

\noindent The strength of our results lies in that

\begin{itemize}
    \item[$\diamond$] they solidify the status of the classes \ppa-$p$ as classes containing interesting well-known problems (adding to the recent results of \citet{GoosKSZ2019}), and
    \item[$\diamond$] they set up an essential toolkit for obtaining more completeness results for the classes, e.g., possibly the \ppa-$p$-completeness of $p$-thief Necklace Splitting.
\end{itemize}

\noindent All of our \ppak{p}-membership results are obtained via a new combinatorial proof for $\mathbb{Z}_p$-versions of Tucker's lemma. This new proof can be seen as a natural generalization of the standard combinatorial proof of Tucker's lemma by \citet{FT81}. Thus, as a byproduct of our techniques we also obtain the following results:
\begin{enumerate}
    \item A combinatorial proof of the BSS theorem. The original proof by \citet{BSS81} is not combinatorial, as it uses various tools from algebraic topology. Using our new technique, we are able to provide the first combinatorial proof for this theorem.
    \item A combinatorial proof of the Necklace Splitting theorem. The existence of such a combinatorial proof had been an open problem since \citep{Alon87-Necklace}. This open problem was solved by \citet{Meunier2014simplotopal} using a rather complicated argument. In contrast, our new combinatorial proof uses more elementary tools and is conceptually simpler than Meunier's proof. 
    \item A stronger statement of the continuous Necklace Splitting theorem which is called Consensus-$1/p$-Division \citep{Alon87-Necklace, SS03-Consensus}. The main advantage of our new theorem is that it actually works for valuation functions that are not necessarily additive and non-negative, for details see \cref{thm:strongConsensusDivision}.
\end{enumerate}
As a result, we believe that this new technique is of independent interest and will be useful for providing combinatorial proofs of other topological existence theorems such as Dold's Theorem \cite{Dold1983}. We remark here that although the original proof of the Necklace Splitting theorem in \citep{Alon87-Necklace} is via the BSS Theorem, our \ppak{p} membership result for Necklace Splitting does \emph{not} go via the \ppak{p} membership result that we prove for BSS. This is due to the fact that for several steps of the proof of \citet{Alon87-Necklace}, it is quite unclear whether they can be carried out in polynomial time. Instead, we construct a reduction directly from \pstartucker{p}, which we prove to be \ppak{p}-complete.

\smallskip

    \noindent In the remainder of this section, we informally state our main problems and 
results, and we give a short and high-level description of our proof 
techniques. We start with a topological theorem that is easy to state in \cref{sec:intro:kPolygonBorsukUlam}, which we call $k$-Polygon Tucker's Lemma. Then, we present
our results about the completeness of the well-studied BSS Theorem in
\cref{sec:intro:BSSTheorem}.  Finally, we briefly explain how we get our new
combinatorial proof of Necklace Splitting and the containment in
$\ppak{p}$, in \cref{sec:intro:Necklace}.

\smallskip

\noindent Throughout this paper, unless otherwise specified, $k$ denotes an integer larger or equal to $2$, and $p$ denotes a prime number.
  
\subsubsection{The $k$-Polygon Borsuk-Ulam Theorem} 
\label{sec:intro:kPolygonBorsukUlam}
 
The $k$-Polygon Borsuk-Ulam theorem can be understood  via the corresponding
statement of Borsuk-Ulam in $2$ dimensions. First, let us introduce the following notion of equivariance for a function. 
Let $S^1$ be the unit circle
in $2$-dimensions and $B^2$ be the unit disk. We say that a function $g : S^1 \to \bbR^2$
is \textit{equivariant to a rotation of $a^{\circ}$ degrees} if whenever we
rotate the input $\vec{x} $ by $a^{\circ}$, the image
$g(\vec{x})$ is also rotated by $a^{\circ}$ degrees. We extend this definition to
functions $f : B^2 \to \bbR^2$ and we say that $f$ is \textit{equivariant to a rotation of
$a^{\circ}$ degrees on the boundary} if the restriction of $f$ to the boundary $S^1$
is equivariant to a rotation of $a^{\circ}$ degrees. 
Using this language, the following
is one of the many equivalent ways to state the classical Borsuk-Ulam Theorem in $2$ dimensions (see the book by  \citet{Mat03BorsukUlam} for various equivalent versions).

\begin{inftheorem}[\textsc{$2$D Borsuk-Ulam Theorem}] \label{itm:2DBorsukUlamTheorem}
    Let $f : B^2 \to \bbR^2$ be a continuous function that is equivariant to a rotation of
  $180^{\circ}$ degrees on the boundary. Then, there exists 
  $\vec{x}^{\star} \in B^2$ such that $f(\vec{x}^{\star}) = 0$.
\end{inftheorem}

\noindent  The generalization, that we call $k$-Polygon Borsuk-Ulam Theorem,
comes as a clean extension of Borsuk-Ulam where instead of assuming equivariance to a 
rotation of $180^{\circ}$ degrees, we assume equivariance to a rotation of 
$(360/k)^{\circ}$ degrees.

\begin{inftheorem}[\textsc{$k$-Polygon Borsuk-Ulam Theorem}] \label{itm:kPolygonBorsukUlamTheorem}
    Let $f : B^2 \to \bbR^2$ be a continuous function that is equivariant to a rotation of
  $\frac{360^{\circ}}{k}$ degrees on the boundary. Then, there exists 
  $ \vec{x}^{\star} \in B^2$ such that $f(\vec{x}^{\star}) = 0$.
\end{inftheorem}

\noindent   Apart from its own mathematical interest, the $k$-Polygon Borsuk-Ulam Theorem is
essential for our results, since it serves as a stepping stone towards showing all our topological
$\ppak{p}$-completeness results. Also, it is the only one of our topological
problems which is defined for any $k \geq 2$, and thus the only problem which we are able to relate to the classes \ppak{k}, for general $k$.
 
  To prove the $k$-Polygon Borsuk-Ulam Theorem, we deviate from the topological 
techniques that have been used in the proof of similar extensions of the Borsuk-Ulam
Theorem \citep{BSS81} and we instead provide a combinatorial proof of a corresponding generalization of Tucker's lemma. We call this lemma $k$-Polygon Tucker's Lemma and an informal statement follows.

\begin{inftheorem}[\textsc{$k$-Polygon Tucker's Lemma}] \label{itm:polygonTucker}
    For $k\geq 3$, let $T$ be a triangulation of $B^2$ with $k$-fold rotationally symmetric boundary. Suppose that every vertex
  $\bx \in T$ has a color $\lambda(\bx) \in \mathbb{Z}_k$ such that $\lambda$ is
  equivariant to a rotation of $\frac{360^{\circ}}{k}$ degrees on the boundary. Then,
  in $T$ there exists (i) a trichromatic triangle, or (ii) an edge with distinct non-consecutive colors.
\end{inftheorem}

\noindent   The coloring function $\lambda$ is equivariant
to a rotation of $\alpha^{\circ}$ degrees on the boundary if whenever the 
argument $\vec{x}$ is rotated by $\alpha^{\circ}$, the color is increased by 
$1 \pmod{k}$. Arguably, for $k=3$ the above statement of $k$-Polygon Tucker's Lemma bears more resemblance to
 Sperner's Lemma than to the original version of Tucker's Lemma, because the solution is necessarily a 
trichromatic triangle. However, this is due to the fact that we are considering only the 
two-dimensional case here. The connection with the original version of Tucker's Lemma will become more
apparent when we present other high-dimensional modulo-$p$ generalizations of Tucker's Lemma.

  In the proof of $k$-Polygon Tucker's Lemma, the only inefficient step is the use
of a \textit{modulo-}$k$ counting argument. A simple way to visualize this 
argument is to imagine that if a space has cardinality 
that is non-zero modulo $k$ and if we can group points in this space that are non-solutions into groups of
size $k$, then in this space there should exist a solution. This kind of existential argument has
been formalized by Papadimitriou in his seminal paper 
\citep{Papadimitriou94-TFNP-subclasses} and there have been various
instantiations of this principle from which we can define corresponding computational
problems \citep{GoosKSZ2019, Hollender19}. In this paper, we mostly rely on 
the following instantiation defined by \citet{Hollender19}:

\begin{quote}
  \imba{k} : Given a directed graph and a vertex that is \emph{imbalanced-mod-$k$}, i.e., $(\mathsf{out}\text{-}\mathsf{degree}) - (\mathsf{in}\text{-}\mathsf{degree}) \neq 0 \pmod k$, find another such vertex.
\end{quote}

\noindent   Our main technical contribution in this section is to show that the 
computational problem associated with $k$-Polygon Tucker's Lemma is 
polynomial time equivalent to the computational version of \imba{k}. Towards this goal, we provide a generalization of the standard \emph{path-following} proof of Tucker's lemma by \citet{FT81}. Importantly, in our proof (which is a reduction to \imba{k}) the edges of the path are \emph{directed} by using a consistent direction of triangulations (inspired by the idea of \citet{Freund84a}). This technique was in fact incorrectly applied in the past to the case of $k=2$, leading to a false statement of $\ppad$-membership for Borsuk-Ulam. However, it turns out that the technique is very relevant in showing the equivalence between topological problems and modulo-$k$ arguments for $k>2$.
We illustrate
the appropriate way to use this technique via the \imba{k} problem and we believe 
that this will be useful for future reductions in $\ppak{k}$.

  The computational equivalence between $k$-Polygon Tucker's Lemma and \imba{k} 
together with the computational equivalence of 
$k$-Polygon Tucker's Lemma and the $k$-Polygon Borsuk-Ulam implies the following 
theorem, which is our main result in this section.

\begin{inftheorem} \label{itm:kPolygon:completeness}
    If $k$ is a prime power, the computational problem associated with the $k$-Polygon Borsuk-Ulam Theorem is \ppak{k}-complete. If $k$ is not a prime power, then it is complete for a subclass of
  \ppak{k}, denoted by $\ppak{k}[\#1]$.
\end{inftheorem}

\noindent The reason for this differentiation depending on whether $k$ is a prime power or not, is that we actually show equivalence of $k$-Polygon Tucker's Lemma with a \emph{special case} of \imba{k}. If $k$ is a prime power, then this special case is in fact \ppak{k}-complete, but in general it can be weaker.

  The formal definitions and the complete proofs, including the proof of \cref{itm:kPolygon:completeness}, about the $k$-Polygon Borsuk-Ulam Theorem
appear in \cref{sec:2dBSS}.

\subsubsection{Complexity of Finding Solutions to the BSS Theorem} \label{sec:intro:BSSTheorem}

  In this section, we present our results regarding the computational complexity of 
finding solutions  guaranteed to exist by the BSS Theorem, a famous
generalization of the celebrated Borsuk-Ulam Theorem. The BSS Theorem should be thought of as the 
corresponding $k$-Polygon Borsuk-Ulam in higher dimensions. We clarify though
that the BSS Theorem only works with $k = p$ where $p$ is a prime, and it applies only when
the number of dimensions is a multiple of $p - 1$. Hence, for $p \ge 5$ the
$p$-Polygon Borsuk-Ulam Theorem does not follow directly from the BSS Theorem. This is why we 
consider it a separate result and devote a separate section to it.

  When moving to more than two dimensions, we need to find an equivariance notion
corresponding to the equivariance of a rotation that we defined in the previous
section. A fundamental feature of a rotation in the plane is that it is a 
\emph{free action} on the boundary, i.e., there is no point on the boundary
$S^1$ that remains fixed if we apply the rotation. This free action property of 
rotations is crucial in the proof of $k$-Polygon Borsuk-Ulam Theorem and without
this property the theorem does not hold. Unfortunately, in higher
dimensions the rotations with respect to any axis no longer possess this property, as they always have a 
fixed point on the boundary $S^{m - 1}$ of the $m$-dimensional ball $B^m$. Thus, in higher dimensions  other operations acting 
 freely on $S^{m - 1}$ emerge.

  For the case of $k = 2$, there is a very simple generalization of the rotation by
$180^{\circ}$ that is a free action in any number of dimensions, namely, \textit{the
point reflection with respect to the origin}, i.e., $x \mapsto -x$. Hence, we have the following
informal statement of the Borsuk-Ulam Theorem for a general number of dimensions.

\begin{inftheorem}[\textsc{Borsuk-Ulam Theorem}] \label{itm:BorsukUlamTheorem}
    Let $f : B^m \to \bbR^m$ be a continuous function that on the boundary is 
  equivariant to point reflection with respect to the origin. Then, there exists
  $\vec{x}^{\star} \in B^m$ such that $f(\vec{x}^{\star}) = 0$.
\end{inftheorem}

\noindent  Observe that the \textit{order} of the point reflection operation is equal to 
$k = 2$, since if we apply the same operation twice, we return to the same
point. It is also trivial to see that the rotation by $360^{\circ}/k$ degrees has 
order $k$, since after $k$ times of applying this operation we return to the same 
point and if we apply this operation less than $k$ times, then we end up on a different 
point. These observations together with the requirement for a free action suggest
that in order to generalize the $k$-Polygon Borsuk-Ulam theorem to higher dimensions we need
a free action of order $k$ on the boundary $S^{m - 1}$ of the $m$-dimensional ball
$B^m$. Unfortunately, finding such operations is not as easy as in the case $k = 2$.
In particular, for $m = 2 \ell + 1$ the sphere $S^{2 \ell}$ has a free action of order $k$
only for $k = 2$ \citep{Hatcher02}. The starting point of the BSS Theorem is 
defining an operation $\alpha_p$ that has order $p$, where $p$ is a prime number,
and defines a free action on the sphere $S^{n (p - 1) - 1}$. These restrictions on $\alpha_p$ 
are the reason why, as we mentioned in the beginning of the section,
BSS only applies to dimensions that are multiples of $p - 1$ and for $k = p$ 
where $p$ is a prime number. Using the operation $\alpha_p$, we can informally state
the BSS Theorem as follows:

\begin{inftheorem}[\textsc{BSS Theorem} \citep*{BSS81}] \label{itm:BSSTheorem}
    Let $f : B^{n (p - 1)} \to \bbR^{n (p - 1)}$ be a continuous function that is
  equivariant with respect to $\alpha_p$ on the boundary. Then, there exists 
  $\vec{x}^{\star} \in B^{n (p - 1)}$ such that $f(\vec{x}^{\star}) = 0$.
\end{inftheorem}

\noindent  The original proof of the BSS Theorem \citep{BSS81} goes through the definition of
indices of functions in algebraic topology. One of our main contributions is to provide 
a combinatorial proof of the BSS Theorem. As in the case of 
$k$-Polygon Borsuk-Ulam, we provide a combinatorial proof of the BSS Theorem via
the corresponding version of Tucker's Lemma, which we call BSS Tucker's Lemma and we informally
state below.

\begin{inftheorem}[\textsc{BSS Tucker's Lemma}] \label{itm:BSSTucker}
    Let $p$ be a prime and $T$ be a triangulation of $B^{n (p - 1)}$ with an $\alpha_p$-symmetric boundary. Suppose that
  every vertex $\vec{x} \in T$ has a color 
  $\lambda(\vec{x}) \in \mathbb{Z}_p \times [n]$ such that $\lambda$ is 
  equivariant with respect to $\alpha_p$ on the boundary. Then, there exists
  a $(p - 1)$-simplex in $T$ that has all the colors $(1, j), \dots (p, j)$ for 
  some $j \in [n]$.
\end{inftheorem}

\noindent Our main technical contribution in this section is to show the following 
statements.
\vspace{-7pt}
\begin{itemize}
  \item[$\triangleright$] The computational problem that is associated with
  BSS Tucker's Lemma is polynomial time equivalent to the computational problem
  associated with the BSS Theorem (\cref{thm:BSStuckerBSSequiv}).
  \item[$\triangleright$] The computational problem associated with 
  $p$-Polygon Tucker's Lemma is reducible to the computational problem associated with
  BSS Tucker's Lemma (\cref{thm:BSS-tucker-ppak-hard}).
  \item[$\triangleright$] The computational problem associated with BSS Tucker's
  Lemma is reducible to the computational problem associated with 
  $\mathbb{Z}_p$-star Tucker's Lemma (\cref{lem:bss-to-pstar}).
\end{itemize}

\noindent $\mathbb{Z}_p$-star Tucker's Lemma is an existence theorem that we define 
informally in the next section and is the basic building block for proving the membership
of Necklace Splitting in $\ppak{p}$. As we will see in the next section, we 
prove that the computational problem associated with $\mathbb{Z}_p$-star Tucker's
Lemma is inside $\ppak{p}$. This result combined with \cref{itm:kPolygon:completeness},
which shows the PPA-$p$-completeness of the
$p$-Polygon Borsuk-Ulam Theorem, and the equivalence of the $p$-Polygon Borsuk-Ulam Theorem and $p$-Polygon Tucker's lemma, implies the main result of this section.

\begin{inftheorem} \label{itm:BSS:completeness}
    The computational problem associated with the BSS Theorem 
  is $\ppak{p}$-complete.
\end{inftheorem}

\noindent We defer the formal definition of the computational problem associated with the BSS Theorem and the complete proofs, including the proof of
 \cref{itm:BSS:completeness}, about the BSS Theorem appear in
\cref{sec:BSS-BSStucker}.
 
\subsubsection{Necklace Splitting with $p$ thieves is in \ppak{p}} \label{sec:intro:Necklace}

  Our main goal in the last part of the paper is to show that the computational
problems associated with the $p$-Necklace Splitting Theorem and the Consensus-$1/p$-Division Theorem are in $\ppak{p}$. Towards this goal we provide a full 
combinatorial proof for both of these problems that simplifies the only existing 
combinatorial proof by \citet{Meunier2014simplotopal}.

Our proof uses a different generalization of Tucker's Lemma which we call $\mathbb{Z}_p$-star Tucker's Lemma. Both
the statement of $\mathbb{Z}_p$-star Tucker Lemma and the proof are simple enough so 
that they do not invoke the involved definition of \emph{simplotopal} complexes of 
\citet{Meunier2014simplotopal}. These simplotopal complexes were used by Meunier in order to obtain a direct proof of the Necklace-Splitting theorem, without proving its continuous version. In contrast, we use standard simplicial complexes, since we are not only interested in the Necklace Splitting theorem, but also in its continuous generalization: the Consensus-$1/k$-Division Theorem (informally defined below).

  Recall that in the $k$-Necklace Splitting problem, $k$ thieves are aiming to split an open 
necklace containing $n$ beads of $t$ different colors, with exactly $k \cdot a_i$
beads of color $i$, for some $a_i \in \mathbb{N}$, such that each thief
receives exactly $a_i$ beads of color $i$. The $k$-Necklace Splitting Theorem states
the following.

\begin{inftheorem}[\textsc{$k$-Necklace Splitting Theorem} \citep{Alon87-Necklace}] \label{itm:pNecklaceSplitting}
    The $k$-Necklace Splitting problem always has a solution with $(k - 1)t$ cuts.
\end{inftheorem}

\noindent  The Consensus-$1/k$-Division problem resembles the continuous version of the
$k$-Necklace Splitting problem. In this problem, each one of $t$ agents has a 
probability measure $\mu_i$ over the unit interval $[0, 1]$. The goal is to cut the
interval $[0, 1]$ into pieces and assign one of $k$ possible colors to each piece
such that every agent measures the total mass of each different color the same.

\begin{inftheorem}[\textsc{Consensus-$1/k$-Division Theorem} \citep{Alon87-Necklace,SS03-Consensus}] \label{itm:pconhalving}
    The Consensus-$1/k$-Division problem always has a solution with $(k - 1)t$ cuts.
\end{inftheorem}

\noindent  Both the $k$-Necklace Splitting Theorem and the Consensus-$1/k$-Division Theorem 
have significant applications to combinatorics and social choice; see \cref{sec:previousWork}.

  As we have already mentioned, the proof that for prime $p$ the computational problems associated with
the $p$-Necklace Splitting Theorem and the Consensus-$1/p$-Division Theorem are 
inside $\ppak{p}$, is based on a different generalization of Tucker's Lemma for 
modulo-$p$ arguments, $\mathbb{Z}_p$-star Tucker's Lemma. The main difference of
$\mathbb{Z}_p$-star Tucker's Lemma with BSS Tucker's Lemma is the domain on which we 
define the triangulation. In the case of BSS Tucker's Lemma, the triangulation is
defined over a convex domain and hence is homeomorphic with the ball $B^m$. On the 
other hand, the triangulation for $\mathbb{Z}_p$-star Tucker's Lemma is defined over 
a \textit{star-convex} set which is not
homeomorphic to the ball $B^m$ anymore. This $d$-dimensional domain, which we denote by $R^d_p$, is a slightly modified version of the domain used by \citet{Meunier2014simplotopal} in his combinatorial proof of the $p$-Necklace Splitting Theorem. It admits a natural free action $\theta_p$ of order $p$. We can informally state $\mathbb{Z}_p$-star Tucker's
Lemma similarly to BSS Tucker's lemma as follows.

\begin{inftheorem}[\textsc{$\mathbb{Z}_p$-star Tucker's Lemma}] \label{itm:ZpStarTucker}
    Let $p$ be a prime and $T$ be a triangulation of $R^{t(p - 1)}_p$. Suppose 
  that every vertex $\vec{x} \in T$ has a color 
  $\lambda(\vec{x}) \in \mathbb{Z}_p \times [t]$ such that $\lambda$ is
  equivariant with respect to $\theta_p$ on the boundary. Then, there exists a
  $(p - 1)$-simplex of $T$ that has all the colors $(1, j), \dots (p, j)$ for some
  $j \in [t]$.
\end{inftheorem}

\noindent Our main technical contributions in this section are summarized in the following 
statements.
\vspace{-7pt}
\begin{itemize}
  \item[$\triangleright$] The computational problems that are associated with the 
  $p$-Necklace Splitting Theorem and the Consensus-$1/p$-Division Theorem are 
  polynomial time reducible to $\mathbb{Z}_p$-star Tucker's Lemma (\cref{thm:con1/p-to-ptucker}).
  \item[$\triangleright$] The computational problem associated with 
  $\mathbb{Z}_p$-star Tucker's Lemma is inside $\ppak{p}$ (\cref{thm:ptucker-in-ppa-p}). This is again proved by a reduction to \imba{p}, but this time we construct a \emph{weighted} directed graph.
  \item[$\triangleright$] The computational problem associated with BSS Tucker's Lemma
  is polynomial time reducible to the computational problem associated with 
  $\mathbb{Z}_p$-star Tucker's Lemma (\cref{lem:bss-to-pstar}).
\end{itemize}

  The above statements combined with the results in the previous sections imply our main result for
this part of the paper.

\begin{inftheorem} \label{itm:Necklace:main}
    The computational problem associated with $\mathbb{Z}_p$-star Tucker's
  Lemma is $\ppak{p}$-complete. The computational problems that are associated with
  the $p$-Necklace Splitting Theorem and the Consensus-$1/p$-Division Theorem are in
  $\ppak{p}$.
\end{inftheorem}

As a corollary of this result, we also obtain that for general $k \geq 2$, $k$-Necklace Splitting and Consensus-$1/k$-Division lie in \ppak{k} \emph{under Turing reductions}. In particular, if $k=p^r$ is a prime power, then the corresponding problems lie in \ppak{p}. Furthermore, our reductions also provide a new combinatorial proof of the Necklace Splitting theorem, that is conceptually simpler and does not use any involved machinery.

  The formal definitions and the complete proofs, including the proof of \cref{itm:Necklace:main}, about $\mathbb{Z}_p$-star Tucker's
  Lemma, the $p$-Necklace Splitting Theorem and the Consensus-$1/p$-Division Theorem appear in
\cref{sec:pstartucker} and \cref{sec:necklace}.

\smallskip

\noindent Our results are summarized in \cref{tab:landscape}, where we also highlight where they fit in the computational landscape of the classes of interest. Relevant further related work is discussed in the next section.

\newcommand{\citeps}[1]{\scriptsize{\citep{#1}}}
\begin{table}
\centering
\resizebox{16cm}{!}{\begin{tabular}{l"l||l|l}
& \ppad & \ppa & \ppak{p} $(p \geq 3)$ \\
\thickhline
\begin{tabular}[c]{@{}l@{}} \small{Topological}\\ \small{Existence}\\ \small{Theorem} \end{tabular}
& \begin{tabular}[c]{@{}l@{}} \textsc{Brouwer}\\ \citeps{Papadimitriou94-TFNP-subclasses}\\ \citeps{CD09-2DBrouwer}\\ \textsc{Hairy-Ball}\\ \citeps{GoldbergH19-HairyBall} \end{tabular}
& \begin{tabular}[c]{@{}l@{}} \textsc{Borsuk-Ulam}\\ \citeps{Papadimitriou94-TFNP-subclasses}\\ \citeps{ABB15-2DTucker} \end{tabular}
& \begin{tabular}[c]{@{}l@{}} \polygonBorsukUlam{p}\\ \textsc{$p$-BSS} \\ \scriptsize{[{\color{black!50!blue} This Work}]} \end{tabular} \\
\hline
\begin{tabular}[c]{@{}l@{}} \small{Combinatorial}\\ \small{Lemma} \end{tabular}
& \begin{tabular}[c]{@{}l@{}} \textsc{Sperner}\\
\citeps{Papadimitriou94-TFNP-subclasses} \\ \citeps{CD09-2DBrouwer} \end{tabular}
& \begin{tabular}[c]{@{}l@{}} \textsc{Tucker}\\ \citeps{Papadimitriou94-TFNP-subclasses}\\ \citeps{ABB15-2DTucker} \end{tabular}
& \begin{tabular}[c]{@{}l@{}} \polygonTucker{p}\\ \textsc{$p$-BSS-Tucker}\\
\pstartucker{p}\\ \scriptsize{[{\color{black!50!blue} This Work}]} \end{tabular}\\
\hline
\begin{tabular}[c]{@{}l@{}} \small{Notable}\\ \small{Problems} \end{tabular}
& \begin{tabular}[c]{@{}l@{}} \textsc{Nash}\\ \citeps{Daskalakis2009}\\ \citeps{chen2009settling}\\ 
\textsc{Market-Equilibrium} \\ \citeps{chen2009bsettling,chen2013complexity}\\
\ \ \ \  \small{\emph{and many more...}} \end{tabular}
& \begin{tabular}[c]{@{}l@{}} \textsc{$2$-Necklace-Splitting}\\
\textsc{Consensus-Halving}\\ \textsc{Discrete-Ham-Sandwich}\\ \citeps{FRG18-Necklace} \end{tabular}
& \begin{tabular}[c]{@{}l@{}} \textsc{Symmetric-Chevalley-mod-$p$}\\ \citeps{GoosKSZ2019}\\
\textsc{$p$-Necklace-Splitting}\\ \textsc{Consensus-$1/p$-Division}\\
\scriptsize{[{\color{black!50!blue} Membership: This Work}]}\\
\scriptsize{[{\color{black!50!red} Hardness: Open}]}
\end{tabular}\\
\end{tabular}}
\caption{An overview of the computational landscape for the related TFNP classes.}\label{tab:landscape}
\end{table}

\subsection{Discussion and Further Related Work}
\label{sec:previousWork}

\textbf{Computational Classes:} As mentioned earlier, among the classes of TFNP, PPAD has been the most successful in capturing the complexity of well-known problems. Besides the complexity of computing a Nash equilibrium \citep{Daskalakis2009,chen2009settling,rubinstein2018inapproximability}, other \ppad-complete problems are computing equilibria in markets \citep{chen2009bsettling,chen2013complexity}, versions of envy-free cake cutting 
\citep{deng2012algorithmic} and fixed-point theorems \citep{mehta2014constant,GoldbergH19-HairyBall}. 

For \ppa, the recent results by \citet{FRG18-Consensus,FRG18-Necklace} have solidified the status of the class as one that contains natural problems. In particular, they showed that $2$-thief Necklace Splitting is \ppa-complete; the proof goes via its continuous version, the \emph{Consensus-Halving} problem of \citet{SS03-Consensus}\footnote{The hardness result of \cite{FRG18-Consensus,FRG18-Necklace} for the Consensus-Halving problem was strengthened recently by \citet{FRHSZ2020consensus-easier} to the case of simpler measures.} Our \ppak{p}-membership result for the problem with $p$ thieves also uses the continuous variant, termed as \emph{Consensus-$1/p$-Division} in \citep{SS03-Consensus}; we note that \citet{Alon87-Necklace} used the same problem in his existence proof, referring to it as a \emph{generalized Hobby-Rice Theorem}. As we explained earlier, \citet{FRG18-Necklace} conjectured that the problem with $p$ thieves is complete for \ppak{p}.

The classes \ppak{p} were introduced by \citet{Papadimitriou94-TFNP-subclasses} for any prime $p$, in the context of classifying a computational version of the Chevalley-Warning theorem \citep{Chevalley1935,Chevalley1935b}. He proved that the corresponding problem \textsc{Chevalley-mod-$p$} lies in \ppak{p}. Recently, \citet{GoosKSZ2019} showed that an explicit version of the problem is complete for \ppak{p}, therefore obtaining the first \ppak{p}-completeness result for a natural problem. The authors of \citep{GoosKSZ2019}, as well as \citet{Hollender19}, independently also extended the definition of the classes \ppak{k} to any $k \geq 2$, and provided several characterizations in terms of their defined-for-primes counterparts. \citet{Hollender19} also investigated the connection with the classes PMOD-$k$, which bear strong resemblance to \ppak{k}, and were defined seemingly independently of Papadimitriou's work by \citet{johnson2011reductions}. From the work of \citet{Hollender19}, we primarily make use of the \ppak{k}-complete computational problem \imba{k}, a generalization of the \ppad-complete problem \textsc{Imbalance} \citep{Beame1998,GoldbergH19-HairyBall}.\\

\noindent \textbf{Necklace Splitting:} The origins of the Necklace Splitting problem (\cref{itm:pNecklaceSplitting}), and its continuous variant (\cref{itm:pconhalving}), can be traced back to work of \citet{neyman1946theoreme} and \citet{hobby1965moment}. The problem was firstly phrased as a necklace splitting problem by \citet{BL82}. Later on, \citet{Goldberg1985} and \citet{Alon1986} proved the first existence results for two thieves, and \citet{Alon87-Necklace} extended the result to the case of any number $k \geq 2$ of thieves. For two thieves, the continuous version was studied independently by \citet{SS03-Consensus}, who termed the problem as ``Consensus-Halving'' and came up with a combinatorial, constructive proof of existence; this proof was later modified into a \ppa\ membership proof by \citet{filos2018hardness}. The importance of obtaining combinatorial proofs of existence was highlighted in a series of papers \citep{DeLongueville2006,Meunier2008,Meunier2014simplotopal,Meunier2012}; note that the proof of \citet{Alon87-Necklace} uses two results of \citet{BSS81} and is therefore not combinatorial. The first fully combinatorial proof for Necklace Splitting (according to the definition of \citet{Ziegler2002}) is due to \citet{Meunier2014simplotopal}. However, the proof comes up with a rather involved construction based on the notion of simplotopal complexes, and thus is notably hard to follow without an advanced understanding of the theory of simplicial complexes. Since our proof is based on a reduction, it is also inherently combinatorial. Very recently, \citet{alon2020efficient} studied approximate versions of both the continuous and the discrete variants, and provided polynomial-time algorithms for either large enough approximations or a larger number of cuts (i.e., cases not captured by the hardness results of \citep{FRG18-Consensus,FRG18-Necklace,FRHSZ2020consensus-easier}).\\

\noindent \textbf{The BSS Theorem:} The BSS Theorem (\cref{itm:BSSTheorem}, \cref{thm:BSStheorem}), due to \citep{BSS81}, is perhaps the most well-known generalization of the Borsuk-Ulam theorem. Besides \citep{Alon87-Necklace}, it has been used to prove existence of other interesting problems, including a generalization of the Kneser-Lov\'{a}sz Theorem \citep{kneser1955aufgabe,lovasz1978kneser} due to \citet{alon1986chromatic}, a generalized van Kampen-Flores Theorem \citep{sarkaria1991generalized} and the generalization to Tverberg's Theorem, proven in \citep{BSS81}. We believe that our \ppa-$p$-completeness result paves the way for studying the complexity of those problems as well.

\section{Preliminaries}

\subsection{Total NP Search Problems}

Let $\{0,1\}^*$ denote the set of all finite length bit-strings. For $x \in \{0,1\}^*$, let $|x|$ be its length. A computational search problem is given by a binary relation $R \subseteq \{0,1\}^* \times \{0,1\}^*$. The associated problem is: given an instance $x \in \{0,1\}^*$, find a $y \in \{0,1\}^*$ such that $(x,y) \in R$, or return that no such $y$ exists. The search problem $R$ is in FNP (\emph{Functions in NP}), if $R$ is polynomial-time computable (i.e., $(x,y) \in R$ can be decided in polynomial time in $|x|+|y|$) and there exists some polynomial $p$ such that $(x,y) \in R \implies |y| \leq p(|x|)$.  Thus, FNP is the search problem version of NP.

The class TFNP (\emph{Total Functions in NP} \citep{Megiddo1991}) contains all FNP search problems $R$ that are \emph{total}: for every $x \in\{0,1\}^*$ there exists $y \in\{0,1\}^*$ such that $(x,y) \in R$. Note that the totality of problems in TFNP does not rely on any ``promise''. Instead, there is a \emph{syntactic} guarantee of totality: for any instance in $\{0,1\}^*$, there is always at least one solution.

Let $R$ and $S$ be total search problems in TFNP. We say that $R$ (many-one) reduces to $S$, denoted $R \leq S$, if there exist polynomial-time computable functions $f,g$ such that
$$(f(x),y) \in S \implies (x,g(x,y)) \in R.$$
Note that if $S$ is polynomial-time solvable, then so is $R$. We say that two problems $R$ and $S$ are (polynomial-time) equivalent, if $R \leq S$ and $S \leq R$.

Sometimes a more general notion of reduction is used. A Turing reduction from $R$ to $S$ is a polynomial-time oracle Turing machine that solves problem $R$ with the help of queries to an oracle for $S$. Note that a Turing reduction that only makes a single oracle query immediately yields a many-one reduction.

\subsection{The Classes \ppak{k}}\label{sec:def-ppa-k}

For $k \geq 2$, \ppak{k} is a subclass of TFNP that aims to capture the complexity of TFNP problems whose totality is proved by using an argument modulo $k$. The classes \ppak{p} (for prime $p$) were introduced by \citet{Papadimitriou94-TFNP-subclasses}. The case $p=2$ corresponds to \emph{parity arguments}, and in that case the class \ppak{2} is simply called \ppa. Recently, the definition of \ppak{k} was extended to any $k \geq 2$ by \citet{GoosKSZ2019,Hollender19}. The class \ppak{k} is defined as the class of TFNP problems that reduce to the following problem.

\begin{definition}[\bipa{k} \citep{Papadimitriou94-TFNP-subclasses}]\label{def:bipartite-k}
Let $k \geq 2$. The problem \bipa{k} is defined as: given a Boolean circuit $C: \{0,1\} \times \{0,1\}^n \to (\{0,1\} \times \{0,1\}^n)^k$ that computes a bipartite graph on the vertex set $(\{0\} \times \{0,1\}^n, \{1\} \times \{0,1\}^n)$ with $|C(00^n)| \in \{1,\dots, k-1\}$, find
\begin{itemize}
    \item $x \neq 00^n$ such that $|C(x)| \notin \{0,k\}$
    \item or $x,y$ such that $y \in C(x)$ but $x \notin C(y)$.
\end{itemize}
\end{definition}

\noindent The circuit $C$ computes a bipartite graph as follows. The output of circuit $C$ on some input $x$ is a list of $k$ bit-strings, each of length $n+1$. If $x \in \{0\} \times \{0,1\}^n$, then we let $C(x)$ denote the set of distinct bit-strings that appear in that list and that lie in $\{1\} \times \{0,1\}^n$. Similarly, if $x \in \{1\} \times \{0,1\}^n$, then $C(x)$ denotes the set of distinct bit-strings that appear in that list and that lie in $\{0\} \times \{0,1\}^n$. For any $x,y \in \{0,1\} \times \{0,1\}^n$, there exists an edge between $x$ and $y$ if and only if $y \in C(x)$ and $x \in C(y)$. It is easy to see that this indeed defines a bipartite graph.

\begin{definition}
Let $k \geq 2$ and $1 \leq \ell \leq k-1$. The problem \bipal{k}{\ell} is defined as \bipa{k} (\cref{def:bipartite-k}) but with the additional condition $|C(00^n)| = \ell$.
\end{definition}

\noindent Note that this condition can be enforced syntactically and so this problem also lies in TFNP (see \citep{Papadimitriou94-TFNP-subclasses} for a definition of ``syntactic''). 

\begin{definition}[\ppakl{k}{\ell}]
Let $k \geq 2$ and $1 \leq \ell \leq k-1$. The class \ppakl{k}{\ell} is defined as the set of all \textup{TFNP} problems that many-one reduce to \bipal{k}{\ell}.
\end{definition}

\noindent The following result relates these special subclasses to the main ones.

\begin{proposition}[\citep{Hollender19}]
$\ppak{k}[\#1] = \cap_{p \in PF(k)} \ppak{p}$, where $PF(k)$ denotes the set of prime factors of $k$.
\end{proposition}

\noindent In this paper, we will also use the following definition of \ppak{k}, which was shown to be equivalent to the standard one \citep{Hollender19}. The class \ppak{k} is the set of all TFNP problems that reduce to the following problem.

\begin{definition}
Let $k \geq 2$. The problem \imba{k} is defined as: given Boolean circuits $S,P : \{0,1\}^n \to (\{0,1\}^n)^k$ with $|S(0^n)|-|P(0^n)| \neq 0 \mod k$, find
\begin{itemize}
    \item $x \neq 0^n$ such that $|S(x)|-|P(x)| \neq 0 \mod k$
    \item or $x,y$ such that $y \in S(x)$ but $x \notin P(y)$, or $y \in P(x)$ but $x \notin S(y)$.
\end{itemize}
For $1 \leq \ell \leq k-1$, \imbal{k}{\ell} is defined with the additional condition $|S(0)|-|P(0)| = \ell$.
\end{definition}

\noindent We obtain a seemingly more general version of this problem by allowing the edges to have integer weights. In that case the imbalance of a vertex is measured as the difference of the weights of all incoming edges and the weights of all outgoing edges. It is easy to see that this problem is in fact equivalent to \imba{k}. First, all the weights can be assumed to be in $\{0,1, \dots, k-1\}$, since we can reduce them modulo $k$. Next, we can split an edge with weight $\ell$ into $\ell$ copies of the edge. Finally, to ensure that we don't have multi-edges, we add a new vertex in the middle of every edge.
Note that the new vertex will be balanced by construction, and will thus not introduce any new solutions.

  It is easy to see that in all the aforementioned problems a solution is always
guaranteed to exist. The search problems are not trivial though, because the graph can
have exponential size with respect to its description. Indeed, the graph is given by 
Boolean circuits that compute the successors and predecessors of every vertex. 

We make use of the following known properties of these classes:

\begin{proposition}[\citet{GoosKSZ2019,Hollender19}]\label{prop:ppak-properties}
It holds that:
\begin{itemize}
    \item for any prime $p$ and any $r \geq 1$, $\ppak{p^r} = \ppak{p}$
    \item for any $k, \ell \geq 1$, $\ppak{k} \subseteq \ppak{k\ell}$
\end{itemize}
\end{proposition}

\paragraph{Topological Definitions and Details:} Our results in the next sections will require several definitions from topology, as well as the corresponding notation. The more experienced reader might already be familiar with some of these concepts, but all the relevant details are included in \cref{ap:definitions}. There, we also define our \emph{value} and \emph{index} functions, which allow us to enumerate over simplices, and access simplices that contain a point $\mathbf{x}$ in the domain respectively.

\section{$k$-Polygon Borsuk-Ulam: a $\ppak{k}[\#1]$-complete Problem in 2D-space}\label{sec:2dBSS}

  In this section we present a generalization of the Borsuk-Ulam Theorem in two 
dimensional space. Surprisingly this theorem is not captured by the BSS Theorem and
hence has its own topological interest. Our proof of this theorem is combinatorial
and is based on a generalization of Tucker's Lemma which we call 
\textit{Polygon Tucker's Lemma}, hence it is very different from the proof of the BSS
Theorem. Our main result is that $k$-Polygon Borsuk-Ulam is complete for $\ppak{k}[\#1]$.
This gives the first topological characterization of the classes $\ppak{k}[\#1]$ but also
reveals in a very intuitive way the relation between the different $\ppak{k}[\#1]$ classes
and $\ppad$. Recall that when $k = p^r$ is a prime power, $\ppak{k}[\#1] = \ppak{k} = \ppak{p}$.
\smallskip

  We start this section with a unified description of Brouwer's Fixed Point Theorem, the Borsuk-Ulam Theorem and our generalization: the $k$-Polygon Borsuk-Ulam Theorem. For 
this we will need the following definition.

\begin{definition}[\textsc{Rotation Operator}] \label{def:rotationOperator}
    We define the $k$-th \textit{rotation operator} 
  $\theta_{k} : \bbR^2 \to \bbR^2$ as follows: 
  $\theta_{k}(\vec{x}) = R_k \vec{x}$, where $R_k$ is the following two 
  dimensional rotation matrix 
  \[R_k = \begin{bmatrix}
            \cos\left(-\frac{2 \pi}{k}\right) & -\sin\left(-\frac{2 \pi}{k}\right) \\
            \sin\left(-\frac{2 \pi}{k}\right) & \cos\left(-\frac{2 \pi}{k}\right)
         \end{bmatrix}. \]
In other words, $\theta_k$ corresponds to a clockwise rotation by an angle of $2 \pi/k$.
\end{definition}

  We continue with a statement of Brouwer's Fixed Point Theorem that is 
different from the standard statement, but is well known to be equivalent to that (e.g., see \citep{Mat03BorsukUlam}).

\begin{theorem}[\textsc{Brouwer's Fixed Point Theorem}] \label{thm:BrouwersTheorem}
    Let $f : B^2 \to \bbR^2$ be a continuous function such that $f(\vec{x}) = \vec{x}$
  for all $\vec{x} \in \partial B^2$. Then there exists $\vec{x}^{\star} \in B^2$
  such that $f(\vec{x}^{\star}) = 0$.
\end{theorem}

  Next, using the same language, we give a statement of the Borsuk-Ulam Theorem.

\begin{theorem}[\textsc{Borsuk-Ulam Theorem}] \label{thm:BorsukUlamTheorem}
    Let $f : B^2 \to \bbR^2$ be a continuous function such that 
  $f(\theta_2(\vec{x})) = \theta_2(f(\vec{x}))$ for all 
  $\vec{x} \in \partial B^2$. Then there exists $\vec{x}^{\star} \in B^2$ such that
  $f(\vec{x}^{\star}) = 0$.
\end{theorem}

  It is clear from the above expression that the Borsuk-Ulam Theorem is a
generalization of Brouwer's Fixed Point Theorem. This observation is in line with 
the fact that $\ppad \subseteq \ppa$, since Brouwer is 
complete for $\ppad$ and Borsuk-Ulam is complete for $\ppa$. We now present our 
extension, that we call $k$-Polygon Borsuk-Ulam Theorem.

\begin{theorem}[\textsc{$k$-Polygon Borsuk-Ulam Theorem}] \label{thm:kPolygonBorsukUlamTheorem}
    Let $f : B^2 \to \bbR^2$ be a continuous function such that 
  $f(\theta_k(\vec{x})) = \theta_k(f(\vec{x}))$ for all 
  $\vec{x} \in \partial B^2$. Then there exists $\vec{x}^{\star} \in B^2$ such that
  $f(\vec{x}^{\star}) = 0$.
\end{theorem}

  As we will see the $k$-Polygon Borsuk-Ulam Theorem is also a generalization of 
Brouwer's Fixed Point Theorem and it is complete for $\ppak{k}[\#1]$ which is also in
line with the fact that $\ppad \subseteq \ppak{k}[\#1]$. Another interesting fact about the
$k$-Polygon Borsuk-Ulam Theorem is that it does not directly follow from the traditional 
generalization of the Borsuk-Ulam Theorem, namely the BSS Theorem, as we will see
in the next section.

\subsection{$k$-Polygon Tucker's Lemma and $k$-Polygon Borsuk-Ulam in $\ppak{k}[\#1]$}
\label{sec:kPolygonTucker}

  In this section, we define $k$-Polygon Tucker's Lemma and we prove
that it is equivalent to the $k$-Polygon Borsuk-Ulam Theorem. Additionally, we provide a 
combinatorial proof of $k$-Polygon Tucker's Lemma using a modulo-$k$
argument. This combinatorial proof puts both $k$-Polygon Tucker and $k$-Polygon Borsuk-Ulam
in the class $\ppak{k}[\#1]$. 

Before showing the equivalence of the two statements, we provide some necessary notation and 
definitions.
\begin{definition}[\textsc{$k$-Polygon \& Nice Triangulation}]
\label{def:kregularPolygonNiceTriangulation}
    For $k \geq 3$, let $W_k$ be the regular $k$-polygon, i.e., regular $k$-gon, centered at $\vec{0} \in \mathbb{R}^2$
  with radius $1$. Let
  $\bu_1, \dots, \bu_k$ denote the vertices of $W_k$, ordered such that $\theta_k(\bu_i) = \bu_{i+1 \pmod{k}}$ for all $i \in [k]$. We define $T^{\star}$ to be
  the triangulation of $W_k$ that includes the simplices 
  $\sigma_i = \chull(\{\vec{0}, \vec{u}_i, \vec{u}_{i + 1 \pmod{k}}\})$ for 
  $i \in [k]$. We call a triangulation $T$ \textit{nice} if it satisfies the following two 
  properties:
  \begin{itemize}
      \item it is a refinement of $T^{\star}$, and
      \item it is symmetric with respect to $\theta_k$ on the boundary. This
  means that for every edge $\psi \in T$ such that $\psi \subseteq \partial W_k$ it 
  holds that $\theta_k(\psi) \in T$.
  \end{itemize}
\end{definition}

\begin{theorem}[\textsc{$k$-Polygon Tucker's Lemma}] \label{thm:polygonTucker}
    For $k\geq 3$, let $W_k$ be a $k$-regular polygon. Fix some nice triangulation 
  $T$ of $W_k$. Suppose that every vertex $\bx \in T$ has a label 
  $\lambda(\bx) \in \mathbb{Z}_k$ such that for any $\by \in \partial T$ we have
  $\lambda(\theta_k(\by)) = \lambda(\by) + 1 \pmod{k}$. Then at least one of the
  following exists: (1) a simplex $\sigma \in T$ with vertices $\vec{v}_1$,
  $\vec{v}_2$, $\vec{v}_3$ such that all the labels $\lambda(\vec{v}_1)$,
  $\lambda(\vec{v}_2)$, $\lambda(\vec{v}_3)$ are different, or (2) an edge 
  $\psi \in T$ with vertices $\vec{v}_1$, $\vec{v}_2$ such that 
  $\lambda(\vec{v}_1) - \lambda(\vec{v}_2) \pmod{k} \notin \{0, 1, -1\}$.
\end{theorem}

\begin{remark}
    The above theorem can be proved if we invoke Dold's 
  Theorem from algebraic topology \cite{Dold1983}. However this proof is 
  not a constructive proof and hence it cannot be used for our purposes and
  for this reason we reprove this theorem using a constructive combinatorial
  proof. Of course, this is only a special case of the general Dold's Theorem
  and it is an interesting open problem what is the relation of Dold's 
  Theorem with the subclasses of TFNP as we state in the Conclusions 
  section.
\end{remark}

\subsubsection{Equivalence of $k$-Polygon Tucker and $k$-Polygon Borsuk-Ulam} \label{sec:kPolygonTuckerkPolygonBorsukUlam:equivalence}

  We start by showing that $k$-Polygon Tucker's Lemma is implied by $k$-Polygon
Borsuk-Ulam Theorem and then we also show the converse, in 
\cref{lem:polygonBorsukUlamImpliesPolygonTucker} and 
\cref{lem:polygonTuckerImpliesPolygonBorsukUlam}.

\begin{lemma} \label{lem:polygonBorsukUlamImpliesPolygonTucker}
    $k$-Polygon Borsuk-Ulam (\cref{thm:kPolygonBorsukUlamTheorem}) implies
 $k$-Polygon Tucker's Lemma (\cref{thm:polygonTucker}).
\end{lemma}

\begin{proof}
    We interpret each label $i \in \mathbb{Z}_k$ as the vector $\bu_i$, which is 
  the $i$-th vertex of the polygon $W_k$. Let $h: W_k \to W_k$  be the piecewise linear extension of 
  the function that has value $\bu_{\lambda(\vec{x})}$ on any vertex $\vec{x} \in T$. Finally 
  we define $g : B^2 \to \bbR^2$ as the composition of $h$ and a homeomorphism between $W_k$ and
  $B^2$ that maps $\vec{0}$ to $\vec{0}$ and is $\theta_k$-equivariant. Notice that by the definition of $h$ and $g$
  and the equivariance assumption on $\lambda$, it holds that 
  $g(\theta_k(\bx)) = \theta_k(g(\bx))$ for $\bx \in \partial B^2$. Then, it
  follows from \cref{thm:kPolygonBorsukUlamTheorem} that there exists an 
  $\by^{\star} \in B^2$ such that $g(\by^{\star}) = 0$ and hence there exists a
  $\bx^{\star} \in W_k$ such that $h(\bx^{\star}) = 0$.
     
    Now, let $\sigma^{\star}$ be a full-dimensional simplex of $T$ that
  contains $\bx^{\star}$. If the vertices of $\sigma^{\star}$ have only two consecutive labels,
  say $1$ and $2$ (corresponding to vectors $\bu_1$ and $\bu_2$), then it is impossible to 
  have $h(\bx^{\star}) = 0$ since it is a non-zero linear
  interpolation of $\bu_1$ and $\bu_2$ and these vectors are linearly independent.
  Hence, it has to be that the vertices of $\sigma^{\star}$ have either two non-consecutive labels or three different labels
  and $k$-Polygon Tucker's Lemma follows.
\end{proof}

\begin{lemma} \label{lem:polygonTuckerImpliesPolygonBorsukUlam}
    $k$-Polygon Tucker's Lemma (\cref{thm:polygonTucker}) implies
  $k$-Polygon Borsuk-Ulam (\cref{thm:kPolygonBorsukUlamTheorem}).
\end{lemma}

\begin{proof}
    Let $h : W_k \to \bbR^2$ be the function obtained from $f$ by using a  
  homeomorphism between $W_k$ and $B^2$ that fixes $\vec{0}$ and is $\theta_k$-equivariant. Using standard arguments,
  that are used in the proof of both Brouwer's Fixed Point Theorem via Sperner's
  Lemma and the proof of Borsuk-Ulam via Tucker's Lemma, it is enough if for every
  $\varepsilon > 0$ we find a point $\bx$ such that $\norm{h(\bx)} \leq \varepsilon$, for
  details we refer to \citep{Mat03BorsukUlam}. Since $h$ is a continuous function in a
  compact set, it is also uniformly continuous. Thus for every $\varepsilon > 0$ there
  exists a $\delta > 0$ such that for any $\bx, \bx' \in W_k$, if 
  $\norm{\bx - \bx'}_2 < \delta$, then 
  $\norm{h(\bx) - h(\bx')} < \varepsilon / k$. Assume that $T$ is a nice 
  triangulation of $W_k$ such that any simplex $\sigma \in T$ has diameter 
  at most $\delta$.

    We define the labeling $\lambda(\bx) = i$ to be equal to the index of the vertex
  $\vec{u}_i$ of $W_k$ that is closest to $h(\bx)$. We break ties between $i$ and $i+1$, by picking $i$. Note that the only other kind of tie that can occur is if $h(\bx)=0$. In that case all labels are tied. If $\bx=0$, we pick one arbitrarily. Otherwise, i.e., if $\bx \neq 0$, we apply the same rule as for $h(\bx)$ above, but for $\bx$ instead. It is easy to check that this tie-breaking is $\mathbb{Z}_k$-equivariant. For any $\bx \in \partial W_k$ because of the equivariance
  assumption on $g$ we have that $\lambda(\theta_k(\bx)) = \lambda(\bx) + 1\pmod{k}$
  and hence the assumptions of $k$-Polygon Tucker's Lemma are satisfied. Now we 
  distinguish two cases.
  \smallskip
  
  \paragraph{$\boldsymbol{k = 3}$.} In this case, $3$-Polygon Tucker's Lemma implies that there
  exists a simplex $\sigma \in T$ whose vertices $\vec{v}_1$, $\vec{v}_2$, 
  $\vec{v}_3$ contain all the labels $1$, $2$, and $3$ respectively. This implies the 
  following set of inequalities by the definition of the labeling $\lambda$
  \[ \norm{h(\vec{v}_1) - \vec{u}_1} \le \min(\norm{h(\vec{v}_1) - \vec{u}_2}, \norm{h(\vec{v}_1) - \vec{u}_3}) \]
  \[ \norm{h(\vec{v}_2) - \vec{u}_2} \le \min(\norm{h(\vec{v}_2) - \vec{u}_1}, \norm{h(\vec{v}_2) - \vec{u}_3}) \]
  \[ \norm{h(\vec{v}_3) - \vec{u}_3} \le \min(\norm{h(\vec{v}_3) - \vec{u}_1}, \norm{h(\vec{v}_3) - \vec{u}_2}). \]
  \noindent Hence, it is easy to see that the maximum angle between any of 
  $h(\vec{v}_1)$, $h(\vec{v}_2)$ and $h(\vec{v}_3)$ is at least $2 \pi/3$. Now, assume
  for the sake of contradiction that for all $\vec{v}_1$, $\vec{v}_2$, $\vec{v}_3$ it holds
  that
  \[ \norm{h(\vec{v}_1)} \ge \varepsilon, ~ \norm{h(\vec{v}_2)} \ge \varepsilon, ~ \norm{h(\vec{v}_3)} \ge \varepsilon. \]
  \noindent This together with the fact that the maximum angle is at least $2 \pi/3$ 
  implies that the maximum distance is at least $2 \sin(2 \pi/6) \varepsilon$. But this
  implies that 
  \[ \max(\norm{h(\vec{v}_1) - h(\vec{v}_2)}, \norm{h(\vec{v}_2) - h(\vec{v}_3)}, \norm{h(\vec{v}_1) - h(\vec{v}_3)}) \ge \sqrt{3} \cdot \varepsilon > \frac{\varepsilon}{3} \]
  which contradicts the definition of the triangulation $T$, where the 
  vertices of the same simplex are at most $\delta$ far from each other and hence 
  their images are at most $\varepsilon/k= \varepsilon/3$ far from each other. This implies that our
  assumption was wrong and hence for at least one $i \in [3]$ it holds that
  $\norm{h(\vec{v}_i)} \le \varepsilon$ and the result for this case follows.
  \smallskip
  
  \paragraph{$\boldsymbol{k > 3}$.} In this case, $k$-Polygon Tucker's Lemma implies that there
  exists an edge $\psi \in T$ whose vertices $\vec{v}_1$ and $\vec{v}_2$
  have labels that differ by more that one, without loss of generality assume that these
  labels are $1$ and $3$ respectively. By the definition of the labeling $\lambda$ this
  implies that
  \[ \norm{h(\vec{v}_1) - \vec{u}_1} \le \min_{i}(\norm{h(\vec{v}_1) - \vec{u}_i}), ~~~ \norm{h(\vec{v}_2) - \vec{u}_3} \le \min_{i}(\norm{h(\vec{v}_2) - \vec{u}_i}) \]
  \noindent Hence, it is easy to see that the angle between $h(\vec{v}_1)$ and $h(\vec{v}_2)$ is at
  least $2 \pi/k$. Now, for sake of contradiction we assume that
  \[ \norm{h(\vec{v}_1)} \ge \varepsilon, ~ \norm{h(\vec{v}_2)} \ge \varepsilon. \]
  \noindent This together with the fact that the angle is at least $2 \pi/k$ implies
  that the distance $\norm{h(\vec{v}_1) - h(\vec{v}_2)}$ is at least 
  $2 \sin(2 \pi / k) \varepsilon > \varepsilon / k$ which contradicts the definition of 
  $T$.
  \smallskip
  
  \noindent Therefore, in both cases for every $\varepsilon > 0$ we can find a $\vec{v} \in W_k$ such that 
  $\norm{h(\vec{v})} \le \varepsilon$. This, by standard arguments and the compactness of $W_k$,
  implies that there exists a point $\vec{x}^{\star} \in W_k$ such that 
  $h(\vec{x}^{\star}) = \vec{0}$ and hence the result follows.
\end{proof}

\subsubsection{Proof of $k$-Polygon Tucker's Lemma} \label{sec:kPolygonTucker:proof}

In this section, we prove $k$-Polygon Tucker's Lemma. A corollary of our proof
combined with \cref{lem:polygonTuckerImpliesPolygonBorsukUlam} is that the
computational problems associated with $k$-Polygon Tucker's Lemma and $k$-Polygon 
Borsuk-Ulam both belong to $\ppak{k}[\#1]$. The proof technique that we introduce here is a generalization of the combinatorial proof of Tucker's lemma given by \citet{FT81}.
\smallskip

  We use a modulo-$k$ argument to prove this theorem. We define a directed 
graph where the vertices correspond to the simplices of $T$, and we also identify the symmetric edges of
$T$ on the boundary of $W_k$ as the same vertex. Then, in this graph we describe a 
rule for defining edges such that there exist three types of vertices:
\begin{enumerate}
  \item the vertex that corresponds to the $0$-dimensional simplex $\{\vec{0}\}$ which 
  will have degree $1$,
  \item vertices that are balanced $\pmod{k}$, i.e., 
  $(\mathsf{out}\text{-}\mathsf{degree}) - (\mathsf{in}\text{-}\mathsf{degree}) = 0 \pmod{k}$,
  \item vertices that are different from $\{\vec{0}\}$ and are not balanced $\pmod{k}$.
\end{enumerate}
\noindent Due to a simple modulo-$k$ argument and because $\{\vec{0}\}$ is not balanced
we can conclude that the constructed graph contains a vertex of type 3. Finally, we 
prove that all vertices of type 3 correspond to either a trichromatic triangle or a 
bichromatic edge with distinct non-consecutive labels and hence $k$-Polygon Tucker's Lemma 
follows. Our proof also gives us a reduction of the computational problem associated
with $k$-Polygon Tucker's Lemma to the problem \imba{k[\#1]}.
\smallskip

\noindent For any simplex $\sigma \in T$ we define $S(\sigma)$ and $\lambda(\sigma)$: 
\begin{itemize}
    \item[$\triangleright$] $S(\sigma) \subseteq [k]$ is the minimal subset of $[k]$
    such that $\sigma$ lies in the cone defined by $\{\vec{u}_i : i \in S(\sigma)\}$,
    \item[$\triangleright$] 
    $\lambda(\sigma) = \{\lambda(\vec{x}) : \vec{x} \text{ is a vertex of } \sigma\}$,
\end{itemize}
and we let $S(\{\vec{0}\}) = \emptyset$.
 
\begin{remark}
    Observe that because $T$ is a refinement of $T^*$, we have that every 
  simplex $\sigma \in T$ is contained in a cone defined by two consecutive 
  vectors $\vec{u}_i, \vec{u}_{i +1}$. Hence, for $\sigma \neq \{\vec{0}\}$,  
  $S(\sigma)$ contains either a single number $i \in [k]$ or two consecutive numbers
  $i, i +1$. 
\end{remark}

\begin{figure}
    \centering
    \includegraphics{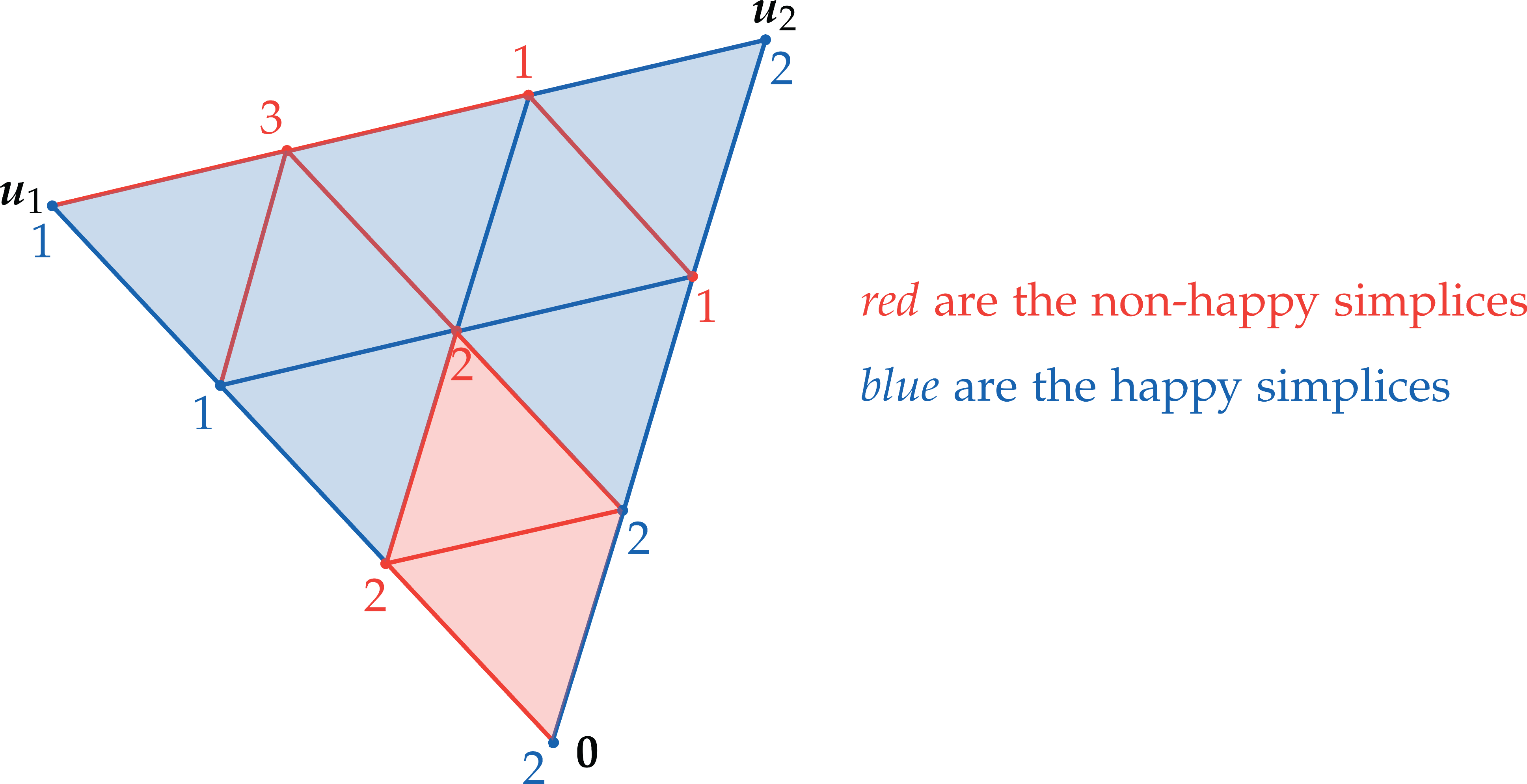}
    \caption{Example that shows happy and non-happy simplices in the cone $\chull(\{\vec{u}_1, \vec{0}, \vec{u}_2\})$.}
    \label{fig:happy}
\end{figure}

\begin{definition}[\textsc{Happy Simplices}] A simplex $\sigma \in T$ is called \textit{happy} if and only if 
$S(\sigma) \subseteq \lambda(\sigma)$. See also
\cref{fig:happy} for an example that explains the definition.
\end{definition}

  We will define a graph $G$ with vertex set $V(G) = T$. In $G$, we will only add
directed edges to the following vertices, which we call \textit{relevant}:
\begin{itemize}
  \item[(a)] vertices that correspond to a simplex $\sigma \in T$ such that 
  $\sigma \not\in \partial T$ and $\sigma$ is happy,
  \item[(b)] vertices that correspond to a simplex $\psi \in \partial T$ such that $\psi$ is happy and:
  \begin{itemize}
      \item $\psi = \{u_1\}$, if $\psi$ is $0$-dimensional
      \item $\psi \subseteq \chull(\{\vec{u}_1, \vec{u}_2\})$, if $\psi$ is $1$-dimensional.
  \end{itemize}
\end{itemize}
\noindent The reason for this distinction between simplices on the boundary and 
simplices not on the boundary is that we want to identify the symmetric simplices on
the boundary as a \emph{super vertex} in order to correctly use a modulo-$k$ argument, as
we described in the sketch of the proof.

\begin{remark}
    The rest of the vertices in $V(G)$ that do not correspond to type (a) or (b) simplices
  can be thought of as having 
  $(\mathsf{out}\text{-}\mathsf{degree}) = (\mathsf{in}\text{-}\mathsf{degree}) = 0$ or
  as having a self-loop. In both cases, these vertices are balanced.
\end{remark}

  We add an edge $(v, v')$ to the graph $G$ only if the simplices $\sigma$ and
$\sigma'$ that correspond to $v$ and $v'$ are both relevant and one of the following rules applies:
\begin{enumerate}
  \item \textbf{case} 
  $\boldsymbol{\sigma \not\in \partial T, \sigma' \not\in \partial T:}$ We
  add the edge if $\sigma' \subseteq \sigma$ and the labels of 
  $\sigma'$ suffice to make $\sigma$ happy, i.e., 
  $S(\sigma) \subseteq \lambda(\sigma')$,
  \item \textbf{case} 
  $\boldsymbol{\sigma \not\in \partial T, \sigma' \in \partial T:}$  Observe that since $v'$ is 
  relevant and $\sigma' \in \partial T$, $v'$ is of type (b). So, instead of checking whether $\sigma' \subseteq \sigma$, we check whether there exists $t \in [k]$ such that $\tau := \theta^{(t)}_k(\sigma') \subseteq \sigma$ and $\tau$ suffices to make $\sigma$ happy, i.e., $S(\sigma) \subseteq \lambda(\tau)$.
\end{enumerate}

\paragraph*{Directing the edges.} The edge between $\sigma$ and $\sigma'$ is directed in the following natural way:
\begin{itemize}
    \item if $\sigma$ is 1-dimensional, then $\sigma$ and $\sigma'$ lie in $\chull(\{0,\vec{u}_i\})$ for some $i$, and the edge is directed ``away from $0$''. Formally, if $\sigma = \{\vec{z}_0, \vec{z}_1\}$ and $\sigma' = \{\vec{z}_1\}$ are connected by an edge, then write $\vec{z}_0 = \alpha_i u_i$ and $\vec{z}_1 = \beta_i u_i$. If $\alpha_i-\beta_i > 0$, then the edge is incoming into $\sigma$. If $\alpha_i-\beta_i < 0$, then the edge is outgoing out of $\sigma$.
    \item if $\sigma$ is 2-dimensional, then $\sigma$ and $\sigma'$ lie in $\chull(\{0,\vec{u}_i,\vec{u}_j\})$ for some $i,j$ with $i-j = \pm 1 \pmod{k}$. If $j=i+1$, then the edge is directed such that ``we keep label $i$ to our right, and label $j=i+1$ to our left, when we move in the direction of the edge''. If $j=i-1$, then the edge is directed such that ``we keep label $j=i-1$ to our right, and label $i$ to our left, when we move in the direction of the edge''. Formally, if $\sigma = \{\vec{z}_0, \vec{z}_1, \vec{z}_2\}$ and $\sigma'=\{\vec{z}_1, \vec{z}_2\}$ are connected by an edge, where $\lambda(\vec{z}_1)=i$ and $\lambda(\vec{z}_2)=j$, then write $\vec{z}_0=\alpha_i \vec{u}_i+\alpha_j \vec{u}_j$, $\vec{z}_1=\beta_i \vec{u}_i+\beta_j \vec{u}_j$ and $\vec{z}_2=\gamma_i \vec{u}_i+\gamma_j \vec{u}_j$. Construct the matrix
    \begin{equation*}
        M= \left[\begin{tabular}{ll}
            $\alpha_i-\beta_i$ & $\alpha_i -\gamma_i$ \\
            $\alpha_j-\beta_j$ & $\alpha_j - \gamma_j$
        \end{tabular}\right]
    \end{equation*}
    If $\det M > 0$, then the edge is incoming into $\sigma$. If $\det M < 0$, then the edge is outgoing out of $\sigma$. Notice that $\det M \neq 0$, because $\sigma$ is a simplex. Furthermore, note that the determinant of the matrix does not change if we switch both $i$ and $j$, and $\vec{z}_1$ and $\vec{z}_2$ (i.e., $\beta$ and $\gamma$). Thus, the direction is well-defined.
\end{itemize}
If $\sigma'$ corresponds to a vertex of type (b), then we apply the rule above to $\sigma$ and $\tau$ instead.

\begin{figure}
    \centering
    \includegraphics{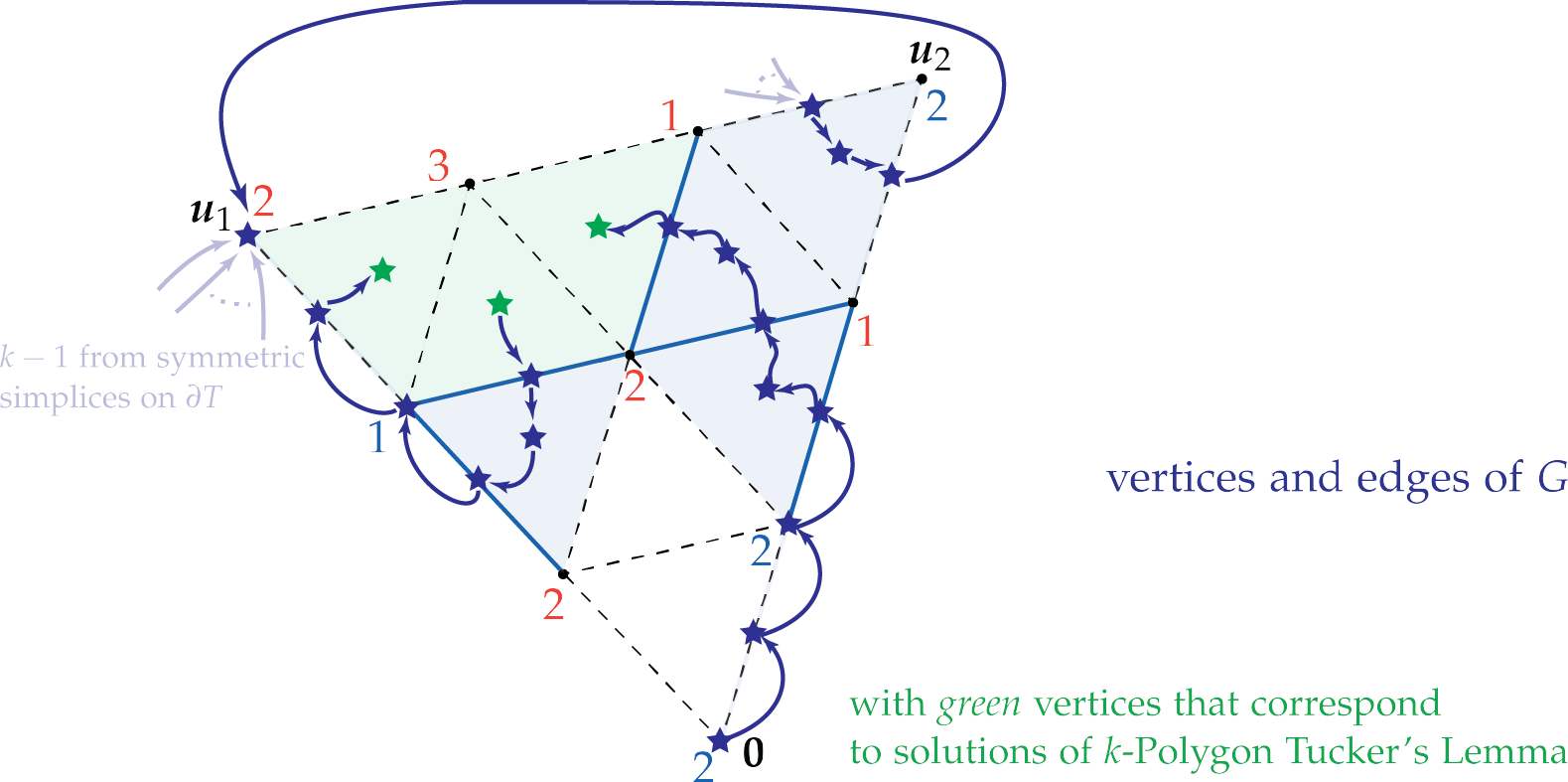}
    \caption{An example of the graph $G$ constructed in the proof of $k$-Polygon Tucker's Lemma, focused on the cone $\chull(\{\vec{u}_1, \vec{0}, \vec{u}_2\})$.}
    \label{fig:containment}
\end{figure}

  Based on the description of the edges in the graph $G$ we can prove the following
Lemmas which complete the proof of $k$-Polygon Tucker's Lemma.

\begin{lemma} \label{lem:kPolygonTucker:proof:trivialSimplex}
    The vertex that corresponds to the simplex $\{\vec{0}\} \in T$ has out-degree $1$ 
  and in-degree $0$.
\end{lemma}

\begin{proof}
    Since $\{\vec{0}\}$ is a 0-dimensional simplex, it does not have any sub-simplices
  and hence it can only be connected to a 1-dimensional simplex, i.e., an edge. An
  edge $\psi \in T$ that is not contained in a linear segment of the form
  $\chull(\{\vec{0}, \vec{u}_i\})$ cannot be a neighbor of $\{\vec{0}\}$, because it requires
  two labels to become happy and $\{\vec{0}\}$ has only one single label. 
  Hence, $\{\vec{0}\}$ can only be connected to a $1$-dimensional simplex $\{\vec{0}, \vec{z}_i\}$ contained 
  in $\chull(\{\vec{0}, \vec{u}_i\})$ for some $i$. 
  Observe also that $S(\{\vec{0}, \vec{z}_i\}) = \{i\}$, which implies that
  $S(\{\vec{0}, \vec{z}_i\}) = \{\lambda(\vec{0})\}$. Therefore, there is exactly one 1-dimensional 
  simplex that is connected to $\{\vec{0}\}$, namely $\{\vec{0}, \vec{z}_i\}$, where $i = \lambda(\vec{0})$.
  
  Furthermore, since $\{\vec{0}\}$ and its neighbor $\{\vec{0}, \vec{z}_i\}$ lie in $\chull(\{\vec{0}, \vec{u}_i\})$, the edge is directed away from $\{\vec{0}\}$. Formally, if we write $\vec{z}_i = \alpha_i u_i$ and $\vec{0}= \beta_i u_i$, it will hold that $\alpha_i-\beta_i = \alpha_i > 0$. This means that the edge is incoming into $\{\vec{0}, \vec{z}_i\}$ and the lemma follows.
\end{proof}

\begin{lemma} \label{lem:kPolygonTucker:proof:solutionSimplex}
    Any vertex $v$ in $G$ that is imbalanced modulo-$k$, i.e.,
  $(\mathsf{out}\text{-}\mathsf{degree}(v)) \neq (\mathsf{in}\text{-}\mathsf{degree}(v)) \pmod{k}$ 
  and does not correspond to the simplex $\{\vec{0}\}$, corresponds to a simplex $\sigma$
  that is either trichromatic or $\lambda(\sigma) \not\subseteq \{i, i+1\}$ for all $i \in [k]$.
\end{lemma}

\begin{proof}
    For this proof, we  consider all three cases for the dimension of the 
  simplex $\sigma^{\star}$ that corresponds to an imbalanced node of $G$ separately.
  \smallskip
  
  \noindent \textbf{Dimension of $\boldsymbol{\sigma^{\star}}$ is $\boldsymbol{0}$.} In this case, $\sigma^{\star} = \{\vec{z}^{\star}\}$. It is easy to see that
  $\sigma^{\star}$ cannot make happy any $2$-dimensional simplex since a  
  $2$-dimensional simplex $\sigma$ has $|S(\sigma)| = 2$. So, $\sigma^{\star}$ can only 
  be a neighbor of a $1$-dimensional simplex $\psi$. Additionally, $\sigma^{\star}$ cannot make happy
  any $1$-dimensional $\psi$ such that $|S(\psi)| = 2$. Thus, a neighbor of $\sigma^{\star}$
  should be contained in the segment
  $\chull(\{\vec{0}, \vec{u}_i\})$ where $i$ is such that $\lambda(\sigma^{\star}) = i$. If
  $\vec{z}^{\star} \neq \vec{u}_i$, then $\sigma^{\star}$ is connected to both of its
  two neighboring $1$-dimensional simplices in the segment $\chull(\{\vec{0}, \vec{u}_i\})$. Then, since the edges are directed away from $\vec{0}$, $\sigma^{\star}$ has one incoming and one outgoing edge. Formally, let $\{\vec{z}_0, \vec{z}^*\}$ and $\{\vec{z}^*,\vec{z}_0'\}$ be the two neighboring $1$-dimensional simplices, and write $\vec{z}_0 = \alpha_i u_i$, $\vec{z}_0' = \alpha_i' u_i$ and $\vec{z}^* = \beta_i u_i$. It is easy to see that $\alpha_i-\beta_i$ and $\alpha_i'-\beta_i$ always have opposite signs, since $\vec{z}_0$ and $\vec{z}_0'$ lie on opposite sides of $z^\star$ on $\chull(\{\vec{0}, \vec{u}_i\})$.
  
  If $\vec{z}^{\star} = \vec{u}_i$, then from the
  definition of relevant vertices of $G$ we have that $\vec{z}^{\star} = \vec{u}_1$ and $\{\vec{u}_1\}$ is happy, i.e., $\lambda(\vec{u}_1)=1$. Thus, by the boundary conditions, for every $t \in [k]$, it holds that $\lambda(\theta_k^{(t)}(\vec{u}_1)) = t + 1$. In other words, $\lambda(\vec{u}_i) = i$ for all $i \in [k]$. As a result, it follows that for each $i \in [k]$, $\{\vec{u}_1\}$ has an edge with the simplex $\{\vec{u}_i, \vec{z}_i\}$ which lies in $\chull(\{\vec{0},\vec{u}_i\})$. This holds because $\theta_k^{(i)}(\{\vec{u}_1\})$ suffices to make $\{\vec{u}_i, \vec{z}_i\}$ happy. Clearly, $\theta_k^{(i)}(\{\vec{u}_1\})$ cannot make any other simplex happy, by the same arguments as above. Finally, note that all the edges are incoming into $\{\vec{u}_1\}$, because edges are directed away from $\vec{0}$ on any $\chull(\{\vec{0},\vec{u}_i\})$. Formally, if we write $\vec{z}_i = \alpha_i u_i$ and $\vec{u}_i = \beta_i u_i$, then we always have $\alpha_i - \beta_i = \alpha_i -1 < 0$, so the edge is outgoing out of $\{\vec{u}_i,\vec{z}_i\}$. Thus, in this case, $\vec{z}^{\star}$ has $k$ neighbors and all of them with the 
  same direction. Therefore, $\vec{z}^{\star}$ is always balanced modulo-$k$. In conclusion, an
  imbalanced vertex of $G$ different from $\{\vec{0}\}$ cannot have dimension $0$.
  \smallskip
  
  \noindent \textbf{Dimension of $\boldsymbol{\sigma^{\star}}$ is $\boldsymbol{1}$.} First, assume 
  that $\sigma^{\star} = \{\vec{z}^{\star}_1, \vec{z}^{\star}_2\}$ belongs to
  one of the line segments $\chull(\{\vec{0}, \vec{u}_i\})$. If additionally 
  $\lambda(\vec{z}^{\star}_1) = \lambda(\vec{z}^{\star}_2)$, then $\sigma^{\star}$
  is happy if
  $\lambda(\vec{z}^{\star}_1) = \lambda(\vec{z}^{\star}_2) = i$. Therefore,
  $\abs{\lambda(\sigma^{\star})} = 1$ and hence $\sigma^\star$ cannot be a neighbor of a
  $2$-dimensional $\sigma$ since $\abs{S(\sigma)} = 2$ and $\sigma^{\star}$
  cannot make $\sigma$ happy. So, $\sigma^{\star}$ has exactly the two neighbors 
  $\{\vec{z}^{\star}_1\}$ and $\{\vec{z}^{\star}_2\}$ since both of them make 
  $\sigma^{\star}$ happy. Furthermore, the vertex is balanced, because one of the edges is incoming and the other one is outgoing, since edges are directed away from $\vec{0}$ on $\chull(\{\vec{0}, \vec{u}_i\})$. Formally, if we write $\vec{z}^\star_1 = \alpha_i u_i$ and $\vec{z}^\star_2 = \alpha_i' u_i$, then $\alpha_i -\alpha_i'$ and $\alpha_i'-\alpha_i$ always have opposite signs.
  
  The next case we consider is when
  $\sigma^{\star} \subseteq \chull(\{\vec{0}, \vec{u}_i\})$ with
  $i = \lambda(\vec{z}^{\star}_1) \neq \lambda(\vec{z}^{\star}_2) = j$. If 
  $i - j \neq \pm 1$, then $\sigma^\star$ yields a solution to $k$-Polygon Tucker, and so is allowed to be imbalanced. On the other hand, if $j = i \pm 1 \pmod{k}$, $\sigma^{\star}$ has
  exactly two neighbors, the simplex $\{\vec{z}^{\star}_1\}$, which makes
  $\sigma^{\star}$ happy, and the unique $2$-dimensional simplex $\sigma = \{\vec{z}^{\star}_1, \vec{z}^{\star}_2, \vec{z}_3\}$ that contains $\sigma^{\star}$ as a face
   and is contained in the cone 
  $\chull(\{\vec{u}_i, \vec{0}, \vec{u}_{j}\})$. Also, one of the edges is incoming and other one outgoing and hence $\sigma^{\star}$ cannot be imbalanced. Formally, write $\vec{z}_2^\star = \beta_i u_i$, $\vec{z}_1^\star = \beta_i' u_i$ and $\vec{z}_3 = \alpha_i u_i + \alpha_j u_j$. Then, it holds that the edge goes from $\{\vec{z}^{\star}_1\}$ to $\{\vec{z}^{\star}_1,\vec{z}^{\star}_2\}$ if $\beta_i'-\beta_i > 0$, and otherwise in the other direction. Furthermore, the edge goes from $\sigma = \{\vec{z}^{\star}_1, \vec{z}^{\star}_2\}$ to $\sigma = \{\vec{z}^{\star}_1, \vec{z}^{\star}_2, \vec{z}_3\}$, if $\det M > 0$, where we can compute that $\det M = (\alpha_i-\beta_i) \alpha_j - (\alpha_i-\beta_i')\alpha_j = \alpha_j (\beta_i'-\beta_i)$. Since $\alpha_j > 0$, it follows that both expressions have the same sign, and thus $\sigma^{\star}$ is balanced.
  
    The final case for $1$-dimensional $\sigma^{\star}$ is when 
  $\sigma^{\star}$ is not contained in any line segment of the form
  $\chull(\{\vec{0}, \vec{u}_i\})$. Let $\sigma^{\star}$
  be contained in the cone $\chull(\{\vec{u_i}, \vec{0}, \vec{u}_{i+1}\})$. If $\sigma^{\star}$
  is not relevant, then it cannot be imbalanced. To be relevant, and hence happy, it must hold that
  $\lambda(\vec{z}^{\star}_1) = i$ and $\lambda(\vec{z}^{\star}_2) = i+1$. If 
  $\sigma^{\star}$ is not in the boundary $\partial T$, then
  $\sigma^{\star}$ is the face of exactly two $2$-dimensional simplices 
  $\sigma$, $\sigma'$ that are both contained in 
  $\chull(\{\vec{u}_i, \vec{0}, \vec{u}_{i + 1}\})$. Therefore, $\sigma^{\star}$ makes
  happy both of them and no other simplex, and consequently it has an edge with both of them and no other edge. Furthermore, one of the edges is incoming and other one is outgoing (from the perspective of $\sigma^{\star}$), so it cannot be imbalanced. Intuitively, if we let $\sigma=\{\vec{z}^{\star}_1,\vec{z}^{\star}_2,\vec{z}_3\}$ and $\sigma'=\{\vec{z}^{\star}_1,\vec{z}^{\star}_2,\vec{z}_3'\}$, then $\vec{z}_3$ and $\vec{z}_3'$ lie on opposite sides of the line defined by $\{\vec{z}^{\star}_1,\vec{z}^{\star}_2\}$. Thus, the basis of $\mathbb{R}^2$ defined by $\{\vec{z}_3-\vec{z}^{\star}_1, \vec{z}_3-\vec{z}^{\star}_2\}$ has opposite orientation compared to the basis $\{\vec{z}_3'-\vec{z}^{\star}_1, \vec{z}_3'-\vec{z}^{\star}_2\}$. As a result, $\det M$ will have a different sign for the two edges. For a formal proof of this, we refer to the proof of \cref{thm:ptucker-in-ppa-p}, where we prove a more general version of this fact.

  If 
  $\sigma^{\star} \in \partial T$, then $\sigma^{\star}$ is relevant only if
  $\sigma^{\star} \subseteq \chull(\{\vec{u}_1, \vec{u}_2\})$. In this case, $\sigma^{\star}$ has exactly $k$ neighbors each 
   of which is a $2$-dimensional simplex that has as a face one of the $k$ 
  symmetric copies of $\sigma^{\star}$. Namely, the neighbors of $\sigma^{\star}$ are $\sigma_1, \dots, \sigma_k$, where $\theta_k^{(i)}(\sigma^\star)$ makes $\sigma_i$ happy for all $i \in [k]$. To see that all $k$ edges are incoming or all are outgoing, notice that if we have $1$ to our right and $2$ to our left when we reach the boundary, then it will hold that we have $i$ on our right and $i+1$ on our left when we reach the boundary at the corresponding position in cone $\chull(\{\vec{0},\vec{u}_i,\vec{u}_{i+1}\})$. More formally, the sign of the determinant of the matrix $M(\sigma_i, \theta_k^{(i)}(\sigma^\star))$ constructed to determine the direction of the edge with $\sigma_i$, will be the same for all $i \in [k]$. For a full formal proof, we again  refer to the proof of \cref{thm:ptucker-in-ppa-p}.

  \noindent \textbf{Dimension of $\boldsymbol{\sigma^{\star}}$ is $\boldsymbol{2}$.} Let 
  $\sigma^{\star} = \chull(\{\vec{z}^{\star}_1, \vec{z}^{\star}_2, \vec{z}^{\star}_3\})$.
  Assume that $\sigma^{\star}$ is contained in the simplex 
  $\chull(\{\vec{u}_i, \vec{0}, \vec{u}_j\})$, with $j = i \pm 1 \pmod{k}$ and
  without loss of generality $\lambda(\vec{z}^{\star}_1) = i$, 
  $\lambda(\vec{z}^{\star}_2) = j$. If $\lambda(\vec{z}^{\star}_3) \not\in \{i, j\}$ 
  then $\sigma^{\star}$ is a trichromatic triangle, and thus it can be imbalanced. 
  Assume that
  $\lambda(\vec{z}^{\star}_3) \in \{i, j\}$ and without loss of generality 
  $\lambda(\vec{z}^{\star}_3) = i$. In this case, $\sigma^{\star}$ has exactly two 
  neighbors: the $1$-dimensional simplices $\psi_1 = \chull(\{\vec{z}^{\star}_1, \vec{z}^{\star}_2\})$ and
  $\psi_2 = \chull(\{\vec{z}^{\star}_3, \vec{z}^{\star}_2\})$. Once again, one edge is incoming and the other one is outgoing, and thus $\sigma^{\star}$ is balanced. Intuitively, this follows from the fact that if we move with $i$ to our left and $j$ to our right, then we can ``enter'' the simplex from one side, and ``exit'' from the other one. More formally, if we write $\vec{z}^\star_1 = \alpha_i u_i + \alpha_j u_j$, $\vec{z}^\star_2 = \beta_i u_i + \beta_j u_j$ and $\vec{z}^\star_3 = \alpha_i' u_i + \alpha_j' u_j$, then it holds that:
  \begin{equation*}
      \det \left[ \begin{tabular}{ll}
        $\alpha_i-\beta_i$   & $\alpha_i-\alpha_i'$ \\
        $\alpha_j - \beta_j$   & $\alpha_j - \alpha_j'$
      \end{tabular} \right] = - \det \left[ \begin{tabular}{ll}
        $\alpha_i'-\beta_i$   & $\alpha_i'-\alpha_i$ \\
        $\alpha_j' - \beta_j$   & $\alpha_j' - \alpha_j$
      \end{tabular} \right]
  \end{equation*}
  by using standard rules about the determinant.
 
  \medskip
  
  Hence, the only imbalanced vertices different from $\{\vec{0}\}$ correspond to either trichromatic simplices or
  simplices such that $\lambda(\sigma) \not\subseteq \{i, i+1\}$ for all $i \in [k]$.
\end{proof}

  Finally, from the definition of the graph $G$ and 
\cref{lem:kPolygonTucker:proof:trivialSimplex}, we have that there has to be a vertex in
$G$ that is imbalanced $\pmod{k}$ and different from $\{\vec{0}\}$. By
\cref{lem:kPolygonTucker:proof:solutionSimplex}, any such vertex proves
the validity of $k$-Polygon Tucker's Lemma.

\subsubsection{Computational Problems and Containment in $\ppak{k}[\#1]$} \label{sec:kPolygonTuckerkPolygonBorsukUlam:containment}

  In this section, we define the computational problems associated with
$k$-Polygon Tucker's Lemma and the $k$-Polygon Borsuk-Ulam Theorem, which we call 
\polygonTucker{k} and \polygonBorsukUlam{k} respectively. Following the ideas 
presented in \cref{sec:kPolygonTucker:proof}, we show that \polygonTucker{k}
is in $\ppak{k}[\#1]$. The membership of \polygonTucker{k} in $\ppak{k}[\#1]$ combined with the results of
\cref{sec:kPolygonTuckerkPolygonBorsukUlam:equivalence} implies that 
\polygonBorsukUlam{k} is also in $\ppak{k}[\#1]$.

  To define the computational problem associated with $k$-Polygon Tucker's Lemma we
need succinct access to the labels $\lambda(\vec{x})$ for the vertices $\vec{x}$ in
a nice triangulation $T$ of $W_k$. This succinct access resembles the one in the 
definition of the computational version of the original Borsuk-Ulam Theorem
\citep{ABB15-2DTucker, Papadimitriou94-TFNP-subclasses} and the computational version 
of Brouwer's Fixed Point Theorem \citep{Papadimitriou94-TFNP-subclasses}. To define
this succinct access, we fix for any $m \in \mathbb{N}$ a triangulation $T(m)$ with 
 diameter (i.e., max distance between two vertices of a simplex)  at most $1/2^m$, where we can refer to a
simplex in $T(m)$ using $O(m + k)$ bits. For our case this triangulation will be the 
following.

\begin{definition}[\textsc{Edge Parallel Triangulation}]\label{def:edge-par-triang}
    For every $m \in \mathbb{N}$, we define $\widehat{T}(m)$ to be the following nice 
  triangulation of $W_k$. Starting from $T^{\star}$ (see  
  \cref{def:kregularPolygonNiceTriangulation}) we define a simplicial complex 
  $\widehat{T}_i(m)$ of every simplex 
  $\sigma^{\star}_i = \chull(\{\vec{u}_i, \vec{0}, \vec{u}_{i + 1 \pmod{k}}\})$ and 
  then  set $\widehat{T}(m) = \cup_{i \in [k]} \widehat{T}_i(m)$. To define $\widehat{T}_i(m)$, we 
  divide the edges $\psi^i_1 = \chull(\{\vec{u}_i, \vec{0}\})$, 
  $\psi^i_2 = \chull(\{\vec{u}_i, \vec{u}_{i + 1 \pmod{k}}\})$ and
  $\psi^i_3 = \chull(\{\vec{0}, \vec{u}_{i + 1 \pmod{k}}\})$ of $\sigma^{\star}_i$ equally into
  $2^{m + 1}$ intervals. Then, from any endpoint of the subintervals of the edge 
  $\psi^i_2$ we consider the lines that are parallel to either $\psi^i_1$ or $\psi^i_3$
  and from any endpoint of the subintervals in $\psi_1^i$ and $\psi_3^i$ we consider the line that is parallel
  to $\psi^i_2$, as shown in \cref{fig:containment}.  We define $\widehat{T}_i(m)$ to be the
  set of simplices that are created by these lines and lie inside $\sigma^{\star}_i$. Namely, the intersection points and endpoints of subintervals  are the $0$-dimensional simplices in $\widehat{T}_i(m)$, the line segments and the triangles between $0$-dimensional simplices are also simplices in $\widehat{T}_i(m)$. It is simple to see that $\widehat{T}(m)$ is nice;  we call $\widehat{T}(m)$ an \textit{edge parallel triangulation} of $W_k$.
\end{definition}

\noindent  For the edge parallel triangulation $\widehat{T}(m)$ of $W_k$ the following facts are
easy to verify.

\begin{fact} \label{fct:parallelTriangulationDiameter}
    The diameter of $\widehat{T}(m)$ is at most $1/2^m$.
\end{fact}

\begin{fact} \label{fct:parallelTriangulationIndex}
    We can indicate uniquely a simplex 
  $\sigma^{\star}_i$ using $\lceil \log(k) \rceil$ bits. Then, using $2 m + 3$ bits we can indicate uniquely a combination of
  two lines: (1) one of the $2 \cdot 2^{m + 1}$ lines that are parallel to either
  $\psi^i_1$ or $\psi^i_3$, and (2) one of the $2^{m + 1}$ lines that are parallel to
  $\psi^i_2$. This combination uniquely determines a point $\vec{x} \in \mathbb{R}^2$
  and we can efficiently check whether this point belongs to $\sigma^{\star}_i$ or 
  not. Of course this way the points that lie on the rays from $\vec{0}$ to $\vec{u}_i$ have two different representations but we can easily resolve this discrepancy by choosing as valid only the representation that is lexicographically first. Hence,  we can uniquely determine any vertex in
  $\widehat{T}(m)$ with $b = \lceil \log(k) \rceil + 2 m + 3$ bits.
\end{fact}

\begin{fact} \label{fct:parallelTriangulationSimplexCheck}
    Given a set of binary vectors $A = \{\vec{a}_i\}_i$ where 
  $\vec{a}_i \in \{0, 1\}^b$, there exists an efficient procedure that determines
  whether $\chull(A)$ is a simplex of $\widehat{T}(m)$ or not.  
\end{fact}

  Because of \cref{fct:parallelTriangulationIndex}, we can assume that the labeling
$\lambda$ of $\widehat{T}(m)$ is given via a circuit $\mathcal{L}$ with 
$b = \lceil \log(k) \rceil + 2 m + 3$ input bits and $\lceil \log(k) \rceil$ output bits. The input to the
circuit $\mathcal{L}$ is the representation of a potential vertex $\vec{x}$ in $\widehat{T}(m)$ according to \cref{fct:parallelTriangulationIndex} 
and the output is the label $\lambda(x) \in [k]$ of this vertex $\vec{x}$. Observe that \cref{fct:parallelTriangulationIndex} also guarantees that it is easy to check whether
an input $\vec{a} \in \{0, 1\}^b$ to the circuit $\mathcal{L}$ corresponds to a valid
vertex $\vec{x} \in \widehat{T}(m)$. We are now ready to define the 
 total search problem that is associated with $k$-Polygon Tucker's Lemma.

\bigskip
\noindent\fbox{%
\colorbox{gray!10!white}{
    \parbox{\textwidth}{%
\noindent\underline{\polygonTucker{k}} 
\smallskip

\noindent \textsc{Input:} A circuit $\mathcal{L} : \{0, 1\}^b \to \{0, 1\}^{\lceil \log(k) \rceil}$,
with $b = \lceil \log(k) \rceil + 2 m + 3$.
\smallskip

\noindent \textsc{Output:} One of the following.
\vspace{-0.07in}
\begin{enumerate}
    \item Two binary vectors $\vec{a}_1, \vec{a}_2 \in \{0, 1\}^b$ such that $\vec{a}_1$, 
    $\vec{a}_2$ correspond to vertices $\vec{x}_1$, $\vec{x}_2$ on $\partial \widehat{T}(m)$ with 
    $\vec{x}_2 = \theta_k(\vec{x}_1)$, but $\mathcal{L}(\vec{a}_2) \neq \mathcal{L}(\vec{a}_1) + 1 \pmod{k}$.
    \vspace{-0.07in}
    \item Three binary vectors $\vec{a}_1, \vec{a}_2, \vec{a}_3 \in \{0, 1\}^b$ such 
    that the simplex $\sigma = \chull(\{\vec{a}_1, \vec{a}_2, \vec{a}_3\})$ belongs to
    $\widehat{T}(m)$ and all the labels $\mathcal{L}(\vec{a}_1)$,
    $\mathcal{L}(\vec{a}_2)$, $\mathcal{L}(\vec{a}_3)$ are different from each other.
    \vspace{-0.07in}
    \item Two binary vectors $\vec{a}_1, \vec{a}_2 \in \{0, 1\}^b$ such that the
    edge $\psi = \chull(\{\vec{a}_1, \vec{a}_2\})$ belongs to $\widehat{T}(m)$ and labels
    $\mathcal{L}(\vec{a}_1)$, $\mathcal{L}(\vec{a}_2)$ are different and they satisfy $\mathcal{L}(\vec{a}_1) - \mathcal{L}(\vec{a}_2) \neq \pm 1 \pmod{k}$.
\end{enumerate}}}}
\bigskip

\begin{lemma} \label{lem:kpolygonTuckerInPPAK}
    It holds that \polygonTucker{k} is in $\ppak{k}[\#1]$.
\end{lemma}

\begin{proof}
    This lemma follows from the proof of $k$-Polygon Tucker's Lemma that we 
  presented in \cref{sec:kPolygonTucker:proof}. We only need to
  add the description of the circuits $\mathcal{S}$, 
  $\mathcal{P}$ that define the graph $G$ constructed in the proof. Then, using these circuits
  we reduce \polygonTucker{k} to 
  \imba{k[\#1]}. As we showed in \cref{lem:kPolygonTucker:proof:solutionSimplex},
  every solution to the resulting instance of \imba{k[\#1]} corresponds to a solution
  of \polygonTucker{k}. Hence, \polygonTucker{k} is in $\ppak{k}[\#1]$.
  
    Since the vertices of the graph $G$ correspond to simplices in $\widehat{T}(m)$ and 
  the maximum degree of any node in $G$ is $k$, we define the circuits $\mathcal{S}$,
  $\mathcal{P}$ to have $3 \cdot b$ binary inputs and $k \cdot (3 b)$ outputs, hence
  $\mathcal{S} : \{0, 1\}^{3 b} \to (\{0, 1\}^{3 b})^k$ and 
  $\mathcal{P} : \{0, 1\}^{3 b} \to (\{0, 1\}^{3 b})^k$. For the construction of the circuits, both
  $\mathcal{S}$ and $\mathcal{P}$ first check whether an input 
  $(\vec{a}_1, \vec{a}_2, \vec{a}_3)$ is \textit{valid} or not. An input 
  $(\vec{a}_1, \vec{a}_2, \vec{a}_3)$ is valid in the following cases.
  \begin{itemize}
    \item[$\triangleright$] If all $\vec{a}_1, \vec{a}_2, \vec{a}_3$ are different, 
    then $(\vec{a}_1, \vec{a}_2, \vec{a}_3)$ is valid if
    $\vec{a}_1, \vec{a}_2, \vec{a}_3$ are in lexicographical order and the simplex
    $\chull(\{\vec{x}_1, \vec{x}_2, \vec{x}_3\})$, where $\vec{x}_i$ is the point in 
    $\mathbb{R}^2$ that corresponds to $\vec{a}_i$ according to \cref{fct:parallelTriangulationIndex}, is a valid simplex of $\widehat{T}(m)$.
    Observe that we can check this using a polynomial size circuit by \cref{fct:parallelTriangulationSimplexCheck}.
    \item[$\triangleright$] If  $\{\vec{a}_1, \vec{a}_2, \vec{a}_3\}$ has two
    different vectors, then $(\vec{a}_1, \vec{a}_2, \vec{a}_3)$ is valid if
    the lexicographically first is $\vec{a}_1$, the lexicographically second is
    $\vec{a}_2 = \vec{a}_3$ and the edge
    $\chull(\{\vec{x}_1, \vec{x}_2\})$, where $\vec{x}_i$ is the point in 
    $\mathbb{R}^2$ that corresponds to $\vec{a}_i$ according to \cref{fct:parallelTriangulationIndex}, is a valid simplex of $\widehat{T}(m)$. Observe
    that we can check this using a polynomial size circuit by \cref{fct:parallelTriangulationSimplexCheck}.
    \item[$\triangleright$] If $\vec{a}_1 = \vec{a}_2 = \vec{a}_3$, then $(\vec{a}_1, \vec{a}_2, \vec{a}_3)$ is valid if the point 
    $\vec{x}_1$ in $\mathbb{R}^2$ that corresponds to $\vec{a}_1$ according to \cref{fct:parallelTriangulationIndex} is a valid vertex of $\widehat{T}(m)$. Observe
    that we can check this using a polynomial size circuit by \cref{fct:parallelTriangulationSimplexCheck}.
  \end{itemize}
  
    If input 
  $(\vec{a}_1, \vec{a}_2, \vec{a}_3)$ is not valid, then both circuits 
  $\mathcal{S}$ and $\mathcal{P}$ output $(\vec{a}_1, \vec{a}_2, \vec{a}_3)$ 
  concatenated with itself $k$ times. On the other hand, if 
  $(\vec{a}_1, \vec{a}_2, \vec{a}_3)$ is valid, then both circuits check whether
  this input corresponds to a relevant simplex as we defined in  
  \cref{sec:kPolygonTucker:proof}. It is easy to see that checking relevance can 
  be done efficiently. Again, if 
  $(\vec{a}_1, \vec{a}_2, \vec{a}_3)$ is not relevant, both circuits $\mathcal{S}$
  and $\mathcal{P}$ output $(\vec{a}_1, \vec{a}_2, \vec{a}_3)$ concatenated with
  itself $k$ times. Finally, if $(\vec{a}_1, \vec{a}_2, \vec{a}_3)$ is valid and 
  relevant, then we define the edges of the corresponding vertex in $G$ as described
  in \cref{sec:kPolygonTucker:proof}. From the construction of $G$, it follows that 
  the successors and the predecessors can be computed efficiently. If the
  vertex in $G$ that corresponds to $(\vec{a}_1, \vec{a}_2, \vec{a}_3)$ has less than
  $k$ incoming or outgoing edges, then we repeat the lexicographically last neighboring
  vertex enough times such that the number of bits in the output of both $\mathcal{S}$
  and $\mathcal{P}$ is $k \cdot (3 b)$.
  
    Using our analysis in \cref{sec:kPolygonTucker:proof}, it follows that the above 
  construction of $\mathcal{S}$ and $\mathcal{P}$ defines a reduction from 
  \polygonTucker{k} to \imbal{k}{1}.
\end{proof}

We now define the computational problem associated with the $k$-Polygon Borsuk-Ulam
Theorem. For this
computational problem we need a representation of the continuous function $f$ that is
the input to the $k$-Polygon Borsuk-Ulam Theorem. Following standard techniques in
the literature, we use arithmetic circuits with gates $\times \zeta$ (multiplication by a constant), $+$, $-$, $<$, $\min$, and 
$\max$ and rational constants to define this function.

\bigskip
\noindent\fbox{%
\colorbox{gray!10!white}{
    \parbox{\textwidth}{%
\noindent\underline{\polygonBorsukUlam{k}} 
\smallskip

\noindent \textsc{Input:} An arithmetic circuit $\mathcal{C} : B^2 \to \bbR^2$, an
accuracy parameter $\varepsilon > 0$ and a Lipschitz-constant $L$.
\smallskip

\noindent \textsc{Output:} One of the following.
\vspace{-0.07in}
\begin{enumerate}
    \item A point $\bx \in S^1$ such that $\norm{\mathcal{C}(\Tilde{\theta}_k(\bx)) - \Tilde{\theta}_k(\mathcal{C}(\bx))} \geq \eta(\varepsilon,k) := \varepsilon/8 k^4$.
    \vspace{-0.07in}
    \item Two points $\bx, \by \in B^2$ such that $\norm{\mathcal{C}(\bx)-\mathcal{C}(\by)} > L \norm{\bx-\by}$.
    \vspace{-0.07in}
    \item A point $\bx^* \in B^2$ such that $\norm{\mathcal{C}(\bx^*)} \leq \varepsilon$.
\end{enumerate}}}}

\bigskip

\noindent The first type of solution corresponds to a violation of the boundary conditions. Since we cannot compute $\theta_k$ exactly, we let $\tilde{\theta}_k$ denote a function that computes $\theta_k$ with error at most $\xi(\varepsilon,L,k) := \frac{\varepsilon}{2^7 L k^4}$. The second type of solution corresponds to a violation of $L$-Lipschitz-continuity. Note that the function computed by the circuit $\mathcal{C}$ might not be continuous, because of the comparison gate. Thus we add this extra violation, which ensures that the function is Lipschitz-continuous. This allows us to relate this problem to \polygonTucker{k} (and in particular to show that it always has a solution).

\smallskip

\begin{lemma} \label{lem:kPolygonBorsukUlamInPPAk}
    The problem \polygonBorsukUlam{k} reduces to the problem \polygonTucker{k}. 
  Therefore, \polygonBorsukUlam{k} is in $\ppak{k}$.
\end{lemma}

\begin{proof}
  This result is obtained by following the idea of the proof of
\cref{lem:polygonTuckerImpliesPolygonBorsukUlam}, and constructing a labeling
$\mathcal{L}$ that uses the circuit $\mathcal{C}$ as a sub-routine to compute
the labels of the corresponding \polygonTucker{k} instance. The details are a
bit tricky because we have to account for small errors in various computations.

In more detail, the regular procedure for picking a label at some point is to first compute an approximate value of its coordinates and then based on this to compute an approximate value of the function. However, this might introduce bogus boundary condition violations, even if the function perfectly satisfies the boundary conditions. We can resolve this as follows. For any vertex $\bx$ that lies on the boundary, but not between $\bu_1$ and $\bu_2$, we first check how far away the label obtained through the regular procedure is from the label that satisfies the boundary condition. If it happens that the regular procedure is close to outputting the label that does not violate the boundary conditions then we enforce the output of this label and ignore the output of the regular procedure. Otherwise, we output the label of the regular procedure. Then, we can show that from any solution of $k$-Polygon Tucker we can either extract an approximate zero of $\mathcal{C}$, or a violation of the boundary conditions of $\mathcal{C}$ (in their approximate version), or a violation of Lipschitz continuity.

  We will use the triangulation $\widehat{T}(m)$ where $m$ is picked
sufficiently small so that the diameter of the triangulation when mapped to
$B^2$ by the homeomorphism is at most $\delta=\frac{\varepsilon}{2^7 L k^4}$. The Boolean circuit
computing $\mathcal{L}$ performs the following operations. Recall that the input
to the circuit is the bits representing the index of the vertex 
$\bx \in W_k$ of the triangulation, as per \cref{fct:parallelTriangulationIndex}.
\begin{enumerate}
    \item Let $\by \in B^2$ denote the image of $\bx$ under the homeomorphism
    used in the proof of \cref{lem:polygonTuckerImpliesPolygonBorsukUlam}.
    Compute an approximation $\Tilde{\by}$ of $\by$ with error at most 
    $\eta_1 = \delta = \frac{\varepsilon}{2^7 L k^4}$, i.e., $\|\by - \Tilde{\by}\| \leq \eta_1$. Note this
    is done in two steps, first compute an approximate value for the coordinates
    of $\bx$, which is given through its index, and compute an approximate value
    of the homeomorphism given the approximate value of the coordinates of 
    $\bx$. These steps can be done efficiently by using standard techniques 
    \citep{Brent1976} so that the total approximation error is $\eta_1$.
    \item Compute $\mathcal{C}(\Tilde{\by})$ exactly.
    \item For each $i \in [k]$, compute an estimate $v_i(\Tilde{\by})$ of the 
    inner product $\langle \mathcal{C}(\Tilde{\by}), \bu_i \rangle$ with error
    at most $\eta_2 = \varepsilon/64 k^4$, where the approximation error is introduced
    in computing the coordinates of $\bu_i$.
    \item If $\bx \notin \partial \widehat{T}(m)$ or if 
    $\bx \in \textup{conv}(\{\bu_1, \bu_2\}) \setminus \{\bu_2\}$, then output 
    the label $\argmax_i v_i(\Tilde{\by})$ (break ties arbitrarily but
    deterministically, e.g., lexicographically).
    \item Otherwise, there exists $\ell \in [k-1]$ and a vertex 
    $\bx' \in \textup{conv}(\{\bu_1, \bu_2\}) \setminus \{\bu_2\}$, such that 
    $\bx = \theta_k^{\ell} (\bx')$. In that case, compute 
    $r = \max_i v_i(\Tilde{\by})$ and $j = \mathcal{L}(\bx') + \ell \pmod k$. Let $\eta_3 = \frac{\varepsilon}{16 k^4} (3 + k)$.
    If $r - v_j(\Tilde{\by}) \le \eta_3$ 
    then output 
    $\mathcal{L}(\bx') + \ell \pmod k$, otherwise output
    $\argmax_i v_i(\Tilde{\by})$ (again break ties arbitrarily but 
    deterministically, e.g., lexicographically).
\end{enumerate}

We will use the following useful invariant: if vertex $\bx$ has label $j$, then
$$\langle \mathcal{C}(\tilde{\by}), \bu_j \rangle \geq \max_i \langle \mathcal{C}(\tilde{\by}), \bu_i \rangle - 2\eta_2 - \eta_3.$$
This follows from our labeling procedure.

We begin by considering standard solutions of \polygonTucker{k} and distinguish between
the cases $k = 3$ and $k > 3$. The other type of solution, namely boundary violations are treated for any $k \geq 3$ at the end of the proof.

\paragraph{$\boldsymbol{k = 3}$.} In this case a standard solution to $\mathcal{L}$ 
consists of a simplex with vertices $\bx_1, \bx_2, \bx_3$ that have labels
$1, 2, 3$ respectively. Let $\Tilde{\by}_1, \Tilde{\by}_2, \Tilde{\by}_3$ denote
the corresponding points computed in the aforementioned step 1. of the 
computation of $\mathcal{L}$ with input $\bx_1$, $\bx_2$, $\bx_3$ respectively. Assume that $\norm{\mathcal{C}(\Tilde{\by}_i)} \ge \varepsilon$ for 
$i = 1, 2, 3$, since we have found a solution otherwise.
Using elementary linear algebra, it is easy to show that for any 
$\Tilde{\by} \in B^2$ with $\norm{\mathcal{C}(\Tilde{\by})} \ge \varepsilon$, it
holds that 
$\max_i \langle \mathcal{C}(\Tilde{\by}), \bu_i \rangle \ge \varepsilon/2$ and
$\min_i \langle \mathcal{C}(\Tilde{\by}), \bu_i \rangle \le -\varepsilon/2$.

Let $j = \argmin_i \langle \mathcal{C}(\Tilde{\by}_1), \bu_i \rangle$. Then it follows that $\langle \mathcal{C}(\Tilde{\by}_1), \bu_j \rangle \le -\varepsilon/2$, which
implies that 
$\langle \mathcal{C}(\Tilde{\by}_j), \bu_j \rangle \le -\varepsilon/4$, since
\[ |\langle \mathcal{C}(\Tilde{\by}_1), \bu_j \rangle - \langle \mathcal{C}(\Tilde{\by}_j), \bu_j\rangle| \leq \|\mathcal{C}(\Tilde{\by}_1) - \mathcal{C}(\Tilde{\by}_j)\| \leq L \|\Tilde{\by}_1 - \Tilde{\by}_j\| \leq L(\delta+2\eta_1) \leq \varepsilon/4 \]
unless $\Tilde{\by}_1$ and $\Tilde{\by}_j$ yield a violation of $L$-Lipschitz continuity of $\mathcal{C}$, in which case we have found a solution. Now, recall that $\bx_j$ has label $j$. By the invariant, it must hold that $\langle \mathcal{C}(\tilde{\by}_j), \bu_j \rangle \geq \max_i \langle \mathcal{C}(\tilde{\by}_j), \bu_i \rangle - 2\eta_2 - \eta_3$. But since $\max_i \langle \mathcal{C}(\tilde{\by}_j), \bu_i \rangle \geq \varepsilon/2$, this means that $\langle \mathcal{C}(\tilde{\by}_j), \bu_j \rangle \geq \varepsilon/2 - 2\eta_2 - \eta_3$, which contradicts the fact that $\langle \mathcal{C}(\Tilde{\by}_j), \bu_j \rangle \le -\varepsilon/4$.

\paragraph{$\boldsymbol{k > 3}$.} In this case a solution to $\mathcal{L}$ consists of an edge $\bx_1,\bx_2$ of the triangulation such that the two vertices have labels $j_1,j_2$ that differ by more than one. Assume that $\|\mathcal{C}(\Tilde{\by}_1)\| \geq \varepsilon$ and $\|\mathcal{C}(\Tilde{\by}_2)\| \geq \varepsilon$. Note that $\tilde{\by}_1$\ lies in some cone $\textup{conv}(\{\boldsymbol{0},\bu_t,\bu_{t+1}\})$. By elementary linear algebra, one can show that
\begin{equation}\label{eq:linear-algebra}
\max_{i \in \{t,t+1\}} \langle \mathcal{C}(\Tilde{\by}_1), \bu_i \rangle - \max_{i \notin \{t,t+1\}} \langle \mathcal{C}(\Tilde{\by}_1), \bu_i \rangle \geq \frac{\varepsilon}{2}(1-\cos(2\pi/k)) \ge \frac{\varepsilon}{k^3}
\end{equation}
Since $\frac{\varepsilon}{k^3} \geq \eta_3+2\eta_2$, it follows that the label $j_1$ of $x_1$ must be one of $t$ or $t+1$. As a result, $j_2 \notin \{t,t+1\}$, since it cannot be adjacent to $j_1$. Furthermore, since $j_1$ was picked as the label of $x_1$, by the invariant it holds that $\langle \mathcal{C}(\tilde{\by}_1), \bu_{j_1} \rangle \geq \max_i \langle \mathcal{C}(\tilde{\by}_1), \bu_i \rangle - 2\eta_2 - \eta_3$. Together with equation \eqref{eq:linear-algebra}, it follows that
$$\langle \mathcal{C}(\tilde{\by}_1), \bu_{j_1} \rangle - \langle \mathcal{C}(\tilde{\by}_1), \bu_{j_2} \rangle \geq \frac{\varepsilon}{k^3} -2\eta_2-\eta_3.$$
On the other hand, since $x_2$ has label $j_2$, this means that
$$\langle \mathcal{C}(\tilde{\by}_2), \bu_{j_1} \rangle - \langle \mathcal{C}(\tilde{\by}_2), \bu_{j_2} \rangle \leq 2\eta_2 + \eta_3.$$
As a result, it must be that $\|\mathcal{C}(\tilde{\by}_1) - \mathcal{C}(\tilde{\by}_2)\| \geq \frac{\varepsilon}{2k^3} - 2\eta_2 - \eta_3$. Note that this is a violation of Lipschitz continuity, since normally we should have
$$\|\mathcal{C}(\tilde{\by}_1) - \mathcal{C}(\tilde{\by}_2)\| \leq L \|\tilde{\by}_1 - \tilde{\by}_2\| \leq L (\delta + 2 \eta_1)$$
and this quantity is smaller than $\frac{\varepsilon}{2k^3} - 2\eta_2 - \eta_3$.

\paragraph{Boundary violations.}
  Let $\bx$ be a violation of the boundary conditions, i.e., the label at vertex
$\bx$ is $j_1$, and the label at $\theta_k (\bx)$ is $j_2 \neq j_1 + 1 \pmod k$.
This means that at least one of the two vertices did not obtain its ``intended''
label. Without loss of generality, assume that $\bx$ has not obtained
its intended label. Let $\ell \in [k - 1]$ and 
$\bx_0 \in \textup{conv}(\{\bu_1,\bu_2\}) \setminus \{\bu_2\}$ be such that
$\theta_k^\ell(\bx_0) = \bx$. Since $\bx$ did not obtain its ``intended'' label,
it follows that $\max_i v_i(\Tilde{\by}) - v_j(\Tilde{\by}) > \eta_3$, where 
$\Tilde{\by}$ is the corresponding point in $B^2$ computed in step 1. This
implies that
\[ \max_i \langle \mathcal{C}(\Tilde{\by}), \bu_i \rangle - \langle \mathcal{C}(\Tilde{\by}), \bu_j \rangle > \eta_3 - 2\eta_2. \]
On the other hand, since 
$\bx_0 \in \textup{conv}(\{\bu_1,\bu_2\}) \setminus \{\bu_2\}$ has label 
$j - \ell \pmod k$, it holds that
\[ \max_i \langle \mathcal{C}(\Tilde{\by}_0), \bu_i \rangle - \langle \mathcal{C}(\Tilde{\by}_0), \bu_{j-\ell \pmod k} \rangle \leq 2 \eta_2. \]
\noindent where $\Tilde{\by}_0$ is the corresponding point of $\bx_0$ in $B^2$ 
computed in step 1. By noting that
\[ \max_i \langle \mathcal{C}(\Tilde{\by}_0), \bu_i \rangle - \langle \mathcal{C}(\Tilde{\by}_0), \bu_{j-\ell \pmod k} \rangle = \max_i \langle \theta_k^\ell(\mathcal{C}(\Tilde{\by}_0)), \bu_i \rangle - \langle \theta_k^\ell(\mathcal{C}(\Tilde{\by}_0)), \bu_j \rangle \]
we thus obtain that there exists $i \in [k]$ such that
\[|\langle \mathcal{C}(\Tilde{\by}), \bu_i \rangle - \langle \theta_k^\ell(\mathcal{C}(\Tilde{\by}_0)), \bu_i \rangle| > (\eta_3-4\eta_2)/2.\]
This implies that $\|\mathcal{C}(\Tilde{\by}) - \theta_k^\ell(\mathcal{C}(\Tilde{\by}_0))\| \geq (\eta_3-4\eta_2)/2$. Since
\[ \mathcal{C}(\Tilde{\by}) - \theta_k^\ell(\mathcal{C}(\Tilde{\by}_0)) = \sum_{i=0}^{\ell-1} \theta_k^i\mathcal{C}(\theta_k^{-i}(\Tilde{\by})) - \theta_k^{i+1}(\mathcal{C}(\theta_k^{\ell-i-1}(\Tilde{\by}_0))) \]
it follows that there exists $i \in \{0,1,\dots, \ell-1\}$ such that 
\[ \|\theta_k^i\mathcal{C}(\theta_k^{-i}(\Tilde{\by})) - \theta_k^{i+1}(\mathcal{C}(\theta_k^{\ell-i-1}(\Tilde{\by}_0)))\| \ge (\eta_3-4\eta_2)/(2k). \]

\noindent Now let $\bx^*= \theta_k^{\ell-i-1}(\bx_0)$. Then $\Tilde{\by}^*$
yields a violation of the boundary conditions for $\mathcal{C}$ by noting that
\begin{equation*}
\begin{split}
  \|\mathcal{C}(\tilde{\theta}_k(\Tilde{\by}^*)) - \tilde{\theta}_k(\mathcal{C}(\Tilde{\by}^*))\| &\geq (\eta_3-4\eta_2)/(2k) - \|\mathcal{C}(\tilde{\theta_k}(\Tilde{\by}^*)) - \mathcal{C}(\theta_k(\Tilde{\by}^*))\| - \|\tilde{\theta_k}(\mathcal{C}(\Tilde{\by}^*)) - \theta_k(\mathcal{C}(\Tilde{\by}^*))\|\\
  &\geq (\eta_3-4\eta_2)/(2k) - L(2\eta_1+\xi(\varepsilon,L,k)) - \xi(\varepsilon,L,k)
\end{split}
\end{equation*}
unless we find a violation of Lipschitz continuity. Indeed, note that
$(\eta_3-4\eta_2)/(2k) - L(2\eta_1+\xi(\varepsilon,L,k)) - \xi(\varepsilon,L,k) \geq \eta(\varepsilon,k)$.
\end{proof}

\setcounter{subsection}{1}
\subsection{$k$-Polygon Borsuk-Ulam and $k$-Polygon Tucker are $\ppak{k}[\#1]$-hard}

  In this section we show the $\ppak{k[\#1]}$-hardness of both \polygonTucker{k} and
\polygonBorsukUlam{k}. We start by observing that 
using the ideas from \cref{sec:kPolygonTuckerkPolygonBorsukUlam:equivalence} we
can prove that \polygonTucker{k} is reducible to \polygonBorsukUlam{k}. Then we show 
that \polygonTucker{k} is $\ppak{k[\#1]}$-hard. Finally the results of this section
together with the results of the previous 
\cref{sec:kPolygonTuckerkPolygonBorsukUlam:containment} imply that both 
\polygonTucker{k} and \polygonBorsukUlam{k} are $\ppak{k[\#1]}$-complete.

\begin{lemma} \label{lem:kPolygonTuckerTokPolygonBorsukUlam}
    The problem \polygonTucker{k} reduces to the problem \polygonBorsukUlam{k}.
\end{lemma}

\begin{proof}
  This Lemma follows from the proof of \cref{lem:polygonBorsukUlamImpliesPolygonTucker} combined with standard 
  arguments on how to construct an arithmetic circuit approximately computing the piecewise linear function $g: B^2 \to \bbR^2$ described in \cref{lem:polygonBorsukUlamImpliesPolygonTucker}. Recall that the function $g$ is constructed by interpolating within simplices of the triangulation using the labels given by $\mathcal{L}$. Given a precision parameter $\gamma$ and the Boolean circuit $\mathcal{L}$, we can construct an arithmetic circuit $\mathcal{C}$ that computes the function $g$ with error at most $\gamma$, using standard techniques \citep{EY10-Nash-FIXP, DaskalakisP11-CLS, GoldbergH19-HairyBall}. To be more precise, we can invoke Theorem E.2 from \citep{FearnleyGHS20} by observing that the function $g$ is polynomially approximately computable via standard numerical analysis techniques \cite{Brent1976}. Furthermore, the function computed by $\mathcal{C}$ will be Lipschitz-continuous with a Lipschitz constant $\tilde{L}$ that has bit-size polynomial in $m$ and $\log (1/\gamma)$ (where $m$ is the parameter of the triangulation used by $\mathcal{L}$).
  
  For any $\bx \in B^2$ that does not lie in a simplex that is a solution of $\mathcal{L}$, it must hold that $\|g(\bx)\| \geq 1/2$. This can easily be shown by taking into account that at most two labels (that are also adjacent) are used for the interpolation in that case. As a result, as long as we ensure that $1/2-\gamma > \varepsilon$, any point with $\|\mathcal{C}(\bx)\| \leq \varepsilon$ will yield a solution of the \polygonTucker{k} instance.
  
  It remains to show that no bogus boundary condition violations occur. Note that unless the simplex containing $\bx \in \partial B^2$ yields a boundary condition violation for $\mathcal{L}$, it must hold that $g(\theta_k(\bx)) = \theta_k(g(\bx))$ (as shown in \cref{lem:polygonBorsukUlamImpliesPolygonTucker}). Thus, it follows that
  \begin{equation*}
  \begin{split}
      \|\mathcal{C}(\tilde{\theta}_k(\bx)) - \tilde{\theta}_k (\mathcal{C}(\bx))\| &\leq \|\mathcal{C}(\theta_k(\bx)) - \theta_k(\mathcal{C}(\bx))\| + \|\mathcal{C}(\tilde{\theta}_k(\bx)) - \mathcal{C}(\theta_k(\bx))\| + \|\tilde{\theta}_k (\mathcal{C}(\bx)) - \theta_k(\mathcal{C}(\bx))\|\\
      &\leq 2 \gamma + \tilde{L} \cdot \xi(\varepsilon,L,k) + \xi(\varepsilon,L,k).
  \end{split}
  \end{equation*}
  Note that we want this quantity to be strictly less than $\eta(\varepsilon,k)$. Thus, we pick $\varepsilon=1/4$ and then set $\gamma$ and $L$ so that
\begin{itemize}
      \item $L \geq \tilde{L}$,
      \item $\gamma < 1/2 - \varepsilon$,
      \item $2 \gamma + \tilde{L} \cdot \xi(\varepsilon,L,k) + \xi(\varepsilon,L,k) < \eta(\varepsilon,k)$,
\end{itemize}
which can easily be achieved by picking $\gamma$ sufficiently small and $L$ sufficiently large. Note that since $\mathcal{C}$ is $\tilde{L}$-Lipschitz-continuous, it will also be $L$-Lipschitz-continuous.
\end{proof}

Now we present our main proposition in this section.
  
\begin{proposition}\label{prop:k-polygon-hardness}
For all $k \geq 3$, \polygonTucker{k}~is $\ppak{k[\#1]}$-\textup{hard}. 
\end{proposition}

Recall that $\ppak{k[\#1]} = \cap_{p \in PF(k)} \ppak{p}$, where $PF(k)$ denotes the set of prime factors of $k$.

\begin{proof}

To prove this hardness result we reduce from \textsc{Bipartite-mod-$k[\#1]$} (see \cref{sec:def-ppa-k} for the formal definition). Recall that in this problem we are given a bipartite graph on the set of nodes $A \cup B$, where $A = \{0\} \times \{0,1\}^n$ and $B = \{1\} \times \{0,1\}^n$, such that the node $00^n \in A$ has degree $1$. The goal is to find any other node that has degree $\neq 0 \mod k$. All nodes have degree in $\{0,1,\dots,k\}$ and we can assume (see e.g., \citep[Section 4.2]{Hollender19}) that the circuit $C$ which implicitly represents the bipartite graph is consistent, i.e., for all $x,y$ we have $y \in C(x)$ iff $x \in C(y)$.
	
Consider the regular $k$-polygon $W_k$ in $\mathbb{R}^2$ with the edge parallel triangulation (\cref{def:edge-par-triang}). For any $i \in \mathbb{Z}_k$ let $R(i,i+1)$ denote the triangle with endpoints $\{\boldsymbol{0},\bu_i,\bu_{i+1}\}$. In this proof we refer to \emph{nodes} of the \textsc{Bipartite-mod-$k[\#1]$} instance and to \emph{vertices} of the triangulation of $W_k$.

To every node $x \in (A \setminus \{00^n\}) \cup B$ we associate a distinct interval on the outer boundary of $R(1,2)$, i.e., on the edge $\chull(\{\bu_1,\bu_2\})$. This interval, which we denote by $K_1(x)$, is picked such that it covers $4k$ vertices of the triangulation. It is easy to see that we can pick the triangulation to be fine enough so that there are indeed enough vertices on $\chull(\{\bu_1,\bu_2\})$ to associate a distinct interval of $4k$ vertices to each node $x \in (A \setminus \{00^n\}) \cup B$. Since the triangulation is symmetric on the boundary with respect to $\theta_k$, we can immediately also extend this association to the outer boundary of $R(i,i+1)$ for all other $i$. In other words, for any $i \in \mathbb{Z}_k$ and any $x \in (A \setminus \{00^n\}) \cup B$, we let $K_i(x) = (\theta_k)^{i-1} (K_1(x))$, which is an interval of $4k$ vertices on $\chull(\{\bu_i,\bu_{i+1}\})$. Thus, for any $x$ there are $k$ distinct intervals on the boundary associated to it, one in each of $\chull(\{\bu_i,\bu_{i+1}\})$, $i \in \mathbb{Z}_k$.

In the rest of this proof we explain how to assign a label in $\mathbb{Z}_k$ to every vertex of the triangulation so that the \polygonTucker{k} boundary conditions are satisfied and any solution (i.e., a trichromatic triangle or an edge with non-consecutive labels) must contain a vertex that lies in some interval $K_i(x)$ where $x$ is a solution-node of \textsc{Bipartite-mod-$k[\#1]$} (i.e., a node with degree $\neq 0 \mod k$). This will ensure that from any solution of the \polygonTucker{k} instance we can easily obtain a solution-node of the \textsc{Bipartite-mod-$k[\#1]$} instance.

We begin by defining the ``environment'' label for every vertex of the triangulation. This corresponds to the standard label that the vertex will have, unless we specify it otherwise in the construction. Any vertex lying in $R(i,i+1) \setminus \chull(\{\boldsymbol{0},\bu_{i+1}\})$ has the environment label $i$. Next, we define ``cables'', which will be used to embed the edges of the \textsc{Bipartite-mod-$k[\#1]$} instance in our construction.

\paragraph{Wires and Cables.}
A ``wire'' has an associated label $i \in \mathbb{Z}_k$ and it simply consists of a path of vertices in the triangulation such that all vertices on the path have the label $i$. A ``cable'' is made out of $k-1$ wires, where each wire has a distinct associated label. The wires are arranged in parallel inside the cable so that only wires with consecutive labels are allowed to touch. More precisely, if the cable uses labels $\mathbb{Z}_k \setminus \{i\}$, then the wires are arranged according to their labels in the order $i+1,i+2,\dots,i+(k-1)$ from right to left, in the forward direction of the cable. Note that while wires are not directed, we can define a direction for every cable, based on the order of the wires inside the cable. A cable using the labels $\mathbb{Z}_k \setminus \{i\}$ is only allowed to exist inside a region with environment label $i$. This ensures that any vertex that is adjacent to either side of the cable, and thus to the wire labeled $i+1$ or $i+(k-1)$, is labeled $i$ and does not introduce a solution. The construction of the cables ensures that the wires are ``isolated'' from each other and from the environment, in the sense that no solution is introduced along the cable. However, if the start or end of a cable using the labels $\mathbb{Z}_k \setminus \{i\}$ is allowed to ``touch'' the environment label $i$, then this will necessarily yield a trichromatic triangle at that point. See \cref{fig:cable} for an illustration of how a cable is constructed.

\begin{figure}
    \centering
    \includegraphics{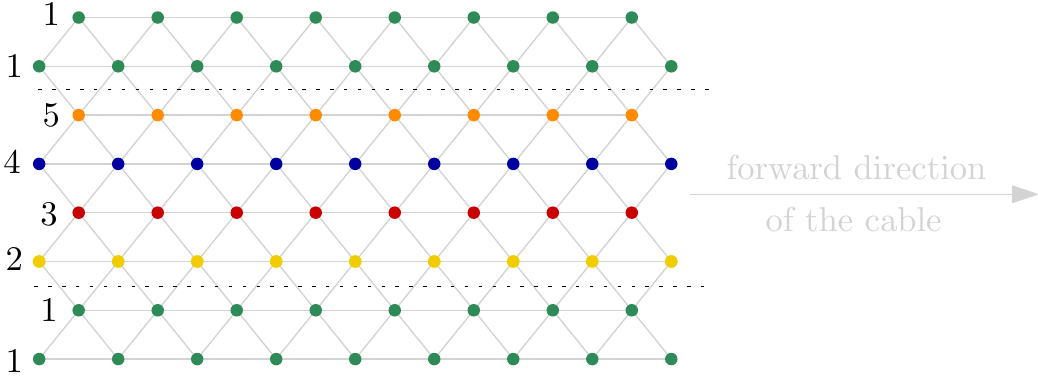}
    \caption{A cable for the case $k=5$. The cable uses the labels $\mathbb{Z}_5 \setminus \{1\}$ and the environment has label $1$. The labels are color-coded as indicated on the left-hand side, i.e., green is label $1$, yellow is label $2$ etc. The forward direction of the cable is indicated by an arrow on the right-hand side. Note that the portion of the cable shown in the figure does not introduce any solution of \polygonTucker{5}.}
    \label{fig:cable}
\end{figure}

It is easy to see that a cable can turn without introducing solutions. Next, let us see how a cable can transition from one environment to another. Consider a cable in environment $i \in \mathbb{Z}_k$, i.e., it uses the labels $\mathbb{Z}_k \setminus \{i\}$ and lies in $R(i,i+1)$. If the cable arrives on the boundary $\chull(\{\boldsymbol{0},\bu_{i+1}\})$ of $R(i,i+1)$, we can transform it into a cable that uses the labels $\mathbb{Z}_k \setminus \{i+1\}$ and continues into $R(i+1,i+2)$, i.e., in the environment $i+1$, on the other side of $\chull(\{\boldsymbol{0},\bu_{i+1}\})$. Importantly, the direction of the cable does not change and we do not introduce any new solutions. This transformation of the cable is shown in \cref{fig:cable-trafo}. The idea is simple. Consider a cable in environment $i$ that arrives on $\chull(\{\boldsymbol{0},\bu_{i+1}\})$ moving forward. The wires are arranged according to their labels in the order $i+1,i+2,\dots,i+(k-1)$ from right to left, in the forward direction of the cable. When the cable reaches $\chull(\{\boldsymbol{0},\bu_{i+1}\})$, the right-most wire (which is labeled $i+1$) is dropped from the cable, i.e., it merges into the environment $i+1$ of $R(i+1,i+2)$. On the other side of the cable, a new wire with label $i = i+1+(k-1)$ is created by using the environment $i$ of $R(i,i+1)$. Thus, we obtain a cable with wires $i+2,\dots,i+1+(k-1)$ (from right to left) in the environment $i+1$, as desired. It is easy to see that this construction does not introduce any new solutions, because we have ensured that non-consecutive labels do not ``touch''.

\begin{figure}
    \centering
    \includegraphics{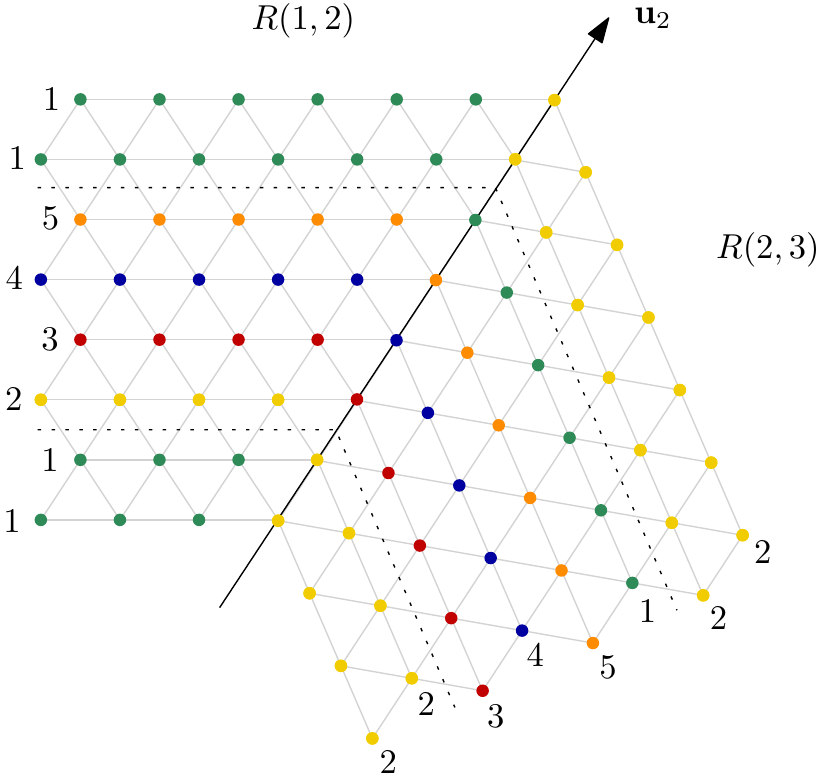}
    \caption{Transformation of a cable for the case $k=5$. In the region $R(1,2)$ the cable uses labels $\mathbb{Z}_5 \setminus \{1\}$ and has environment $1$. When the cable reaches region $R(2,3)$, the construction shown in the figure ensures that from now on, the cable uses labels $\mathbb{Z}_5 \setminus \{2\}$ and has environment $2$. The cable uses the labels $\mathbb{Z}_5 \setminus \{2\}$ and the environment has label $2$. Note that the transformation of the cable as shown in the figure does not introduce any solution of \polygonTucker{5}.}
    \label{fig:cable-trafo}
\end{figure}

When a cable starts or ends on the outer boundary of $R(i,i+1)$, i.e., on $\chull(\{\bu_{i},\bu_{i+1}\})$, the \polygonTucker{k} boundary conditions force a cable to start or end at the corresponding position in each of the regions $R(j,j+1)$, $j \in \mathbb{Z}_k \setminus \{i\}$. These $k-1$ cables will have the same direction as the original cable. In other words, if the cable in region $R(i,i+1)$ ends on the outer boundary, then the $k-1$ corresponding cables will also end on their corresponding boundary. Similarly, if the original cable starts on the boundary, then the $k-1$ corresponding cables will start on their corresponding boundary.

\paragraph{Construction of the instance.}
Before we begin, let us introduce the following useful terminology.
Every node $x$ of the \textsc{Bipartite-mod-$k[\#1]$} instance has some number $\ell \in \{0,1,\dots, k\}$ of neighbors, as given by $C(x)$. For any $i \in [\ell]$, we define the ``$i$th neighbor of $x$'' to be the $i$th node in the lexicographically ordered list of neighbors of $x$.

We are now ready to begin describing the instance we construct.
Recall that we have defined an environment label for every vertex of the triangulation except $\boldsymbol{0}$. Now consider the labeling where every vertex is simply labeled by its environment label. Clearly, this labeling satisfies the \polygonTucker{k} boundary conditions. Furthermore, no matter how we pick the label of $\boldsymbol{0}$, there will be a solution there, since $\boldsymbol{0}$ is adjacent to all environments. Now it is easy to see that we can ``move'' this solution by locally modifying some of the labels. Namely, instead of having all labels meet at $\boldsymbol{0}$, we can instead construct a cable that uses the labels $\mathbb{Z}_k \setminus \{1\}$ and moves into the region with environment label $1$. As a result, there is no longer a solution at $\boldsymbol{0}$, but instead there is now a solution at the end of the cable. In more detail, every environment label except $1$ yields a wire with the corresponding label and the wires are arranged into a cable that uses the labels $\mathbb{Z}_k \setminus \{1\}$. \cref{fig:origin-cable} illustrates this construction in the case $k=5$.

\begin{figure}[ht]
    \centering
    \includegraphics{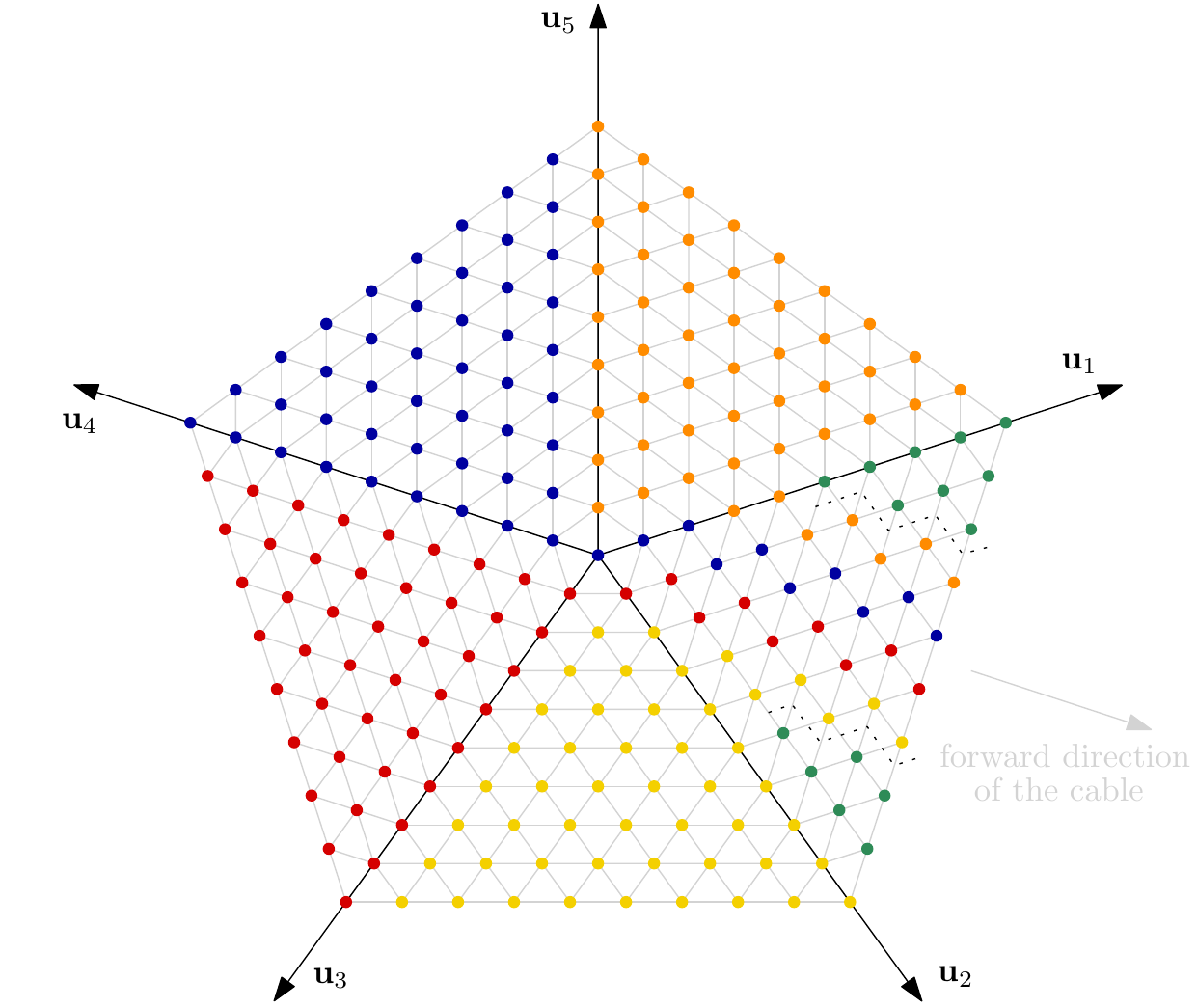}
    \caption{Construction around the origin for the case $k=5$. Even though all the environments ``meet'' at the origin, the construction shown in the figure ensures that there is no solution around the origin, but instead a cable is created. The labels are color-coded as in \cref{fig:cable} and \cref{fig:cable-trafo}.}
    \label{fig:origin-cable}
\end{figure}

Recall that the node $00^n \in A$ has exactly one neighbor $y \in B$. Let $j \in [k]$ be the number such that $00^n$ is the $j$th neighbor of y. The cable that we just constructed at the center of the instance will be routed so that it ends at $K_j(y)$ (a segment on the outer boundary of $R(j,j+1)$, as defined above). Furthermore, for any $x \in A \setminus \{00^n\}$ and $y \in B$ such that $x$ and $y$ are neighbors, there will be a cable starting at $K_i(x)$ and ending at $K_j(y)$, where $i,j \in [k]$ are such that $x$ is the $j$th neighbor of $y$, and $y$ is the $i$th neighbor of $x$. This ensures that the following properties hold.

Consider a node $x \in A \setminus \{00^n\}$ that has $\ell$ neighbors, where $\ell \in \{0,1,\dots,k\}$.
If $\ell = 0$, i.e., $x$ is an isolated node, then for all $i \in [k]$ there is no cable at $K_i(x)$. In particular, there is no \polygonTucker{k} solution in any $K_i(x)$.
If $\ell \in [k]$, then for all $i \in [\ell]$ there is a cable starting at $K_i(x)$ (and ending at some $K_j(y)$). For all $i \in [k] \setminus [\ell]$, there is a start of a cable at $K_i(x)$, but the cable just stops immediately and does not go anywhere. This ensures that the \polygonTucker{k} boundary conditions are satisfied, but also that there is a \polygonTucker{k} solution at $K_i(x)$ for all $i \in [k] \setminus [\ell]$ (because of an exposed end of cable). Thus, we obtain that
\begin{itemize}
	\item if $\ell \in [k-1]$, there is a solution at $K_i(x)$ for some $i \in [k]$,
	\item if $\ell \in \{0,k\}$, then there is no solution in any $K_i(x)$.
\end{itemize}

Similarly, consider a node $y \in B$ that has $\ell$ neighbors, where $\ell \in \{0,1,\dots,k\}$.
If $\ell = 0$, i.e., $y$ is an isolated node, then for all $j \in [k]$ there is no cable at $K_j(y)$. In particular, there is no \polygonTucker{k} solution in any $K_j(y)$.
If $\ell \in [k]$, then for all $j \in [\ell]$ there is a cable ending at $K_j(y)$ (that started at some $K_i(x)$). For all $j \in [k] \setminus [\ell]$, there is an end of a cable at $K_j(y)$, but the cable just started there and does not come from anywhere else. This ensures that the \polygonTucker{k} boundary conditions are satisfied, but also that there is a \polygonTucker{k} solution at $K_j(y)$ for all $j \in [k] \setminus [\ell]$ (because of an exposed start of cable). Thus, we obtain that
\begin{itemize}
	\item if $\ell \in [k-1]$, there is a solution at $K_j(y)$ for some $j \in [k]$,
	\item if $\ell \in \{0,k\}$, then there is no solution in any $K_j(y)$.
\end{itemize}

As a result, if the cables can indeed be constructed to connect the various $K_i(x)$ and $K_j(y)$ as desired, then any \polygonTucker{k} solution of the instance will have to be next to some $K_i(x)$ such that $x$ is a solution-node or next to some $K_j(y)$ such that $y$ is a solution node. One immediate obstacle to routing the cables as desired is that we are working in two dimensions and it is very likely that cables will have to cross each other. Fortunately, there is a simple ``trick'' that has been used in prior work to resolve this issue \citep{CD09-2DBrouwer}. Consider the following idea: cut the two cables that want to cross each other at the point of crossing. This creates two ends of cables and two starts of cables. It is easy to see that we can connect an end of cable with a start of cable, and the other end with the other start, so that no crossing occurs anymore. This modification of the cables is completely local and does not have any impact on the rest of the instance.

\paragraph{Constructing the labeling function.}
Since we want to be able to construct a circuit for the labeling function, we need to be a bit more precise about the path followed by every cable. One way to achieve this is to reserve a separate ``circular lane'' for each pair $(y,j)$, where $y \in B$ and $j \in [k]$. A circular lane is a path of sufficient width that simply stays parallel to the outer boundary in each region $R(i,i+1)$ and thus makes a full ``circle'' around the center of the domain. By picking a fine enough triangulation, we can ensure that there is a separate, disjoint circular lane $L_{y,j}$ for each pair $(y,j)$, where $y \in B$ and $j \in [k]$. Then the cable going from some $K_i(x)$ to some $K_j(y)$ will be routed as follows. Starting at $K_i(x)$, move perpendicularly to the outer boundary towards the inside of the domain, until the circular lane $L_{y,j}$ is reached. Next, follow the lane in clockwise direction. When following the lane, the cable might have to transition from some environment to the next and this is implemented as described earlier. The cable stops following the lane when it reaches the part of the lane lying just ``above'' $K_j(y)$. In other words, the cable stops following the lane when it is at the point where it can just turn left and move straight towards the boundary to end up at $K_j(y)$, as desired. Clearly, we can pick the triangulation fine enough so that this routing is indeed well defined.

This construction has two advantages. First of all, it ensures that any crossing involves at most two cables, not more. We can then use the trick described above to locally resolve these crossings. Furthermore, it ensures that we can construct a Boolean circuit computing the labeling function. Indeed, given any vertex of the triangulation we can easily determine on which circular lane and on which path perpendicular to the outer boundary it lies. This gives us enough information to then use the \textsc{Bipartite-mod-$k[\#1]$} circuit $C$ as a sub-routine to determine whether the vertex lies on a cable and if so, how exactly the cable behaves locally (including possible crossing-avoiding trick). The circuit $C$ only needs to be queried a constant number of times and thus the resulting circuit for the labeling will have polynomial size with respect to the size of $C$ and $n$. We omit the full details, since this part of the proof is essentially the same as in prior work (see e.g., \citep{CD09-2DBrouwer}).
\end{proof}

  Based on the results presented in this and the previous section we get the  
following theorem.

\begin{theorem} \label{thm:kPolygon:completeness}
    The problems \polygonTucker{k} and \polygonBorsukUlam{k} are both $\ppak{k}[\#1]$-complete.
\end{theorem}
In particular, if $k=p^r$ is a prime power, then \polygonTucker{k} and \polygonBorsukUlam{k} are \ppak{p}-complete.

\section{The BSS Theorem is \ppak{p}-complete}
\label{sec:BSS-BSStucker}

The main result of this section is the \ppak{p}-completeness of the computational problem associated with the BSS Theorem. We refer to this problem as \textsc{$p$-BSS} and we define it formally in \cref{sec:compBSS}. We state the main theorem of the section below.

\begin{theorem}
\label{thm:tucker-completeness}
For every prime $p$, the problem \textsc{$p$-BSS} is \ppak{p}-complete.
\end{theorem}

\noindent In order to prove the theorem, first we show that the BSS theorem is equivalent to a generalization of Tucker's Lemma, which we call \tucker, and then we define the corresponding computational problems and show that they are equivalent. Then, we show the \ppak{p}-completeness of \tucker, which then implies \cref{thm:tucker-completeness}. We show the \ppak{p}-hardness via a reduction from the \polygonTucker{p} problem, proven to be \ppak{p}-complete in the previous section. The membership in \ppak{p} is proven in \cref{sec:pstartucker}, where we show that \tucker\ reduces to another variant of Tucker's lemma, which we call \pstartucker{p} (\cref{lem:bss-to-pstar}), and for which we prove membership in \ppak{p}.

\subsection{The BSS Theorem and Equivalent Formulations}
\label{sec:bss}

Our notation follows that of \citet{BSS81}. 
Let $p \geq 2$ prime, $n \in \mathbb{N}$ and let $P = \{(\bv_1,\bv_2, \dots, \bv_p) \in (\mathbb{R}^n)^p : \sum_{i=1}^p \bv_i = \boldsymbol{0}\}$ and $P_2 = \{\bu \in P : \|\bu\|_2 \leq 1\}$.
It holds that

\[P \cong \mathbb{R}^{n(p-1)} \text{ and } P_2 \cong  B^{n(p-1)}. \] 

Note that $P$ is a hyperplane of $\bbR^{np}$ with dimension $n(p-1)$. Let $\mathcal{B}$ be an orthogonal basis of $P$ in $\bbR^{np}$ and let $\phi:P_2 \rightarrow B^{n(p-1)}$ be the function that maps $\bx \in P_2$ to its coefficients in base $\mathcal{B}$. Then, $\phi$ is a homeomorphism of $P_2$ and $B^{n(p-1)}$. Notice that $\phi(\partial P_2) = S^{n(p-1)-1}$ and that $\phi$ and $\phi^{-1}$ are efficiently computable. 
The same mapping shows that  $P$ and $\mathbb{R}^{n(p-1)}$ are homeomorphic.

Let 
$$\theta(\bv_1, \dots, \bv_p) = (\bv_p,\bv_1, \dots, \bv_{p-2},\bv_{p-1}).$$

The function $\theta$ has order $p$; namely the composition by itself $p$ times is equal to the identity. Also,  for all $i \in \{1, \dots, p-1\}$ and all $\bx \in P \setminus \{\boldsymbol{0}\}$, $\theta^{i}(\bx) \neq \bx$. In other words,  for any $i < p$, $\theta^{i}$ restricted to $P \setminus \{\boldsymbol{0}\}$ has no fixed points. Hence, $\theta$ acts freely on $P \setminus \{\boldsymbol{0}\}$. It follows with a similar argument that the function $\theta$ acts freely also on $P_2 \setminus \{\boldsymbol{0}\}$ and on $\partial P_2$.

\medskip

The original statement of the BSS Theorem requires the notion of CW-complexes, but since in our work we do not use CW-complexes, we will not define them formally. Intuitively, a CW-complex consists of building blocks that can be topologically glued together.

Let $p$ be a prime, $n \geq 1$ and $X$ be a CW-complex consisting of $p$ copies of the $n(p-1)$-dimensional ball glued on their boundaries.

Let $\alpha: \mathbb{R}^{n(p-1)} \rightarrow \mathbb{R}^{n(p-1)}$ such that $\alpha = \phi \circ \theta \circ \phi^{-1}$. Note that $\alpha$ acts freely on $\mathbb{R}^{n(p-1)} \setminus \{\boldsymbol{0}\}$, $B^{n(p-1)} \setminus \{\boldsymbol{0}\}$ and $S^{n(p-1)-1}$. Let $\omega$ be the extension of $\alpha$ on $X$ defined as follows:
\[\omega(\by,r,q) = (\alpha \by, r, q+1 \pmod p),\]
where $(\by,r,q)$ denotes the point of the $q$-th ball with radius $r$ and direction $\by \in S^{n(p-1)-1}$.

The map $\omega$ is a free action on $X$ and the following theorem holds:
\begin{theorem}[BSS Theorem, \citep*{BSS81}]\label{thm:BSStheorem}
\label{thm:BSS}
For the mapping $\omega$ and any continuous map $h:X \rightarrow \bbR^n$,
there exists an $\bx \in X$ such that $h(\bx) = h(\omega \bx) = \dots = h(\omega^{p-1}\bx) $.
\end{theorem}

The following equivalent formulations of \cref{thm:BSS} are useful for defining the computational problem related to BSS.

\begin{theorem}[BSS Theorem, equivalent formulations]\label{thm:BSSalt}
The following statements are equivalent to the BSS Theorem:
\begin{enumerate}
    \item Let $g: P_2 \to P$ be continuous and such that $g(\theta \bx) = \theta g(\bx)$ for all $\bx \in \partial P_2$. Then, there exists $\bx \in P_2$ such that $g(\bx)=0$.
    \item Let $g: B^{n(p-1)} \to \mathbb{R}^{n(p-1)}$ be continuous and such that $g(\alpha \bx) = \alpha g(\bx)$ for all $\bx \in S^{n(p-1)-1}$. Then, there exists $\bx \in B^{n(p-1)}$ such that $g(\bx)=0$.
\end{enumerate}

\end{theorem}

\begin{proof}
We first show that the two statements are equivalent and then that statement (2) is equivalent to \cref{thm:BSS}.
\paragraph{(1) $\Longleftrightarrow$ (2)} The equivalence follows from  $P_2 \cong B^{n(p-1)}$ and $P \cong \mathbb{R}^{n(p-1)}$, as well as from the equivariance of $\phi$.

\paragraph{(2) $\Longleftrightarrow$ (BSS)} 
We first show that the BSS Theorem implies (2). Recall that the CW-complex $X$ consists of $p$ copies of $B^{n(p-1)}$ with their boundaries ``glued'' together. 
Define $h: X \to \mathbb{R}^n$ as follows: for $\bx = (\by,r,i) \in X$, let $h(\bx) = [\phi^{-1} \circ g \circ \alpha^{1-i} (r\by)]_{2-i} \in \mathbb{R}^n$, where for $j \in \mathbb{Z}_p \equiv [p]$, $[\cdot]_j$ denotes the $j$-th component of an element in $(\mathbb{R}^n)^p$ (i.e., $[(\bv_1, \dots, \bv_p)]_j = \bv_j$). The mapping $h$ is well-defined and continuous on the glued boundary, because for all $i$ and $\by \in S^{n(p-1)-1}$, we have $h(\by,1,i) = [\phi^{-1} \circ g \circ \alpha^{1-i}(\by)]_{2-i} = [\theta^{1-i} \phi^{-1} \circ g (\by)]_{2-i} = [\phi^{-1} \circ g (\by)]_1$ (which does not depend on $i$). Finally, note that any $\bx = (\by,r,1) \in X$ with $h(\bx)=h(\omega \bx)= \dots = h(\omega^{p-1}\bx)$ yields $\bz = r\by \in B^{n(p-1)}$ with $[\phi^{-1} \circ g(\bz)]_1=\dots=[\phi^{-1} \circ g(\bz)]_p$, which implies $\phi^{-1} \circ g(\bz) = 0$ (by definition of $P$), and thus $g(\bz)=0$.

Conversely, (2) implies BSS.
We identify $B^{n(p-1)}$ with the first ball in the CW-complex $X$. Given a continuous function $h:X \to \mathbb{R}^n$, define $g: B^{n(p-1)} \to \mathbb{R}^{n(p-1)}$ by $g(\bx)= \phi(h(\omega^p\bx)-h(\omega^{p-1} \bx), h(\omega^{p-1} \bx) - h(\omega^{p-2} \bx), \dots, h(\omega \bx)-h(\bx))$. It is easy to check that $g(\alpha \bx) = \alpha g(\bx)$ for all $\bx \in S^{n(p-1)-1}$ by noting that $\omega \bx = \alpha \bx$ for such $\bx$.
\end{proof}

\subsection{The \tucker\ Lemma}

In this section, we define a generalization of Tucker's Lemma, that we call the \tucker\ Lemma and show that it is equivalent to the \bss\ Theorem.
The \tucker\ Lemma applies to triangulations of $P_2$ that have some special properties. 

For $j \in [n]$, let $\be_j \in \mathbb{R}^n$ be the $j$-th unit vector. For $(i,j) \in \bbZ_p \times [n]$, let
$$\be^{i,j}=\frac{1}{2(p-1)}(-\be_j, \dots, -\be_j, (p-1)\be_j, -\be_j, \dots, -\be_j) \in P$$
where the term $(p-1)\be_j$ is in the $i$-th position. For each $\bs = (s_1,\dots, s_n) \in [p]^n$, we define the simplex $\sigma_{\bs} = \{\boldsymbol{0}\} \cup \{\be^{i,j} \text{ s.t. for each } j \in [n], i \in \bbZ_p\setminus\{s_j\}\}$. Note that $|V(\sigma_{\bs})| = (p-1)n + 1 $.

\begin{lemma}
 $T^* = \{\tau: \tau \subseteq \sigma_{\bs} \text{ with } \bs \in [p]^n\}$ is a triangulation  of $P_2$.
\end{lemma}

\begin{proof}
Let $C = \chull\left((\be^{i,j})_{(i,j) \in \bbZ_p \times [n]}\right)$. Then, by definition of $T^*$, $C \cong \norm{T^*}$ . So, the lemma follows from the fact that $C \cong P_2$.
\end{proof}

\begin{definition}
We say that a triangulation $T$ of $P_2$ is \emph{nice} if it satisfies the following two conditions:
\begin{itemize}
    \item If $\sigma \in T \cap \partial P_2$ then $\theta \sigma \in T$.
    \item $T$ refines the triangulation $T^*$ (that is, for each $\sigma \in T$ there
is $\tau \in T^*$ with $\sigma \subseteq \tau$ ).
\end{itemize}
\end{definition}

\begin{theorem}[\tucker\ Lemma]
\label{thm:bss-tucker}
Let $T$ be a nice triangulation of $P_2$. Let $\lambda = (\lambda_1, \lambda_2) : V(T) \to \mathbb{Z}_p \times [n]$ be a labeling such that for all $\bx \in \partial T$, $\lambda(\theta \bx) = (\lambda_1(\bx)+1, \lambda_2(\bx))$. Then, there exists a $(p-1)$-simplex $\sigma$ of $T$ such that $\lambda(\sigma) = \mathbb{Z}_p \times \{j\}$ for some $j \in [n]$, where $\lambda(\sigma) = \{\lambda(\bx) : \bx \in V(\sigma)\}$.
\end{theorem}

\begin{lemma}
\label{lm:BBSimpliestucker}
    The BSS Theorem (\cref{thm:BSSalt}) implies
  the \tucker~Lemma (\cref{thm:bss-tucker}).
\end{lemma}

\begin{proof}
We interpret each label $(i,j) \in \mathbb{Z}_p \times [n]$ as the vector $\be^{i,j}$ and we set $g$ to be the extension of $\lambda$ to a piecewise linear function on $P_2$. Notice that $g(\theta \bx) = \theta g(\bx)$ for $\bx \in \partial P_2$.  Then, it follows from \cref{thm:BSSalt} that there exists an $\bx \in P_2$ such that $g(\bx) = 0$. The lemma follows by noting that any convex combination of different vectors $\be^{i,j}$ that equals $\boldsymbol{0}$ must contain $p$ vectors $\be^{i_1,j_1}, \dots, \be^{i_p,j_p}$ such that $\{i_k, j_k\}_{k \in [p]} = \mathbb{Z}_p \times \{j\}$. Hence, the point $\bx$ lies in a simplex with a $(p-1)$-dimensional face $\sigma$ such that $\lambda(\sigma) = \mathbb{Z}_p \times \{j\}$ for some $j \in [n]$. 
\end{proof}

\begin{lemma}
\label{lm:tuckerImpliesBSS}
    The \tucker~Lemma (\cref{thm:bss-tucker}) implies
  the BSS Theorem (\cref{thm:BSSalt}).
\end{lemma}

\begin{proof}

We show that \tucker~implies Statement (1) of \cref{thm:BSSalt}.
Using standard arguments, it suffices to show that for every $\varepsilon > 0$ we can find a point $\bx$ such that $\norm{g(\bx)}_{\infty} \leq \varepsilon$. 
Since $g:P_2 \rightarrow P$ is a continuous function in a compact set, it is also uniformly continuous. Thus, for every $\varepsilon > 0$, there exists $\delta$ such that if $\norm{\bx - \bx'}_2 < \delta$, then $\norm{g(\bx) - g(\bx')}_{\infty} < \varepsilon/n$. 
Assume that $T$ is a nice triangulation of $P_2$ with diameter at most $\delta$.

For any point $\bx \in V(T)$, the label $\lambda(\bx)=(i^*,j^*)$ is defined as follows. For $(i,j) \in \bbZ \times [n]$, let $[g(\bx)]_{i,j}$ denote the $(i,j)$-coordinate of $g(\bx) \in P$. First, consider the case where $g(\bx) \neq \boldsymbol{0}$. Then, pick $j^* = \argmax_{j \in [n]} \max_{i \in [p]} [g(\bx)]_{i,j}$. Break ties by picking the smallest such $j$. Let $S = \{i : [g(\bx)]_{i,j^*} = \max_\ell [g(\bx)]_{\ell,j^*}\}$ and pick $i^* = T_p(S)$, where $T_p$ is the $\bbZ_p$-equivariant tie-breaking function of \cref{def:tie-breaking}. Note that $S \notin \{\emptyset, [p]\}$, because $g(\bx) \neq \boldsymbol{0}$ and by the choice of $j^*$. If $g(\bx) = \boldsymbol{0}$, then if $\bx = \boldsymbol{0}$, we assign an arbitrary label, and if $\bx \neq \boldsymbol{0}$, we use the same procedure as above but with $\bx$ instead of $g(\bx)$.

With this definition of the labeling, it is easy to check that for any $\bx \in \partial T$, we always have $\lambda(\theta \bx) = (\lambda_1(\bx)+1, \lambda_2(x))$, since $g(\theta \bx) = \theta g(\bx)$. Hence, by \cref{thm:bss-tucker} there exists a $\sigma = \{\bx_1, \dots, \bx_p\} \in T$ and a $j^* \in [n]$ such that $\lambda(\bx_i) = (i,j^*)$ for all $i \in [p]$.  

The lemma follows by showing that there exists $i \in [p]$ such that $\norm{g(\bx_i)}_{\infty} \leq \varepsilon$. First of all, note that for any $\bx \in P$, by definition we have that $\norm{g(\bx)}_\infty \leq n \cdot \max_{i,j} [g(\bx)]_{i,j}$. Now, assume for the sake of contradiction that $\norm{g(\bx_i)}_\infty > \varepsilon$ for all $i \in [p]$. Then, it follows that $[g(\bx_i)]_{i,j^*} > \varepsilon/n$ for all $i \in [p]$. But by the choice of diameter for the triangulation and uniform continuity of $g$, this implies that $[g(\bx_1)]_{i,j^*} > 0$ for all $i \in [p]$, which contradicts $\bx_1 \in P$.
\end{proof}

\subsection{Computational problems: $p$-\tucker~and $p$-\bss}\label{sec:compBSS}

Motivated by \cref{thm:BSSalt} and \cref{thm:bss-tucker}, we define the computational problems corresponding to the BSS Theorem and the \tucker\ Lemma. Combining the proof ideas in \cref{lm:BBSimpliestucker} and \cref{lm:tuckerImpliesBSS}, together with some efficiency requirements of the triangulation, as per \cref{def:indexValueCircuits}, we show that the two computational problems are polynomially equivalent.
 
As before, we assume that the input functions are represented as arithmetic circuits with operations $\times \zeta$, $+$, $-$, $<$, $\min$, and $\max$ and rational constants. Similarly to our definition of \polygonBorsukUlam{k}, and for the same reasons, we will add a solution type to ensure that it is Lipschitz-continuous.

\medskip

The computational analogue of the BSS theorem is based on statement (2) of \cref{thm:BSSalt}. It takes as input an arithmetic circuit, which evaluates the function $g: B^{n(p-1)} \rightarrow \bbR^{n(p-1)}$. In order to deal with the polynomial-size representation of any of the solutions we allow as a solution any approximate violation of the equivariance condition of the BSS theorem, as we did with \polygonBorsukUlam{p} in Section \ref{sec:2dBSS}.

\medskip

\noindent\fbox{%
\colorbox{gray!10!white}{
    \parbox{\textwidth}{%
\noindent\textbf{\underline{$p$-BSS:}} 

\smallskip

\noindent \textsc{Input:} An integer $n \geq 1$, an accuracy parameter $\varepsilon > 0$, 
a Lipschitz constant $L$, 
and an arithmetic circuit $\mathcal{C}$.
\smallskip

\noindent \textsc{Output:} 
\vspace{-0.07in}
\begin{enumerate}
    \item A point $\bx \in S^{n(p-1)-1}$ such that $\norm{\mathcal{C}(\alpha \bx) - \alpha \mathcal{C}(\bx)} > \eta(\varepsilon, p) := \varepsilon/8 p^4$
    \vspace{-0.07in}
    \item Two points $\bx, \by \in B^{n(p-1)}$ such that $\norm{\mathcal{C}(\bx) - \mathcal{C}(\by) } > L \norm{\bx - \by}$
    \vspace{-0.07in}
    \item A point $\bx^* \in B^{n(p-1)}$ such that $\norm{\mathcal{C}(\bx^*)}_{\infty} \leq \varepsilon$
\end{enumerate}

\noindent A valid output of the $p$-\bss~problem is either a point that violates the boundary condition of \cref{thm:BSSalt}, two points that violate the Lipschitz-continuity or an approximate root of the function $g$.}}}

\medskip

The computational analogue of the \tucker\ Lemma is based on \cref{thm:bss-tucker} and is parameterized by a ``triangulation scheme'' $\mathcal{T}$. Namely, given an $m \in \mathbb{N}$, $\mathcal{T}(m)$ yields a nice triangulation $T$ with diameter at most $1/2^m$. The triangulation $T$ is given through two arithmetic circuits, $\idx$ and $\val$ (see \cref{def:indexValueCircuits}), that have size polynomial in $n$ and $m$. Assume that for $m =0$, $\mathcal{T}$ yields the $\idx^*$ and $\val^*$ circuits of $T^*$.

\medskip

\noindent\fbox{%
\colorbox{gray!10!white}{
    \parbox{\textwidth}{%
\noindent\textbf{\underline{$p$-\tucker[$\mathcal{T}$]:}} 

\smallskip

\noindent \textsc{Input:} An arithmetic  circuit $\lambda$ that outputs a number in $\bbZ_p \times [n]$. 

\smallskip

\noindent \textsc{Output:} 
\vspace{-0.07in}
\begin{enumerate}
    \item A vertex $\bx \in \partial T$ such that $\lambda(\theta \bx) \neq (\lambda_1(\bx)+1, \lambda_2(\bx))$
    \vspace{-0.07in}
    \item A simplex $\sigma^* \in T$ such that $\lambda(\sigma^*) = \bbZ_p \times \{j\}$ for some $j \in [n]$
\end{enumerate}

\noindent The valid outputs of the $p$-\tucker~problem correspond either to points that
violate the boundary condition of \cref{thm:bss-tucker} or to a fully labeled simplex in $T$.}}}

\paragraph{Remark:} To efficiently check whether $\bx \in \partial T$, we use the $\idx^*$ and $\val^*$ circuits of $T^*$, which exist by assumption on $\mathcal{T}$ for $m = 0$. If $\val^*(\idx^*(\bx))$ does not contain $\boldsymbol{0}$, then by definition of $T^*$ and the fact that $T$ is a nice triangulation $\bx \in \partial T$.

\begin{theorem}
\label{thm:BSStuckerBSSequiv}
$p$-\tucker[$\mathcal{T}$]~and $p$-\bss~are polynomially equivalent.
\end{theorem}

\begin{proof}
First, we show that $p$-\tucker[$\mathcal{T}$]~reduces to $p$-\bss~and then that $p$-\bss~reduces to $p$-\tucker[$\mathcal{T}$]. For simplicity of the presentation and in order for the main ideas to be clear we present here a proof sketch of these reductions where we assume that $\eta(\varepsilon, p) = 0$ in the $p$-\bss problem. The complete proof then follows using the tedious but straightforward case analysis that we used in the proof of Lemma \ref{lem:kPolygonBorsukUlamInPPAk}. 

\paragraph{$p$-\tucker[$\mathcal{T}$]~$\leq$~$p$-\bss: } Pick $\varepsilon < \frac{1}{(np)^2}$. Define the circuit $\mathcal{C}$ using the procedure described in \cref{lm:BBSimpliestucker} and the homeomorphism $\phi$ of $B^{n(p-1)}$ and $P_2$. Note that $\mathcal{C}$ is $L$-Lipschitz-continuous for some $L = O(2^m)$.
Thus, a solution of this instance of $p$-\bss~is:
\begin{enumerate}
    \item a point $\bx \in S^{n(p-1)-1}$ such that $\mathcal{C}(\alpha \bx) \neq \alpha \mathcal{C}(\bx)$. In this case, let $\sigma$ be the simplex such that $\phi^{-1}(\bx) \in \norm{\sigma}$, then there exists a vertex $\bv \in V(\sigma)$ such that $\lambda(\theta\bv) \neq (\lambda_1(\bv)+1, \lambda_2(\bv))$. Observe that $\sigma = \val(\idx(\phi^{-1}(\bx)))$ and that it has at most $n(p-1)+1$ vertices. Hence, we can efficiently find $\bv$.
    \vspace{-0.07in}
    \item a point $\bx^* \in B^{n(p-1)}$ such that $\norm{\mathcal{C}(\bx^*)}_{\infty} \leq \varepsilon$. In this case, $\phi^{-1}(\bx^*)$ must lie in a simplex with a fully labeled $(p-1)$-dimensional face; this follows from the choice of $\varepsilon$ and the proof ideas of \cref{lm:BBSimpliestucker}. Observe that $\phi^{-1}(\mathcal{C}(\bx^*))$ is the convex combination of at most $n(p-1) + 1$ different $e^{i,j}$'s. Hence, there must be a vector $\be^{i^*,j^*}$ that appears in the convex combination with coefficient at least $\frac{1}{n(p-1)+1}$.
    If there is an $i'$ such that $\be^{i',j^*}$ does not appear in the convex combination, then $\phi^{-1}(\mathcal{C}(\bx^*))$ has at least one coordinate at least as large as $\frac{1}{n(p-1) + 1}$. Hence, $\norm{\phi^{-1}(\mathcal{C}(\bx^*))}_{2} \geq \norm{\phi^{-1}(\mathcal{C}(\bx^*))}_{\infty} \geq \frac{1}{n(p-1) + 1}$, which means that $\norm{\mathcal{C}(\bx^*)}_{2} \geq \frac{1}{n(p-1) + 1}$. This is a contradiction since it implies that  $\norm{\mathcal{C}(\bx^*)}_{\infty} \geq \frac{\norm{\mathcal{C}(\bx^*)}_{2}}{\sqrt{np}} \geq \frac{1}{(np)^2} > \varepsilon$.

    The point $\phi^{-1}(\bx^*)$ lies in the simplex $\sigma = \val(\idx(\phi^{-1}(\bx)))$, which has at most $n(p-1) + 1$ vertices. Hence, finding the $(p-1)$-dimensional fully labeled face can be done efficiently. 
\end{enumerate}

\paragraph{$p$-\bss~$\leq$~$p$-\tucker[$\mathcal{T}$]:} Set $m$ in $\mathcal{T}$ such that $ 1/2^m  \leq \frac{\varepsilon}{n L}$. The labeling $\lambda$ is defined as in \cref{lm:tuckerImpliesBSS} using as function $g: P_2 \rightarrow P$ the function given by $\phi^{-1} \circ \mathcal{C} \circ \phi $, where $\phi$ is the homeomorphism of $B^{n(p-1)}$ and $P_2$; note that all operations can be described with an arithmetic circuit. A solution of this instance of $p$-\tucker[$\mathcal{T}$] is:
\begin{enumerate}
    \item a vertex $\bx \in \partial T$ such that $\lambda(\theta \bx) \neq (\lambda_1(\bx)+1, \lambda_2(\bx))$. In this case, it must hold that $\mathcal{C}(\alpha \circ \phi(\bx)) \neq \alpha \mathcal{C}(\phi(\bx))$. Thus, $\phi(\bx)$ is a solution of $p$-\bss.
    \vspace{-0.07in}
    \item a simplex $\sigma^* \in T$ such that $\lambda(\sigma^*) = \bbZ_p \times \{j\}$ for some $j \in [n]$. In this case, following the proof of \cref{lm:tuckerImpliesBSS}, there exists a vertex $\bx$ in $\sigma^*$ such that $\norm{g(\bx)}_{\infty} \leq \varepsilon$ (or a violation of $L$-Lipschitz-continuity). Then, since $\phi$ as defined in \cref{sec:bss} preserves the $\ell_2$ distances between $P_2$ and $B^{n(p-1)}$, $\norm{\mathcal{C}(\phi(\bx))}_{\infty} \leq \norm{\mathcal{C}(\phi(\bx))}_{2} \leq \varepsilon$. Thus, $\phi(\bx)$ is a solution of $p$-\bss.
\end{enumerate}
\end{proof}

\begin{remark}[Kuhn's triangulation for $p$-\tucker]\label{rem:bss-kuhn}

In order to use Kuhn's triangulation we work on the domain $C_\infty = \{(c^1, \dots, c^p) \in ([0,1]^n)^p | \forall j \in [n], \exists i \in [p] : c^i_j = 0 \}$ instead of $P_2$. These are coordinates with respect to the vectors $e^{i,j}$. Note that $C_\infty \cong P_\infty$, where $P_\infty = \{\sum_{i,j} c^i_j e^{i,j} | (c^1, \dots, c^p) \in C_\infty\}$, and clearly $P_\infty \cong P_2$.

We can triangulate $C_\infty$ by using Kuhn's triangulation (\cref{def:kuhn}) to triangulate each cube $\{(c^1, \dots, c^p) \in C_\infty | \forall j \in [n] : c^{i_j}_j = 0 \}$ for each $(i_1, \dots, i_n) \in [p]^n$. By the properties of Kuhn's triangulation, it immediately follows that this yields a triangulation of $C_\infty$. In particular, the triangulations of two cubes ``match'' on their common subspace. Since we constructed the triangulation separately on each cube, it follows that it refines the triangulation $T^*$. Furthermore, for any simplex $\sigma$ lying on the boundary of $C_\infty$, it follows that $\theta \sigma$ is also a simplex of the triangulation. This is easy to see, because $\theta$ just changes the order of the coordinates, and the Kuhn triangulation is invariant with respect to such transformations by definition. Thus, Kuhn's triangulation is indeed a nice triangulation.
\end{remark}

\subsection{\tucker~is \ppak{p}-complete}\label{sec:tucker-hardness}
\label{sec:tucker-completeness}
Having defined the computational problems corresponding to the BSS Theorem and to \tucker, we are ready to prove \cref{thm:tucker-completeness}. The following theorem follows from \cref{lem:bss-to-pstar}, presented in \cref{sec:pstartucker}.

\begin{theorem}\label{thm:BSS-tucker-in-ppak}
For any prime $p$, $p$-\tucker[$\mathcal{T}]$ is in \ppak{p}. 
\end{theorem}

Next we prove that $p$-\tucker\ is \ppak{p}-hard, through a reduction from \polygonTucker{p}.

\begin{theorem}\label{thm:BSS-tucker-ppak-hard}
For any prime $p \geq 3$, $p$-\tucker~is \ppak{p}-hard, even for fixed dimension $n \geq 1$.
\end{theorem}

Note that for $p=2$, $p$-\tucker~corresponds to the standard version of Tucker's lemma, which is known to be \ppa-hard for any $n \geq 2$ \citep{ABB15-2DTucker}.

\begin{proof}
Here we prove that $p$-\tucker~is \ppak{p}-hard for $n=1$. The hardness for any $n \geq 1$ then follows from \cref{lm:n1sufficies}, which gives a reduction from $n$ to $n+1$.

To show the hardness for $n=1$ we will reduce from \polygonTucker{p} to $p$-\tucker. Instead of the triangulation we used in the presentation of \polygonTucker{p}, we will use Kuhn's triangulation (\cref{def:kuhn}). It can be shown that the hardness of \polygonTucker{p} proved in \cref{prop:k-polygon-hardness}, also holds if we use Kuhn's triangulation, since there is a simple homeomorphism between the two domains. We omit the details for this.

Let $\lambda$ be an instance of \polygonTucker{p} with Kuhn's triangulation of size $m$. Thus, the domain for this problem can be written as
$$A = \{(c_1, \dots, c_p) \in (U_m)^p | \exists i \in [p], \forall j \notin \{i, i+1\} : c_j = 0\}.$$

We construct an instance of $p$-\tucker\ with $n=1$ on the domain $C_\infty$ with Kuhn's triangulation of size $m$. Recall that the set of vertices can be written as $\overline{C}_\infty := \{(c_1, \dots, c_p) \in U_m^p | \exists i \in \mathbb{Z}_p : c_i = 0\}$. Let $D := \cup_{i \in \mathbb{Z}_p} D_i$, where $D_i := \{(c_1, \dots, c_p) \in \overline{C}_\infty | \forall j \notin \{i, i+1\} : c_j = 0\}$.

We are going to embed the \polygonTucker{p} domain $A$ into $D$ in the most natural way. Since we are using Kuhn's triangulation, the restriction of the triangulation of $\overline{C}_\infty$ to $D$ corresponds to Kuhn's triangulation on that domain, and thus to Kuhn's triangulation of $A$. We define $\lambda' : \overline{C}_\infty \to \mathbb{Z}_p$:
\begin{equation*}
\lambda'(c_1, \dots, c_p) = \left\{ \begin{tabular}{ll}
    $\lambda(c_1,\dots,c_p)$ & if $(c_1,\dots,c_p) \in D$ \\
    $T_p(\{j : c_j = 0\})$ & otherwise
\end{tabular} \right.
\end{equation*}
where $T_p$ is the $\mathbb{Z}_p$-equivariant tie-breaking function defined in \cref{def:tie-breaking}. Note that $\lambda'$ is well-defined, because if $(c_1,\dots,c_p) \in D$, then it corresponds to a point in the \polygonTucker{p} domain. Furthermore, if $(c_1,\dots,c_p) \notin D$, then $|\{j : c_j = 0\}| \in [1,p-1]$, so is a valid input to $T_p$.

Since $\theta D = D$, $\lambda(\theta c) = \lambda(c) + 1$ and $T_p(\{j: (\theta c)_j = 0\}) = T_p(\{j : c_j = 0\}) + 1$, it follows that $\lambda'$ satisfies the boundary conditions.

Let $c^1, c^2, \dots, c^p$ be a $(p-1)$-simplex of $\overline{C}_\infty$ that carries all the labels in $\mathbb{Z}_p$ (with respect to $\lambda'$). We now show how this yields a solution to the \polygonTucker{p} instance. Without loss of generality, we can assume that $c^1, c^2, \dots, c^p$ are ordered in the order in which the simplex is defined by Kuhn's triangulation. Namely, for every $i \in \{1, \dots, p-1\}$, there exists $j_i$ such that $c^{i+1}_{j_i} = c^i_{j_i} + 1/m$ and $c^{i+1}_{j} = c^i_{j}$ for all $j \neq j_i$. Furthermore, the $j_i$ are all distinct. Note that if for some $i^*$, $c^{i^*} \notin D$, then $c^i \notin D$ for all $i \geq i^*$. The following cases can occur:
\begin{itemize}
    \item the simplex does not intersect $D$: it follows that $c^1 \neq 0 \in D$, and thus there exists $j$ such that $c^1_j > 0$. But then $c^i_j > 0$ for all $i$ and the label $j$ cannot be obtained by any vertex of this simplex.
    \item the intersection of the simplex with $D$ is a face of dimension $2$: then the three vertices of this face have pairwise distinct labels, and thus yield a solution to \polygonTucker{p}.
    \item the intersection of the simplex with $D$ is a face of dimension $1$: then it must hold that $c^1,c^2 \in D$ and $c^3, \dots, c^p \notin D$. We distinguish between the two sub-cases:
    \begin{itemize}
        \item $c^2_{j_1} > 0$ and $c^2_j = 0$ for all $j \neq j_1$. Then, since $c^3 \notin D$, it follows that $j_2 \notin \{j_1-1,j_1,j_1+1\}$. By definition of the $j_i$, we have that $c^3_{j_1} > 0$ and $c^3_{j_2} > 0$, which implies that $c^3, \dots, c^p$ can only have labels in $\mathbb{Z}_p \setminus \{j_1,j_2\}$. As a result, $c^1$ and $c^2$ must have labels $j_1$ and $j_2$ (in any order). This yields a solution to \polygonTucker{p}, because the labels are distinct and non-consecutive.
        \item there exists $j^*$ such that $c^2_{j^*} > 0$, $c^2_{j^*+1} > 0$ and $c^2_j = 0$ for all $j \notin \{j^*,j^*+1\}$. Since $c^3 \notin D$, it must hold that $j_2 \notin \{j^*,j^*+1\}$. Thus we have $c^3_{j_2} > 0$, and the vertices $c_3, \dots, c_p$ can only obtain the labels $\mathbb{Z}_p \setminus \{j^*,j^*+1, j_2\}$. It follows that the simplex cannot possibly be fully labeled.
    \end{itemize}
    \item the intersection of the simplex with $D$ is a face of dimension $0$: then $c^2 \notin D$, and thus there exist distinct $j,j'$ (and non-consecutive, but we don't need this here) such that $c^2_j > 0$ and $c^2_{j'} > 0$. But then, the vertices $c_2, \dots, c_p$ can only obtain the labels $\mathbb{Z}_p \setminus \{j,j'\}$. It follows that the simplex cannot possibly be fully labeled.
\end{itemize}
This completes the proof.
\end{proof}

\section{The \pstartucker{p} Lemma: Statement and \ppak{p}-completeness} \label{sec:pstartucker}

In this section, we introduce a $\mathbb{Z}_p$-generalization of Tucker's Lemma. We further define the associated computational problem \pstartucker{p}, and show that it is \ppak{p}-complete. In the next section, we will use this problem to prove the membership of $p$-thief Necklace Splitting in \ppak{p}.

In \pstartucker{p}, the coordinates of the vertices lie on a star-like domain. This domain was used by \citet{Meunier2014simplotopal} for the first fully combinatorial proof of necklace splitting with $p$ thieves.

For any prime $p$ and any $m \geq 1$, we define $R_{p,m} = \{0\} \cup \{*^i j : i \in \mathbb{Z}_p, j \in [m]\}$, where $[m] = \{1,2,\dots, m\}$. For ease of notation, we also let $*^i 0 = 0$ for all $i \in \mathbb{Z}_p$. The symbols $*^1, \dots, *^p$ should be interpreted as $p$ different ``signs'' that will generalize the use of ``$+$'' and ``$-$'' in Tucker's Lemma. The way to picture $R_{p,m}$ is as follows: the point $0$ lies at the center and there are $p$ segments of length $m$ leaving from $0$ in $p$ different directions. In that sense, we also call $R_{p,m}$ a $p$-star. The boundary of the $p$-star is the set of points $*^1 m, \dots, *^p m$. 

The $\mathbb{Z}_p$-action $\theta$ is defined on $R_{p,m}$ in the natural way, i.e., $\theta(*^ij) = *^{i+1} j$ (recall that $i \in \mathbb{Z}_p$). In particular, $\theta(0)=0$. For any $d \geq 1$, $\theta$ can be extended to $R_{p,m}^d := (R_{p,m})^d$ by simply applying $\theta$ separately to each coordinate. Note that $\theta$ is a free action when restricted to the boundary of $R_{p,m}^d$ (i.e., the points that have at least one coordinate of the form $*^{\cdot} m$). See \cref{fig:zpstartucker}.

There is a very natural metric on $R_{p,m}$. $\textup{dist}(*^{i_1}j_1,*^{i_2}j_2)$ is defined to be $|j_1-j_2|$ if $i_1=i_2$, and $j_1+j_2$ otherwise. We let $\textup{dist}_\infty(\cdot, \cdot)$ denote the generalization to $R_{p,m}^d$ (where we take the maximum). Finally, we triangulate the domain $R_{p,m}^d$ by using Kuhn's triangulation on every subcube (see \cref{rem:zp-star-kuhn} below for details). We can now state a Tucker's lemma for this domain.

\begin{theorem}[$\mathbb{Z}_p$-star Tucker's Lemma]\label{thm:ptucker}
Let $p$ be prime and $m,t \geq 1$, $d=t(p-1)$, and $T$ be Kuhn's triangulation of $R_{p,m}^d$. Let $\lambda: R_{p,m}^d \to R_{p,t} \setminus \{0\}$\footnote{Notice that the set $R_{p, t} \setminus \{0\}$ is isomorphic to the set $[p] \times [t]$.} be any labeling that satisfies $\lambda(\theta x) = \theta \lambda(x)$ for all $x \in \partial R_{p,m}^d$. Then there exists a $(p-1)$-simplex $x_1, \dots, x_p$ of $T$ and $j \in [t]$ such that $\lambda(x_i) = *^i j$ for all $i \in [p]$.
\end{theorem}

In particular, by the properties of Kuhn's triangulation, it holds that for every solution $x_1, \dots, x_p$, we have $\textup{dist}_\infty(x_i,x_k) \leq 1$ for all $i,k \in [p]$.

Note that $\mathbb{Z}_2$-star Tucker's Lemma corresponds to the standard version of Tucker's Lemma. Since we are interested in the computational aspect, we also define the naturally corresponding TFNP problem.

\bigskip

\noindent\fbox{%
\colorbox{gray!10!white}{
    \parbox{\textwidth}{%
    
    \noindent\textbf{\underline{\pstartucker{p}:}} 

\smallskip

\noindent \textsc{Input:} $m,t \geq 1$, $d=t(p-1)$ and a Boolean circuit computing a labeling $\lambda: R_{p,m}^d \to R_{p,t} \setminus \{0\}$ that satisfies $\lambda(\theta x) = \theta \lambda(x)$ for all $x \in \partial R_{p,m}^d$
\smallskip

\noindent \textsc{Output:} A $(p-1)$-simplex $x_1, \dots, x_p \in R_{p,m}^d$ and $j \in [t]$ such that $\lambda(x_i) = *^i j$ for all $i \in [p]$
}}}
\bigskip

Note that the property ``$\lambda(\theta x) = \theta \lambda(x)$ for all $x \in \partial R_{p,m}^d$'' can be enforced syntactically. Thus, \pstartucker{p} is not a promise problem. 

The main result of the section is the following theorem.

\begin{theorem}
For all primes $p$, \pstartucker{p} is \ppak{p}-complete.
\end{theorem}
The proof of the theorem will follow from \cref{thm:ptucker-in-ppa-p} and \cref{lem:bss-to-pstar} below. The hardness is obtained by a reduction from $p$-\tucker~, which is \ppak{p}-hard for all primes $p$, as shown in \cref{thm:BSS-tucker-ppak-hard}.

\begin{remark}\label{rem:zp-star-kuhn}
The domain $R_{p,m}^d$ can be triangulated in a standard way as follows. We start by subdividing the domain into hypercubes $\{*^{i_1} (a_1-1), *^{i_1} a_1\} \times \dots \times \{*^{i_d} (a_d-1), *^{i_d} a_d\}$ for $a_1,\dots, a_d \in [m]$ and $i_1, \dots, i_d \in \mathbb{Z}_p$. Then, we can use Kuhn's triangulation on each hypercube.

Similarly to \cref{rem:bss-kuhn}, the triangulation $T$ of $R_{p,m}^d$ has the following nice properties:
\begin{enumerate}
    \item The restriction of $T$ on any sub-orthant of $R_{p,m}^d$ (i.e., a subspace of the form $A_1 \times A_2 \times \dots \times A_d$, where $A_\ell = \{*^{i_\ell} j : 0 \leq j \leq m\}$ or $A_\ell = \{0\}$) yields a triangulation of that sub-orthant.
    \item On the boundary of $R_{p,m}^d$, the triangulation $T$ is symmetric with respect to $\theta$ : for any simplex $\sigma$ of $T$ that lies on the boundary of $R_{p,m}^d$, the simplices $\theta \sigma, \theta^2\sigma, \dots, \theta^{p-1} \sigma$ are also simplices of $T$ (that also lie on the boundary).
    \item $T$ is computationally efficient, in the sense that we can perform pivoting and indexing operations in polynomial time.
\end{enumerate}
\end{remark}

\subsection{\pstartucker{p} is in \ppak{p}}
First, we prove the membership of \pstartucker{p}\ in \ppak{p}. We have the following theorem.

\begin{theorem}\label{thm:ptucker-in-ppa-p}
For all primes $p$, \pstartucker{p} lies in \ppak{p}.
\end{theorem}

 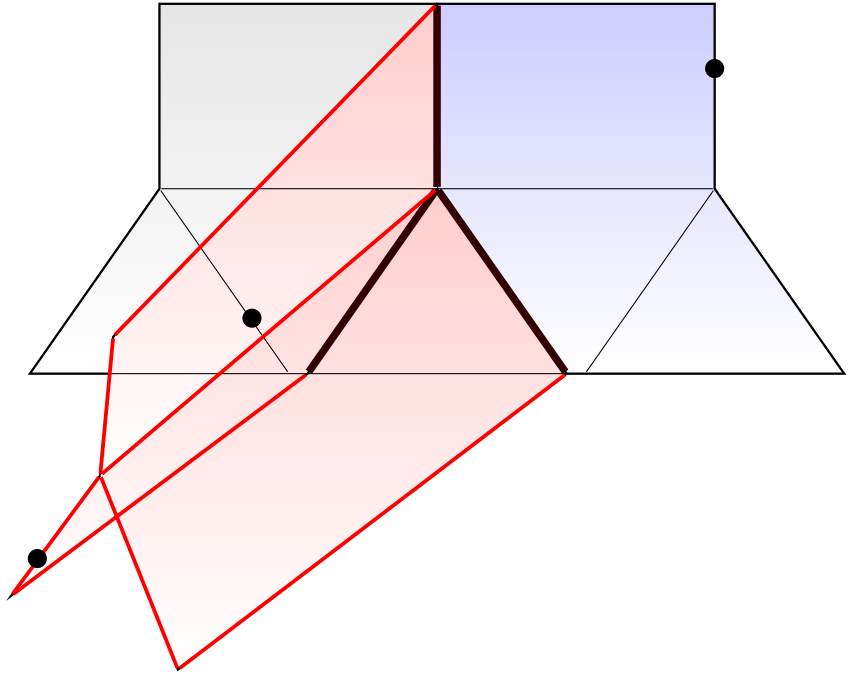
\begin{figure}
 \centering
 \begin{tikzpicture}[scale=0.7]
	\node[inner sep=0pt] (left) at (100pt,0pt) {};
	\node[inner sep=0pt] (right) at (400pt, 0pt) {};
	\draw (left)--(right);
	
	\node[inner sep=0pt] (topleft) at (100pt, 100pt) {};
	\node[inner sep=0pt] (topright) at (400pt, 100pt) {};
	\draw (left) -- (topleft);
	\draw (right) -- (topright);
	\draw (topright) -- (topleft);
	
	\node[inner sep=0pt] (midtop) at (250pt, 100pt) {};
	\node[inner sep=0pt] (midbottom) at (250pt, 0pt) {};
	\draw [line width=1mm, red!20!black] (midtop) -- (midbottom);
	
	\node[inner sep=0pt] (leftleftbot) at (30pt, -100pt) {};
	\node[inner sep=0pt] (rightleftbot) at (170pt, -100pt) {};
	\draw (left) -- (leftleftbot);
	\draw (left) -- (rightleftbot);
	
	\node[inner sep=0pt] (leftmidbot) at (180pt, -100pt) {};
	\node[inner sep=0pt] (rightmidbot) at (320pt, -100pt) {};
	
	\draw [line width=1mm, red!20!black] (midbottom) -- (leftmidbot);
	\draw [line width=1mm, red!20!black]  (midbottom) -- (rightmidbot);
	
	\node[inner sep=0pt] (leftrightbot) at (330pt, -100pt) {};
	\node[inner sep=0pt] (rightrightbot) at (470pt, -100pt) {};
	
	\draw (right) -- (leftrightbot);
	\draw (right) -- (rightrightbot);
	
	\draw(leftleftbot) -- (rightrightbot);
	
	\node[inner sep=0pt] (redtopleft) at (75pt, -80pt) {};
	\node[inner sep=0pt] (redmidleft) at (68pt, -155pt) {};
	\node[inner sep=0pt] (redbotleft) at (20pt, -220pt) {};
	\node[inner sep=0pt] (redbotright) at (110pt, -260pt) {};

	\draw [line width=0.5mm, red] (midtop) -- (redtopleft);
	\draw [line width=0.5mm, red] (redmidleft) -- (redtopleft);
	\draw [line width=0.5mm, red] (redmidleft) -- (midbottom);
	
	\draw [line width=0.5mm, red] (redmidleft) -- (redbotleft);
	\draw [line width=0.5mm, red] (redmidleft) -- (redbotright);
	
	\draw [line width=0.5mm, red] (redbotleft) -- (leftmidbot);
	\draw [line width=0.5mm, red] (redbotright) -- (rightmidbot);
	
	\begin{pgfonlayer}{background}
    \fill[draw, line width=0.3mm,top color=gray!20,bottom color=white] (midbottom.center) -- (midtop.center) -- (topleft.center) -- (left.center) -- (leftleftbot.center) -- (leftmidbot.center);
    \end{pgfonlayer}
    
    \begin{pgfonlayer}{background}
    \fill[draw, line width=0.3mm,top color=blue!20,bottom color=white] (midbottom.center) -- (midtop.center) -- (topright.center) -- (right.center) -- (rightrightbot.center) -- (rightmidbot.center);
    \end{pgfonlayer}
	
	\begin{pgfonlayer}{background}
    \fill[draw, line width=0.3mm,top color=red!20,bottom color=white] (midtop.center) -- (redtopleft.center) -- (redmidleft.center) -- (midbottom.center);
    \end{pgfonlayer}
    
    \begin{pgfonlayer}{background}
    \fill[draw, line width=0.3mm,top color=red!20,bottom color=white] (midbottom.center) -- (rightmidbot.center) -- (redbotright.center) -- (redmidleft.center);
    \end{pgfonlayer}
    
    \begin{pgfonlayer}{background}
    \fill[draw, line width=0.3mm,top color=red!20,bottom color=white] (midbottom.center) -- (leftmidbot.center) -- (redbotleft.center) -- (redmidleft.center);
    \end{pgfonlayer}

	\node[circle, fill=black,scale=0.7] (dot1) at (400pt, 65pt) {};
	\node[circle, fill=black,scale=0.7] (dot2) at (150pt, -70pt) {};
	\node[circle, fill=black,scale=0.7] (dot3) at (34pt, -200pt) {};

\end{tikzpicture}
 \caption{A view of the domain $R_{p,m}^d$ for $p=3$ and $d=2$. Note that this corresponds to $R_{3,m} \times R_{3,m}$. The three black points are in correspondence under $\theta$. The three thick lines at the center of the picture correspond to the place where the three pieces are ``glued'' together.}\label{fig:zpstartucker}
 \end{figure}

The result is proved by reducing the problem to \imba{p}. In particular, this also provides a combinatorial proof of $\mathbb{Z}_p$-star Tucker's Lemma (\cref{thm:ptucker}). This proof can be seen as a generalization of the combinatorial proof of Tucker's Lemma given by \citet{FT81}. We provide an overview of the proof below; the full proof can be found in \cref{sec:app:ptucker-in-ppa-p}.

\subsubsection{Proof Overview of \cref{thm:ptucker-in-ppa-p}}\label{sec:proofoverview}

As noted earlier, \pstartucker{2} corresponds to the standard version of Tucker's Lemma and the domain is equivalent to $\{-m,-(m-1)\dots,0,\dots,m\}^d$. The computational problem is known to lie in \ppa\ (recall that $\ppa = \ppak{2}$) by using an argument given by \citet{FT81}. More precisely, Freund and Todd gave a constructive proof of Tucker's Lemma and as noted by \citep{Papadimitriou94-TFNP-subclasses,ABB15-2DTucker}, this yields a reduction to \ppa. The constructive proof relies on a path-following argument on a graph where the nodes are simplices of the triangulation. We start by giving some details about their argument, since our proof is a generalization of their construction.

The nodes of the graph $G$ consist of all simplices of the triangulation that satisfy some properties that depend on the labels of the simplex (and the coordinate subspace orthant in which the simplex lies). Following the presentation of the proof given by \citet{Mat03BorsukUlam}, we call these ``happy simplices''. Undirected edges are added between happy simplices based on some simple rules (e.g., if they share a facet and that facet has some desired labels, etc). Given the definition of the edges, it is easy to show that the happy simplex $0^d$ has degree $1$ and any other node of degree $1$ is either a solution, or is a happy simplex lying on the boundary of the domain. In the latter case, because of the boundary conditions, this means that there must be another such simplex that lies on the antipodally opposite side of the domain and also has degree $1$. Thus, merging these two nodes into a single one yields a vertex of degree $2$, eliminating these ``fake'' solutions.

Our first contribution is to note that the edges of the graph can be directed in a \emph{consistent} way in Freund and Todd's construction. Namely, any non-merged degree-$2$ vertex has one incoming and one outgoing edge, and any merged vertex has either two incoming edges, or two outgoing edges. This yields a reduction to \textsc{Imbalance-mod-$2$}. Note that if all degree $2$ vertices were always perfectly balanced, then we would obtain a reduction to \textsc{End-of-Line}, which is impossible, unless $\ppad = \ppa$.

When we move to the case $p>2$, the notions used to define the graph can be generalized in a natural way, despite the unusual domain $R_{p,m}^d$. While the ability to direct edges was not actually needed for $p=2$, it now becomes absolutely necessary. Indeed, for any degree-$1$ happy simplex on the boundary, there are now $p-1$ other such simplices (by using $\theta$). Merging these into a single node yields degree $p$. We show that directing the edges yields a merged node that is balanced modulo $p$ (namely, all $p$ edges are incoming, or they are all outgoing).

However, another difficulty arises for $p>2$. Recall that a path can visit simplices of various dimensions. The vertices where the dimension changes are special vertices, that we call super-happy simplices. These super-happy simplices have one edge with a same or lower-dimensional happy simplex, and $k$ edges with $k$ different higher-dimensional happy-simplices, where $k \in [p-1]$. Directing the edges as before, yields that the $k$ edges are directed the same way, and in the opposite direction to the single edge. By changing the way the direction of edges is defined, it is possible to salvage the situation for $p=3$. However, this fails for any $p \geq 5$.

The solution is to carefully assign \emph{weights} to all edges. The weight of an edge only depends on the nature of the coordinate subspace orthants in which it lies, in particular the dimension. With these weights, we show that any vertex that is not a solution is now balanced modulo $p$ (except the trivial solution $0^d$). Namely, the non-solution vertices of the graph are:
\begin{itemize}
    \item the trivial solution $0^d$: all its edges are outgoing and it has degree $(-1)^t \mod p$
    \item the merged simplices on the boundary: $p$ edges, all incoming/outgoing, all the same weight
    \item happy, but not super-happy simplices: once incoming edge, one outgoing, both same weight
    \item super-happy simplices: one incoming edge with weight $w$ and $k \in [p-1]$ outgoing edges, each with weight $w/k$ (or opposite direction for all edges)
\end{itemize}
Thus, apart from the trivial solution and any actual solutions, all vertices are balanced modulo $p$.

\medskip

\begin{remark}[Path-following arguments]
Even though this proof is a natural generalization of the argument by \citet{FT81}, it is not a \emph{path-following argument} for $p \geq 3$. Indeed, in the case where $p \geq 3$, it is not clear how we could explore this graph by following a path that is guaranteed to end at a solution. In fact, we provide strong evidence that it is not possible to prove \pstartucker{p} by a path-following argument. Since \pstartucker{p} is \ppak{p}-hard (see next Section), and since a path-following proof of \pstartucker{p} would presumably show that the problem lies in \ppa, this would imply that $\ppak{p} \subseteq \ppa$, for prime $p \geq 3$. However, this is not expected to hold \citep{johnson2011reductions,GoosKSZ2019,Hollender19}.

Similarly, if one can show that \neck{p} is \ppak{p}-hard for some prime $p \geq 3$, then this would provide strong evidence that the Necklace Splitting theorem with $p$ thieves cannot be proved by a path-following argument.
\end{remark}

\subsection{\pstartucker{p} is \ppak{p}-hard}

In this section we prove that \pstartucker{p} is \ppak{p}-hard, by reducing from $p$-\tucker, which is \ppak{p}-hard by \cref{thm:BSS-tucker-ppak-hard}.

\begin{proposition}\label{lem:bss-to-pstar}
For any prime $p$, \pstartucker{p} is \ppak{p}-hard, even for fixed dimension $t \geq 1$.
\end{proposition}

\begin{proof}
We show this by reducing from $p$-\tucker\ with $n=t$.
Let $\lambda$ be an instance of $p$-\tucker\ for some prime $p$ and some $n \geq 1$, where we use Kuhn's triangulation.

In this case the domain of $p$-\tucker\ corresponds to
$$C_\infty = \{(c^1, \dots, c^p) \in U^{np}_m | \forall j \in [n], \exists i \in [p] : c^i_j = 0 \}.$$
On the other hand, the domain of \pstartucker{p} can be described as
$$A = \{(a^1, \dots, a^p) \in U^{np(p-1)}_m | \forall j,k \in [n] \times [p-1], \exists i^* \in [p] : a^i_{j,k} = 0 \textup{ for all } i \in [p] \setminus \{i^*\}\}.$$
Note that $\theta (c^1,\dots, c^p) = (c^p,c^1, \dots, c^{p-1})$ and $\theta (a^1,\dots, a^p) = (a^p,a^1, \dots, a^{p-1})$.

Define $\Psi : A \to C_\infty$ as $\Psi(a^1, \dots, a^p) = (\psi(a^1), \dots, \psi(a^p))$, where $\psi_j(a^i) = \max\{a^i_{j,k} : k \in [p-1]\}$ for all $i \in [p]$, $j \in [n]$. Note that $\Psi$ is well-defined, namely if $a \in A$, then $\Psi(a) \in C_\infty$. Indeed, for any $j \in [n]$, it holds that $|\{i \in [p]| \psi_j(a^i) > 0\}| = |\{i \in [p]| \exists k \in [p-1] : a^i_{j,k} > 0\}| \leq |\{(i,k) \in [p] \times [p-1] | a^i_{j,k} > 0\}| \leq p-1$, since for every $k \in [p-1]$ there exists at most one $i \in [p]$ such that $a^i_{j,k} > 0$ (for any fixed $j$). Furthermore, it is easy to see that $\Psi(\theta a) = \theta \Psi(a)$ by construction.

Now define the labeling $\lambda': A \to \mathbb{Z}_p \times [n]$ by $\lambda'(a) := \lambda(\Psi(a))$. Since $\Psi(A) \subseteq C_\infty$, $\lambda'$ is well-defined. Furthermore, for any $a \in \partial A$, it holds that $\Psi(a) \in \partial C_\infty$. Thus, we get that for any $a \in \partial A$, $\lambda'(\theta a) = \lambda(\Psi(\theta a)) = \lambda(\theta \Psi(a)) = \theta \lambda(\Psi(a)) = \theta \lambda'(a)$, i.e., $\lambda'$ satisfies the boundary conditions.

If $\sigma = \{z_1, \dots, z_p\}$ is a $(p-1)$-simplex of $A$ such that $\lambda'(\sigma) = \mathbb{Z}_p \times \{j\}$ for some $j \in [n]$, then the set of vertices $\Psi(\sigma) = \{\Psi(z_1), \dots, \Psi(z_p)\}$ satisfies $\lambda(\Psi(\sigma)) = \mathbb{Z}_p \times \{j\}$. Thus, the proof is completed by the following claim:
\begin{claim}
If $\sigma = \{z_1, \dots, z_p\}$ is a $(p-1)$-simplex in Kuhn's triangulation of $A$, then $\Psi(\sigma) = \{\Psi(z_1), \dots, \Psi(z_p)\}$ is a simplex in Kuhn's triangulation of $C_\infty$.
\end{claim}

\begin{proof}
Since we use Kuhn's triangulation, without loss of generality, we can assume that $z_1, \dots, z_p$ are ordered such that $z_1 \leq z_2 \leq \dots \leq z_p$ (component-wise) and $\|z_1-z_p\|_\infty \leq 1/m$. By construction of $\Psi$ it holds that $\Psi(a) \leq \Psi(a')$ whenever $a \leq a'$. Thus, $\Psi(z_1) \leq \dots \leq \Psi(z_p)$. In order to show that $\Psi(\sigma)$ is a simplex in Kuhn's triangulation of $C_\infty$, it remains to prove that $\|\Psi(z_1) - \Psi(z_p)\|_\infty \leq 1/m$. To see this, note that if $\|a-a'\|_\infty \leq 1/m$, then for all $i \in [p]$ and $j \in [n]$, we get that $|\psi_j(a^i) - \psi_j(a'^i)| = |\max\{a^i_{j,k} : k \in [p-1]\} - \max\{a'^i_{j,k} : k \in [p-1]\}| \leq 1/m$.
\end{proof}

This concludes the proof of \cref{lem:bss-to-pstar}.
\end{proof}

\section{Necklace Splitting with $p$ Thieves lies in \ppak{p}} \label{sec:necklace}

In this section, we prove our main result regarding the Necklace Splitting Theorem; we prove that the associated computational problem \neck{p} lies in \ppak{p} for any prime $p$.

\begin{theorem}
\label{thm:inclusions}
For every prime $p$, \neck{p} is in \ppak{p}. 
\end{theorem}

As a corollary, we also obtain that:
\begin{itemize}
    \item \neck{p^r} is in \ppak{p^r} for any prime $p$ and $r \geq 1$
    \item \neck{k} lies in \ppak{k} \emph{under Turing reductions} for any $k \geq 2$.
\end{itemize}

\noindent As we explained in the introduction, our reduction will go via the continuous version of the problem, the
\textsc{$\varepsilon$-Consensus-$1/p$-Division} problem, which is the computational analogue of the Consensus-$1/p$-Division problem of \citet{SS03-Consensus}. We will show the inclusion of \textsc{$\varepsilon$-Consensus-$1/p$-Division} in \ppak{p} via a reduction to \pstartucker{p}, which also implies the \ppak{p} membership for \neck{p}.
We prove the following main statement:

\begin{theorem}\label{thm:con1/p-to-ptucker}
For any prime $p$, \textsc{Consensus-$1/p$-Division} reduces to \pstartucker{p}.
\end{theorem}

The proof is presented in \cref{sec:con1/p-to-ptucker-proof}, where we prove an even stronger version of this theorem. Indeed, we show that this result holds for any probability measures (not only step functions), as long as they are efficiently computable and sufficiently continuous (in some precise sense).

We start with the complete definitions of the computational problems corresponding to $k$-thief Necklace Splitting and Consensus-$1/k$-Division. \bigskip

\noindent\fbox{%
\colorbox{gray!10!white}{
    \parbox{\textwidth}{%
\noindent \underline{\neck{k}} \citep{Papadimitriou94-TFNP-subclasses,FRG18-Necklace}

\smallskip

\noindent \textsc{Input:} An open necklace with $n$ beads, each of which has one of $t$ colors. 

There are exactly $a_i k$ beads of color $i =1,\dots,t$, where $a_i \in \mathbb{N}$.
\smallskip

\noindent \textsc{Output:} A partitioning of the necklace into $k$ (not necessarily connected) pieces such that each piece contains exactly $a_i$ beads of color $i$, using at most $(k-1)t$ cuts.
}}} \bigskip

As proved by \citet{Alon87-Necklace}, this problem always has a solution. Furthermore, for any proposed splitting of the necklace it is easy to check if it is a solution. Thus, the problem \neck{k} lies in TFNP.

\citet{Alon87-Necklace} proved the necklace-splitting theorem by showing existence for a more general continuous version and then ``rounding'' a solution of the continuous problem to obtain a splitting of the necklace. When investigating the complexity of \neck{k}, it is also convenient to consider the more general continuous version. Even though the continuous theorem is termed as a ``generalized Hobby-Rice theorem'' by \citet{Alon87-Necklace}, we instead use the term \emph{Consensus-$1/k$-Division} proposed by \citet{SS03-Consensus}. \bigskip

\noindent\fbox{%
\colorbox{gray!10!white}{
    \parbox{\textwidth}{%
\noindent \underline{\textsc{$\varepsilon$-Consensus-$1/k$-Division}} \citep{filos2018hardness,FRG18-Consensus} 

\smallskip

\noindent \textsc{Input:} $\varepsilon > 0$ and continuous probability measures $\mu_1, \dots, \mu_t$ on $[0,1]$. 

The probability measures are given by their density functions on $[0,1]$, which are step functions (explicitly given in the input).
\smallskip

\noindent \textsc{Output:} A partitioning of the unit interval into $k$ (not necessarily connected) pieces $A_1,\dots,A_k$ using at most $(k-1)t$ cuts, such that $|\mu_j(A_i) - \mu_j(A_\ell)| \leq \varepsilon$ for all $i,j,\ell$.
}}} \bigskip

The fact that this problem also lies in TFNP immediately follows from showing that it lies in \ppak{k} under Turing reductions (\cref{thm:in-ppak}). Note that we can equivalently ask for $\mu_j(A_i) = 1/k \pm \varepsilon$ for all $i,j$. Indeed, the computational problems are equivalent.

Furthermore, using the same technique as \citet[Theorem 5.2]{EY10-Nash-FIXP}, one can show that if $\varepsilon$ is sufficiently small (with respect to the representation size of the step functions), then one can efficiently compute an \emph{exact} solution from an $\varepsilon$-approximate solution. It follows that \textsc{$\varepsilon$-Consensus-$1/k$-Division} is equivalent to exact \textsc{Consensus-$1/k$-Division}. In particular, the problem always has an exact solution that is rational. Thus, we will sometimes refer to this problem just as \textsc{Consensus-$1/k$-Division}.

Alon's rounding procedure yields a reduction from \neck{k} to exact \textsc{Consensus-$1/k$-Division}. \citet{FRG18-Consensus} extended this result by showing that \neck{k} reduces to \textsc{$\varepsilon$-Consensus-$1/k$-Division}, even when $\varepsilon$ is not small enough to ensure that we can get an exact solution.

\begin{proposition}[\citet{Alon87-Necklace,FRG18-Consensus}]\label{prop:kneck-to-con1/k}
For any $k \geq 2$, \neck{k} reduces to \textsc{$\varepsilon$-Consensus-$1/k$-Division}.
\end{proposition}

Before we proceed with the proof of \cref{thm:con1/p-to-ptucker}, we present the consequences of \cref{thm:ptucker-in-ppa-p,thm:con1/p-to-ptucker}, in terms of computational complexity and mathematical existence.

\subsubsection*{Consequences: Computational Complexity}

\begin{theorem}\label{thm:in-ppak}
For any $k \geq 2$, \neck{k} and \textsc{Consensus-$1/k$-Division} lie in the Turing closure of \ppak{k}. In particular, if $k=p^r$ where $p$ is prime and $r \geq 1$, then the problems lie in \ppak{p}.
\end{theorem}
The Turing closure of \ppak{k} is the class of all TFNP problems that Turing-reduce to a \ppak{k}-complete problem (e.g., \imba{k}). Note that when $k$ is not a prime power, \ppak{k} is not believed to be closed under Turing reductions \citep{GoosKSZ2019,Hollender19}.

\begin{proof}
\cref{thm:ptucker-in-ppa-p} and \cref{thm:con1/p-to-ptucker} immediately imply that for any prime $p$, \textsc{Consensus-$1/p$-Division} lies in \ppak{p}. As noted by \citet[Proposition 3.2]{Alon87-Necklace}, for any $k,\ell \geq 2$, a \textsc{Consensus-$1/(k\ell)$-Division} can be obtained by first finding a \textsc{Consensus-$1/k$-Division} -- which divides the interval into $k$ (not necessarily connected) pieces -- and then finding a \textsc{Consensus-$1/\ell$-Division} of each of the $k$ pieces. Note that we obtain (at most) the desired number of cuts. Thus, we can solve an instance of \textsc{Consensus-$1/(k\ell)$-Division} by first solving an instance of \textsc{Consensus-$1/k$-Division}, and then $k$ instances of \textsc{Consensus-$1/\ell$-Division}.

In particular, for any prime $p$ and any $r \geq 1$, \textsc{Consensus-$1/p^r$-Division} can be solved by solving $1+p+p^2+ \dots + p^{r-1}$ instances of \textsc{Consensus-$1/p$-Division}. Thus, we obtain a Turing reduction from \textsc{Consensus-$1/p^r$-Division} to \textsc{Consensus-$1/p$-Division}. Since \textsc{Consensus-$1/p$-Division} lies in \ppak{p} and \ppak{p} is closed under Turing reductions \citep{GoosKSZ2019,Hollender19}, it follows that \textsc{Consensus-$1/p^r$-Division} lies in \ppak{p}.

Now consider any $k = \prod_{i=1}^m p_i^{r_i}$, where $m \geq 1$, $r_i \geq 1$ and the $p_i$ are distinct primes. Then, \textsc{Consensus-$1/k$-Division} can be solved by a query to \textsc{Consensus-$1/p_1^{r_1}$-Division}, then $p_1^{r_1}$ queries to \textsc{Consensus-$1/p_2^{r_2}$-Division} (which can be turned into a single query to \ppak{p_2}), then $p_1^{r_1}p_2^{r_2}$ queries to \textsc{Consensus-$1/p_3^{r_3}$-Division} (which can also be turned into a single query to \ppak{p_3}), etc. Thus, \textsc{Consensus-$1/k$-Division} can be solved by a query to \ppak{p_1}, then a query to \ppak{p_2}, then one to \ppak{p_3}, $\dots$, and finally a query to \ppak{p_m}. Since $\ppak{p_i} \subseteq \ppak{k}$ for $i = 1, \dots, m$ (\cref{prop:ppak-properties}), it follows that there is a Turing reduction from \textsc{Consensus-$1/k$-Division} to a \ppak{k}-complete problem (e.g., \imba{k}).

Since \neck{k} reduces to \textsc{Consensus-$1/k$-Division} (\cref{prop:kneck-to-con1/k}), the results also hold for \neck{k}.
\end{proof}

\subsubsection*{Consequences: Mathematical Existence}

\cref{thm:ptucker-in-ppa-p,thm:con1/p-to-ptucker} yield a reduction from \textsc{Consensus-$1/p$-Division} to \imba{p}. Since every instance of \imba{p} has a solution (and the proof of this is trivial), we obtain a proof that \textsc{Consensus-$1/p$-Division} always has a solution. Thus, this proves that if the probability measures are step functions (described by rational numbers), there always exists a consensus-$1/p$-division. While we have given the proof in terms of a reduction (since this is required for our complexity results), it can also be written as a mathematical proof of existence (without any computational considerations).

Once existence of a consensus-$1/p$-division for step functions has been proved, a constructive argument by \citet[Section 2]{Alon87-Necklace} also gives existence for $p$-necklace-splitting. Putting everything together, the proof of $p$-necklace-splitting thus obtained is a fully combinatorial proof that does not use any advanced machinery and is easier to follow than existing proofs. Indeed, as we mentioned in the introduction, the original proof by \citet{Alon87-Necklace} used the BSS theorem of \citet{BSS81} as a black box. The only other fully combinatorial proof by \citet{Meunier2014simplotopal}, while quite elegant, is significantly more involved.

Going back to consensus-$1/p$-division, the proof we obtain (which uses $\mathbb{Z}_p$-star Tucker's lemma) actually works for any probability measures, not only step functions. Moreover, similarly to \citet{SS03-Consensus}, our proof does not make use of the fact that the measures are additive and non-negative. Thus, we obtain a stronger version of the consensus-$1/p$-division theorem given by \citet[Theorem 1.2]{Alon87-Necklace} (which he calls a generalization of the Hobby-Rice theorem). Let $\mathcal{B}([0,1])$ denote the Borel $\sigma$-algebra on the unit interval and let $\lambda$ denote the Lebesgue measure on the unit interval. Finally, let $\triangle$ denote the symmetric difference, i.e., $A \triangle B = (A \setminus B) \cup (B \setminus A)$.

\begin{theorem} \label{thm:strongConsensusDivision}
Let $p$ be any prime and $t \geq 1$. Let $v_1, \dots, v_t : \mathcal{B}([0,1]) \to \mathbb{R}$ be such that for all $1 \leq j \leq t$ $v_j$ satisfies the following continuity condition:
for all $\varepsilon > 0$ there exists $\delta > 0$ such that
$$|v_j(A) - v_j(B)| \leq \varepsilon \quad \text{for all } A,B \in \mathcal{B}([0,1]) \text{ that satisfy } \lambda(A \triangle B) \leq \delta.$$
Then, there exists a consensus-$1/p$-division. Namely, it is possible to partition the unit interval into $p$ (not necessarily connected) pieces $A_1, \dots, A_p$ using at most $(p-1)t$ cuts, such that $v_j(A_i) = v_j(A_\ell)$ for all $1 \leq i,\ell \leq p$, $1 \leq j \leq t$.
\end{theorem}

Before we move on to the proof, let us briefly explain why we only obtain the result for prime $p$. In the usual setting where the valuations are probability measures, it is enough to prove the statement for primes. Indeed, using the standard argument by \citet[Proposition 3.2]{Alon87-Necklace}, if a consensus-$1/k$-division and a consensus-$1/\ell$-division always exist, then a consensus-$1/(k\ell)$-division exists. But Alon's argument makes use of the additivity property of the measures. Indeed, consider the non-additive setting and say that we are trying to show that a consensus-$1/4$-exists. We know that a consensus-$1/2$-division exists and this yields a partition of $[0,1]$ into $A_1$ and $A_2$. We have that $v_j(A_1) = v_j(A_2)$ for all $j$. Following Alon's argument, find a consensus-$1/2$-division of $A_1$ and of $A_2$. This yields $A_{11} \cup A_{12} = A_1$ and $A_{21} \cup A_{22} = A_2$ such that $v_j(A_{11}) = v_j(A_{12})$ and $v_j(A_{21}) = v_j(A_{22})$ for all $j$. However, we might not have $v_j(A_{11}) = v_j(A_{22})$, since we no longer have $v_j(A_{11}) + v_j(A_{12}) = v_j(A_{21}) + v_j(A_{22})$.

If we assume that the valuations are additive (even just finite additivity), but still allow them to take negative values, then Alon's argument works as before and thus the result again holds for any $k \geq 2$.

\begin{proof}
Using $\mathbb{Z}_p$-star Tucker's lemma (\cref{thm:ptucker}) it follows that for any $\varepsilon > 0$, there exists an $\varepsilon$-approximate consensus-$1/p$-division, i.e., we can partition the interval into $p$ pieces $A_1, \dots, A_p$ (using at most $(p-1)t$ cuts) such that $|v_j(A_i) - v_j(A_\ell)| \leq \varepsilon$ for all $1 \leq i,\ell \leq p$, $1 \leq j \leq t$. Indeed, it suffices to follow the same steps as in the proof of \cref{thm:con1/p-to-ptucker-general}.

Let $R_p$ denote the continuous version of $R_{p,m}$. $R_p$ consists of $p$ copies of the segment $[0,1]$ (namely $*^1[0,1]$, $\dots$, $*^p[0,1]$) that share the same origin (i.e., $*^10 = \dots = *^p0$). Using the interpretation given in the proof of \cref{thm:con1/p-to-ptucker-general}, every point in $R_p^d$ corresponds to a way to partition $[0,1]$ into $p$ (not necessarily connected) pieces using at most $(p-1)t$ cuts. Since $(R_p^d,\textup{dist}_\infty)$ is a compact metric space, every sequence must have a subsequence that converges. Thus, a sequence $(x_n)_n$, where $x_n \in R_p^d$ is a $1/2^n$-approximate consensus-$1/p$-division, must have a subsequence that converges to some $x \in R_p^d$. For any $i,j$ the function $f: R_p^d \to \mathbb{R}$, $y \mapsto v_j(A_i(y))$ is continuous. It follows that $x$ must correspond to an exact consensus-$1/p$-division.
\end{proof}

\subsection{Reduction from {\normalfont \scshape Consensus-$1/p$-Division} to \pstartucker{p}}\label{sec:con1/p-to-ptucker-proof}

In order to make our \ppak{p}-membership result as strong as possible, we define a computational problem that is much more general than \textsc{$\varepsilon$-Consensus-$1/p$-Division}. Namely, we allow any computationally reasonable probability measures that are also sufficiently continuous.

The probability measures are given by their cumulative functions. Let $\mathcal{F}$ be a class of cumulative distribution functions on $[0,1]$. Thus, for any $f \in \mathcal{F}$ and any $a \in [0,1]$, $f(a)$ is the probability of the interval $[0,a]$ according to $f$. For any $f \in \mathcal{F}$, let $\sz(f)$ denote the size of the representation of $f$. E.g., if $f$ is represented as a circuit, then $\sz(f)$ is the size of the circuit. For any rational number $x$, $\sz(x)$ denotes the representation length of $x$, i.e., the length of the binary representation of the denominator and numerator of $x$. We will require two properties from $\mathcal{F}$:

\begin{itemize}
    \item $\mathcal{F}$ is polynomially computable: there exists a polynomial $q_1$ such that for all $f \in \mathcal{F}$ and all rational $x \in [0,1]$, $f(x)$ can be computed in time $q_1(\sz(f)+\sz(x))$.
    \item $\mathcal{F}$ is polynomially continuous: there exists a polynomial $q_2$ such that for all $f \in \mathcal{F}$ and all rational $\widehat{\varepsilon} > 0$, there exists rational $\widehat{\delta} > 0$ with $\sz(\widehat{\delta}) \leq q_2(\sz(f)+\sz(\widehat{\varepsilon}))$ such that $|x-y| \leq \widehat{\delta} \implies |f(x)-f(y)| \leq \widehat{\varepsilon}$ for all $x,y \in [0,1]$.
\end{itemize}
These properties are quite natural and they were used by \citet{EY10-Nash-FIXP} in the context of fixed point problems. In particular, they hold when $\mathcal{F}$ is the class of all cumulative distribution functions given by step function densities (represented explicitly). But they also hold for much more general families.

\begin{definition}
Let $k \geq 2$ and let $\mathcal{F}$ be a polynomially computable and polynomially continuous class of cumulative distribution functions on $[0,1]$. The problem \textsc{$\varepsilon$-Consensus-$1/k$-Division[$\mathcal{F}$]} is defined exactly as \textsc{$\varepsilon$-Consensus-$1/k$-Division}, except that the probability measures are given by cumulative distribution functions in $\mathcal{F}$.
\end{definition}

Notice that \textsc{$\varepsilon$-Consensus-$1/k$-Division} corresponds to the special case where $\mathcal{F}$ is the class of all cumulative distribution functions given by step function densities (represented explicitly). Thus, the following is a stronger version of \cref{thm:con1/p-to-ptucker}.

\begin{theorem}\label{thm:con1/p-to-ptucker-general}
Let $p$ be prime and $\mathcal{F}$ be a polynomially computable and polynomially continuous class of cumulative distribution functions on $[0,1]$. Then \textsc{$\varepsilon$-Consensus-$1/p$-Division[$\mathcal{F}$]} reduces to \pstartucker{p}.
\end{theorem}

\begin{proof}
Let $\varepsilon > 0$ and $\mu_1, \dots, \mu_t$ be probability measures on $[0,1]$ given by functions in $\mathcal{F}$. We consider the domain $D = R_{p,m}^d$, where $d=t(p-1)$ and $m \geq 1$ will be set later. A point in $D$ represents a way to partition $[0,1]$ into $p$ (not necessarily connected) pieces using at most $t(p-1)$ cuts. This is a slight modification of the domain that was used by \citet{Meunier2014simplotopal} to encode a splitting of a necklace. Intuitively it can be explained as follows. Let $x=(*^{i_1} j_1, \dots, *^{i_d} j_d) \in D$. We interpret each element $i \in \mathbb{Z}_p$ as a different color. Then:
\begin{enumerate}
    \item Paint the whole interval $[0,1]$ with the color $1$.
    \item For $\ell = 1,2, \dots, d$ : paint $[0,j_\ell/m]$ with the color $i_\ell$
\end{enumerate}
Note that applying a fresh coat of paint on a previously painted part of the interval covers up the old paint. The way the interval $[0,1]$ is colored at the end of this procedure gives us the partition encoded by $x \in D$. An important advantage of this encoding is that it is sufficiently continuous in a certain sense. Indeed, small changes in the coordinates of $x$ have a small effect on the corresponding partition. Other simpler encoding schemes do not have this property.

Formally, this encoding can be described as follows. Add a ``fake'' $0$th coordinate $*^{i_0} j_0 = *^1 m$. Place cuts at all positions $j_0/m, j_1/m, \dots, j_d/m$. This subdivides the interval $[0,1]$ into at most $d+1=t(p-1)+1$ subintervals. Then, allocate the subinterval $[a,b]$ to $i_{\widehat{\ell}} \in \mathbb{Z}_p$, where $\widehat{\ell} = \max \{0 \leq \ell \leq d : j_\ell/m \geq b\}$.

This encoding also behaves nicely with respect to $\theta$. For any $x \in \partial D$, $\theta x$ encodes the same partition as $x$, except that $i$ has been replaced by $i+1$, for all $i \in \mathbb{Z}_p$. This is easy to see since for any $x \in \partial D$, there exists $\ell \geq 1$ such that $j_\ell = m$ and thus the ``fake'' coordinate $*^{i_0} j_0$ does not play any role.

We are now ready to define the labeling $\lambda : D \to R_{p,t} \setminus \{0\}$. This labeling is a natural generalization of the one used by \citet{SS03-Consensus}. Given $x \in D$, construct the partition it encodes, namely $A_1(x), \dots, A_p(x)$. Then, for all $i \in \mathbb{Z}_p$ and $j \in [t]$, let $\mu_{j,i}(x) = \mu_j(A_i(x))$, i.e., the total measure of type $j$ that is allocated to $i$. Finally, set $\lambda(x) = *^i j$, where $i,j$ are determined as follows:
\begin{enumerate}
    \item Pick $j \in [t]$ that maximizes $\max_{i_1,i_2} |\mu_{j,i_1}(x) - \mu_{j,i_2}(x)|$. Break ties by picking the minimum such $j$.
    \item Then, pick $i \in \mathbb{Z}_p$ that maximizes $\mu_{j,i}(x)$. If there are multiple $i$'s that maximize this, break ties by picking the one such that $\min A_i(x)$ is minimal (i.e., such that $A_i(x)$ contains the point closest to the left end of the unit interval).
\end{enumerate}

By using the observation above about the behavior of $\theta$ on $\partial D$, it is easy to see that $\lambda(\theta x) = \theta \lambda(x)$ for all $x \in \partial D$. Thus, $\lambda$ is a valid instance of \pstartucker{p} and we obtain a solution $x_1, \dots, x_p \in D$ and $\widehat{\jmath} \in [t]$, such that $\textup{dist}_\infty (x_i,x_k) \leq 1$ and $\lambda(x_i) = *^i \widehat{\jmath}$ for all $i,k \in [p]$. It remains to show that by picking $m$ large enough, we obtain a solution to \textsc{$\varepsilon$-Consensus-$1/p$-Division[$\mathcal{F}$]}.

Let $\varepsilon' = \frac{\varepsilon}{2t(p-1)}$. Since $\mathcal{F}$ is polynomially continuous, we can pick $m$ large enough so that the value $\mu_j([0,a])$ changes by at most $\varepsilon'$, if $a$ moves by $1/m$. Note that $m$ has representation length polynomial in the size of the instance. If a single coordinate of $x$ changes by $1$, $\mu_{j,i}(x)$ changes by at most $\varepsilon'$ for all $i,j$. Since there are $d=t(p-1)$ coordinates, it follows that $|\mu_{j,i}(x_k) - \mu_{j,i}(x_\ell)| \leq \varepsilon' t (p-1) = \varepsilon/2$ for all $i,j$  and for all $k,\ell \in [p]$.

Let $x := x_1$. By construction of the labeling, we obtain that for all $j,i,\ell$
$$|\mu_{j,i}(x) - \mu_{j,\ell}(x)| \leq \max_{i_1,i_2} |\mu_{\widehat{\jmath},i_1}(x) - \mu_{\widehat{\jmath},i_2}(x)| \leq \varepsilon$$
The first inequality holds because $\lambda(x) = *^1 \widehat{\jmath}$. The second inequality holds because if we instead had $\mu_{\widehat{\jmath},i_1}(x) > \mu_{\widehat{\jmath},i_2}(x) + \varepsilon$ for some $i_1, i_2$, then it would follow that $\mu_{\widehat{\jmath},i_1}(x_{i_2}) > \mu_{\widehat{\jmath},i_2}(x_{i_2})$, contradicting $\lambda(x_{i_2}) = *^{i_2} \widehat{\jmath}$. Thus, $x$ corresponds to an $\varepsilon$-approximate solution.
\end{proof}

\section{Conclusion and Future Work}

Our topological characterization of \ppak{p} can possibly enable us to obtain similar membership or hardness results for other interesting problems. For example, are the problems whose totality is established via the BSS Theorem, like the Chromatic Number of Kneser Hypergraphs\footnote{The computational version of this problem would be of the form: given a coloring (as a circuit) that cannot possibly be correct everywhere, because it does not use enough colors, find any edge where it makes a mistake.} studied in \citep{alon1986chromatic} in \ppak{p}? Are they \ppak{p}-complete? We believe that due to its simplicity, our \polygonTucker{p}\ problem can be a very useful tool for obtaining hardness results for these problems. What about other problems that generalize problems that are known to be in \ppa\ or are even \ppa-complete? For example, \citet{FRG18-Necklace} showed that the discrete Ham-Sandwich problem is also \ppa-complete. Is there a generalization of the problem that could be complete for \ppak{p}? A computational version of the Center Transversal Theorem \citep{dol1992generalization,zivaljevic1990extension} might be a good candidate. Another interesting open problem is to investigate the connection of the general statement of Dold's Theorem \citep{Dold1983} from algebraic topology with the subclasses of TFNP. Finally, although our paper takes a definitive step in the direction of resolving the complexity of $p$-thief Necklace Splitting and Consensus-$1/p$-Division, proving a \ppak{p}-hardness result remains a challenging open problem. In very recent work \citep{FRHSZ2020consensus-easier}, we have made a first step in that direction, by providing a significantly simpler proof (and strengthening) of the \ppak{2}-hardness result of \citet{FRG18-Necklace}, as well as the first hardness result for Consensus-$1/3$-Division, showing that it is hard for the class \ppad. Showing the \ppad-hardness of the Consensus-$1/p$-Division problem for $p > 3$ is also a very interesting first step that might be easier than showing the \ppak{p}-hardness.

\medskip

\subsubsection*{Acknowledgments}
Alexandros Hollender is supported by an EPSRC doctoral studentship (Reference 1892947). Katerina Sotiraki is supported in part by NSF/BSF grant \#1350619, an MIT-IBM grant, and a DARPA Young Faculty Award,  MIT Lincoln Laboratories and Analog Devices. Part of this work was done while the author was visiting the Simons Institute for the Theory of Computing. Manolis Zampetakis is supported by a Google Ph.D. Fellowship and NSF Award \#1617730.  \\

\noindent We would like to thank Paul Goldberg and Fr{\'e}d{\'e}ric Meunier for helpful discussions and comments, as well as the anonymous reviewers for their suggestions that helped improve the paper.

\bigskip

\bibliographystyle{plainnat}
\bibliography{biblio}

\clearpage

\appendix

\section{Topological Definitions}
\label{ap:definitions}

In this section, we include all the necessary notation and topological definitions that are used throughout the paper. For a more detailed exposition on simplicial complexes, we refer the interested reader to \citep{Mat03BorsukUlam}.

\paragraph{Notation:} Let $B^n = \{x \in \mathbb{R}^n : \|x\|_2 \leq 1\}$ denote the $n$-dimensional unit ball and $S^{n-1} = \partial B^n$ be the corresponding unit sphere.

\begin{definition}[\textsc{Homeomorphism}]
A \emph{homeomorphism} of topological spaces $(X_1, \mathcal{O}_1)$ and
$(X_2, \mathcal{O}_2)$ is a bijection $\phi: X_1 \rightarrow X_2$ such that for every $U \subseteq X_1$, $\phi(U) \in \mathcal{O}_2$ if and only if $U \in \mathcal{O}_1$. In other words, a bijection $\phi: X_1 \rightarrow X_2$ is a
homeomorphism if and only if both $\phi$ and $\phi^{-1}$ are continuous. If there is a homeomorphism $\phi: X_1 \rightarrow X_2$, we write $X \cong Y$.
\end{definition}

We say that a function $f$ has \emph{order} $p$ if $f^{p} = f$, where the notation $f^{i}$ denotes to the composition of $f$ by itself $i$ times.
\begin{definition}[\textsc{Free Action}]
Let $f: X \rightarrow Y$ be a function of order $p$ and let $P$ be a set. We say that $f$ acts freely on $P$ if for all $x \in X$ and all $i \in \{1, \dots, p-1\}$, $f^{i}(x) \neq x$. 
\end{definition}

\begin{definition}[\textsc{Affine Independence}]
We call the points $\bv_1, \bv_2, \dots, \bv_k$ \emph{affine dependent} if there exist numbers $a_1, a_2, \dots, a_k \in \bbR$ not all 0 such that $\sum\limits_{i = 1}^k a_i\bv_i = \boldsymbol{0}$ and $\sum\limits_{i = 1}^k a_i = 0$. Otherwise, $\bv_1, \dots, \bv_k$ are called \emph{affine independent}.
\end{definition}

Geometrically, some examples of simplices are points, lines and triangles. Formally, the definition requires the notion of affine independence.

\begin{definition}[\textsc{Simplex}]
A \emph{simplex} $\sigma$ is the convex hull of a finite set $A$ of affine independent vectors in $\bbR^n$. The points in  $A$ are called the \emph{vertices} of $\sigma$ and denoted by $V(\sigma)$. The dimension of $\sigma$ is equal to $|A| - 1$. Namely, a $k$ dimensional simplex, called $k$-simplex for short, has $k+1$ vertices. The convex hull of an arbitrary subset of the vertices of $\sigma$ is called a \emph{face} of $\sigma$. A proper face of $\sigma$ is called \emph{facet}.
\end{definition}

From the above definitions, it holds that every face is itself a simplex. For simplicity, we denote a simplex as the set of its vertices.\\

\subsection{Simplicial Complexes, Value \& Index Functions and Triangulations}
Very central to our paper is the notion of \emph{geometric simplicial complexes}, which are used to describe subspaces of $\bbR^d$. These subspaces consist of simple building blocks, such as points, line segments, triangles, tetrahedra, that are pasted together.

\begin{definition}[\textsc{Simplicial Complex}]
A \emph{simplicial complex} $K$ is a non-empty set of simplices that satisfies the following properties:
\begin{itemize}
    \item Each face of a simplex $\sigma \in K$ is also a simplex in $K$.
    \item The intersection $\sigma_1 \cap \sigma_2$ of any two simplices $\sigma_1, \sigma_2 \in K$ is a face of both $\sigma_1$ and $\sigma_2$. 
\end{itemize}
The union of the simplices in $K$ is called the polyhedron of $K$ and is denoted by $\norm{K}$. The dimension of $K$ is $\mathsf{dim(K)} := \max_{\sigma \in K}\{\mathsf{dim}(\sigma)\}$ and the vertex set of $K$, denoted by $V(K)$, is the union of the vertex sets of all its simplices.

We denote by $\Sigma_\sigma$ for a simplex $\sigma$ the simplicial complex that contains all simplices $\tau$ such that $\tau \subseteq \sigma$.
\end{definition}

According to the above definition, zero-dimensional simplicial complexes correspond to points and one-dimensional simplicial complexes to sets of non-intersecting line segments as shown in \cref{fig:complex}.

\begin{figure}
    \centering
    \includegraphics{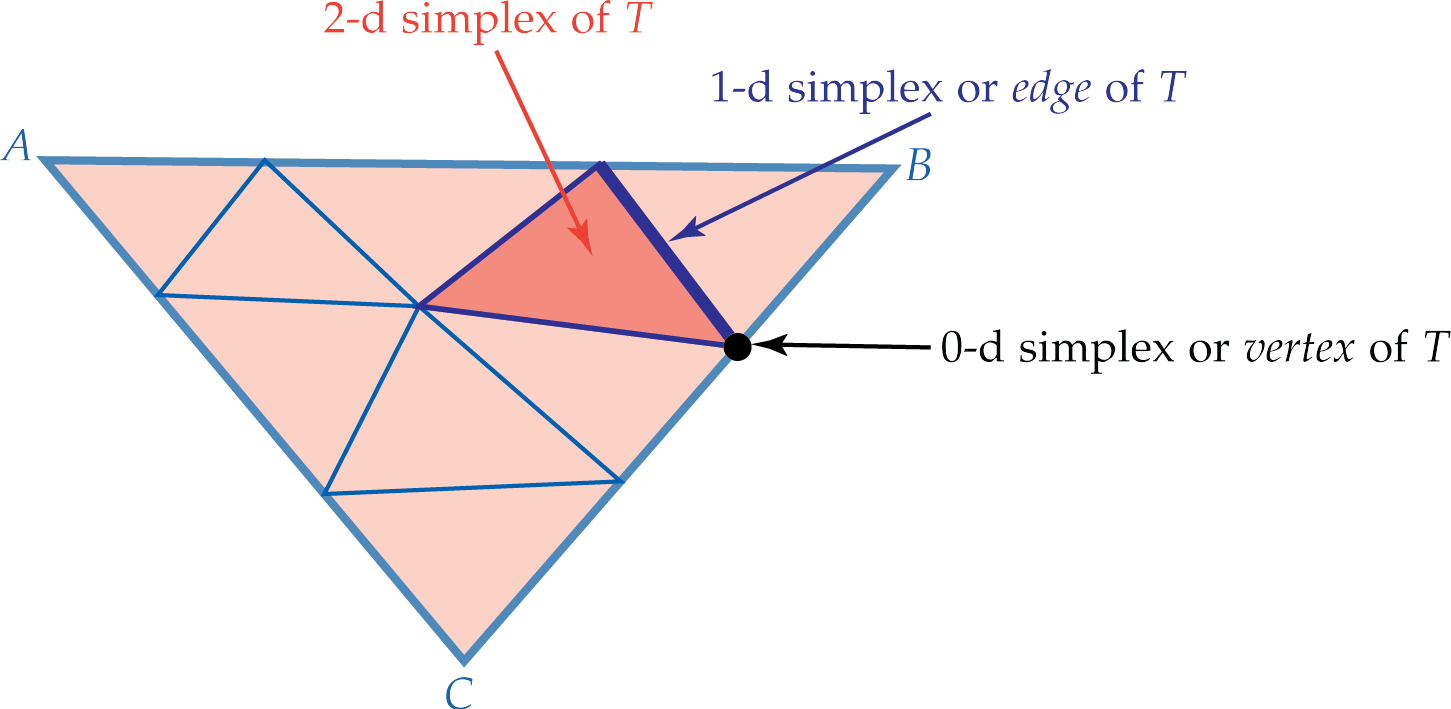}
    \caption{A simplicial complex $T$ that defines a triangulation of the triangle $A-B-C$.}
    \label{fig:complex}
\end{figure}

The notion of triangulation relates simplicial complexes with topological spaces. 
\begin{definition}[\textsc{Triangulation}]
A simplicial complex $K$ is a triangulation of a topological space $X$ if $\norm{K} \cong X$.
\end{definition}

For instance, the boundary of the $n$-simplex $\sigma_n$, namely a simplicial complex containing all proper faces of $\sigma_n$, is a triangulation of the sphere $S^{n-1}$.

Triangulation is a very powerful tool in studying the computational complexity of topological problems, because they allow us to partition a simplicial complex into smaller simplices that are connected in useful ways. We will mainly use the \emph{Kuhn triangulation}, which is described in more detail below. We refer the interested reader to \citep{Mat03BorsukUlam} and \citep{munkres1984} for further information.

We define two functions of a triangulation, $\idx$ and $\val$, that are essential for the definition of our computational problems and our reductions.

\begin{definition}[\textsc{Value \& Index Functions}] \label{def:indexValueCircuits}
    Let $K$ be a simplicial complex consisting of $k$
  simplices, including the non-full dimensional ones and
  let $M > k$. We define the \textit{value function}
  $\val: [M] \rightarrow \bar{K}$, where $\bar{K} = K \cup \{\emptyset\}$, to be an
  efficiently computable function such that
  \begin{enumerate}
  \item $\val$ is bijective on $K$,
    \item if $\val(x) = \{\vec{v}_1, \dots, \vec{v}_{\ell}\}$, then $x$ 
    is called the index of the simplex $\sigma \in K$ with vertices 
    $\{\vec{v}_1, \dots, \vec{v}_{\ell}\}$, and
    \item if $\val(x) = \emptyset$ then $x$ does not correspond to a valid 
    index of any non-empty simplex in $K$,
  \end{enumerate}
  
    We also define the \textit{index function} 
  $\idx:\bbR^n \rightarrow [M]$ to be an efficiently computable function such
  that if $\vec{x} \in \norm{K}$ then $\vec{x} \in \norm{\val(\idx(\vec{x}))}$. 
\end{definition}

Intuitively, $\val$ provides a way to enumerate over the simplices and $\idx$ on input a point $\bx$ returns the simplex that contains $\bx$. 

\bigskip

\begin{definition}[\textsc{Kuhn's Triangulation} \citep{kuhn1960some}]\label{def:kuhn}
Kuhn's triangulation is a standard way to triangulate a domain that is a cube. For any $n \in \mathbb{N}$, the cube $[0,1]^n$ is triangulated with granularity $m \in \mathbb{N}$ as follows:
\begin{enumerate}
    \item the set of vertices of the triangulation is $U_m^n$, where $U_m = \{0,1/m,2/m, \dots, m/m\}$
    \item every $x \in (U_m \setminus \{1\})^n$ is the base of the cube containing all vertices $\{y : y_i \in \{x_i,x_i + 1/m\}\}$
    \item every such cube is subdivided into $n!$ $n$-dimensional simplices as follows: for every permutation $\pi$ of $[n]$, $\sigma = \{y^0,y^1, \dots, y^n\}$ is a simplex, where $y^0 = x$ and $y^i = y^{i-1} + \frac{1}{m}e_{\pi(i)}$ for all $i \in [n]$ ($e_i$ is the $i$th unit vector)
\end{enumerate}
It is easy to see that Kuhn's triangulation has the following properties:
\begin{itemize}
    \item For any simplex $\sigma = \{z^1, \dots, z^k\}$ it holds that $\|z^i-z^j\|_\infty \leq 1/m$ for all $i,j$, and there exists a permutation $\pi$ of $[k]$ such that $z^{\pi(1)} \leq \dots \leq z^{\pi(k)}$ (component-wise).
    \item The restriction of Kuhn's triangulation of $[0,1]^n$ on some subspace $A_1 \times A_2 \times \dots \times A_n$ of $[0,1]^n$, where for each $i \in[n]$, $A_i \in \{\{0\}, [0,1]\}$, coincides with Kuhn's triangulation of that subspace.
    \item Every $n$-dimensional simplex can be indexed by its smallest vertex (component-wise), which is also the base of the cube containing the simplex, and by the permutation that yields this simplex within this cube. Given some index, it is easy to check whether this is a valid index, and if so, output the vertices of the simplex. Thus, the $\idx$ function can be computed efficiently.
    \item Given a point $x \in [0,1]^n$, we can efficiently determine the index of a simplex that contains it as follows. First find the base $y$ of a cube of $U_m^n$ that contains $x$. Next, find a permutation $\pi$ such that $x_{\pi(1)}-y_{\pi(1)} \geq \dots \geq x_{\pi(n)}-y_{\pi(n)}$. Then, it follows that $(y,\pi)$ is the index of a simplex containing $x$. Thus, the $\val$ function can be computed efficiently.
    \item Given an $n$-simplex $\{z^0, \dots, z^n\}$ and $i \in \{0,1,\dots, n\}$, we can efficiently compute the index of the other $n$-simplex that also has $\{z^0, \dots, z^n\} \setminus \{z_i\}$ as a facet. This is called a \emph{pivot} operation.
\end{itemize}
\end{definition}

\subsection{The Borsuk-Ulam Theorem and Tucker's Lemma}

Here, we provide the definitions of the problems that we generalize. Note that in \cref{sec:2dBSS}, we explain how the Borsuk-Ulam Theorem can be interpreted under a more general definition, which explains how our definition of Polygon Borsuk-Ulam is indeed a generalization. We start with the Borsuk-Ulam Theorem, which is usually stated as ``for every continuous function from $S^n$ to $\mathbb{R}^n$, there exists a point $\mathbf{x} \in \mathbb{R}^n$, such that $f(\mathbf{x})=f(-\mathbf{x})$''. We present an alternative definition (as stated in \citep{Mat03BorsukUlam}), which is more appropriate for our results.

\begin{theorem}[\textsc{Borsuk-Ulam Theorem} \citep{Borsuk1933,Mat03BorsukUlam}]
For every antipodal mapping $f:S^n \rightarrow \mathbb{R}^n$ (i.e., a function $f$ which is continuous and $f(\mathbf{x}) = f(\mathbf{-x})$), there exists a point $\mathbf{x} \in S^n$ such that $f(\mathbf{x})=0$. 
\end{theorem}

Tucker's lemma is a well-known combinatorial existence theorem, which is an analogue of the Borsuk-Ulam Theorem. It is usually stated on the $d$-dimensional unit ball or sphere. For computational purposes the following version is more commonly used.

\begin{theorem}[\textsc{Tucker's lemma} \citep{tucker1945some}]
Let $m,d \geq 1$. Let $\lambda: ([-m,m] \cap \mathbb{N})^d \to \{\pm 1, \pm 2, \dots, \pm d\}$ be any labeling that satisfies $\lambda(-x) = - \lambda(x)$ for all $x \in ([-m,m] \cap \mathbb{N})^d$ such that $\exists i$ with $|x_i|=m$. Then there exist two points $x,y \in ([-m,m] \cap \mathbb{N})^d$ with $\|x-y\|_\infty \leq 1$ that have opposite labels, i.e., $\lambda(x) = - \lambda(y)$.
\end{theorem}
The corresponding computational problem \textsc{Tucker} is known to be PPA-complete \citep{Papadimitriou94-TFNP-subclasses,ABB15-2DTucker}, even if the dimension is fixed to be $d=2$.

\subsection{$\mathbb{Z}_p$-equivariant tie-breaking}

For some of our constructions, we will require a tie-breaking that is $\mathbb{Z}_p$-equivariant and efficiently computable. Observe that non-efficient tie breaking rules exist by carefully choosing one representative for any equivalence class but this is not sufficient for our proofs. We define an efficiently computable rule below for any set $S \in 2^{\mathbb{Z}_p} \setminus \{\emptyset, \mathbb{Z}_p\}$, let $S+i := \{x+i : x \in S\}$.

\begin{definition}[\textsc{$\mathbb{Z}_p$-equivariant tie-breaking}]
\label{def:tie-breaking}
For any prime $p$, the $\mathbb{Z}_p$-equivariant tie-breaking function $T_p : 2^{\mathbb{Z}_p} \setminus \{\emptyset, \mathbb{Z}_p\} \to \mathbb{Z}_p$ is computed as follows on input $S \in 2^{\mathbb{Z}_p} \setminus \{\emptyset, \mathbb{Z}_p\}$:
\begin{enumerate}
    \item For every $i \in \bbZ_p$, write $S+i$ as a bit-string of length $p$. Namely, construct the bit-string $b(i)$ where the $j$th bit from the left indicates whether $j \in S+i$ (for $j=0, \dots, p-1$).
    \item Output $-i^* \in \mathbb{Z}_p$, where $i^* = \argmax_{i} b(i)$ ($b(i)$ interpreted as a number in binary).
\end{enumerate}
\end{definition}

\begin{lemma}
For any prime $p$, the $\mathbb{Z}_p$-equivariant tie-breaking function $T_p : 2^{\mathbb{Z}_p} \setminus \{\emptyset, \mathbb{Z}_p\} \to \mathbb{Z}_p$ is well-defined and satisfies for any $S \in 2^{\mathbb{Z}_p} \setminus \{\emptyset, \mathbb{Z}_p\}$:
\begin{itemize}
    \item $T_p(S) \in S$
    \item $T_p(S+i) = T_p(S)+i$ for all $i \in \mathbb{Z}_p$.
\end{itemize}
\end{lemma}

\begin{proof}
If $p$ is prime and $S+ i = S$ for some $i \in \bbZ_p \setminus\{0\}$, then $S \in \{\emptyset, \mathbb{Z}_p\}$. This follows from observing that the corresponding bit strings $b(0)$ and $b(i)$ must be equal and that this implies that the bits of $b(0)$ with index $\{k\cdot i\}_{k \in \bbZ_p}$ are all equal. Since $p$ is prime, $\{k\cdot i\}_{k \in \bbZ_p} = \bbZ_p$.

Hence, $|\{S+i: i \in \mathbb{Z}_p\}| = p$ for any $S \in 2^{\mathbb{Z}_p} \setminus \{\emptyset, \mathbb{Z}_p\}$. Thus, the bit-strings $b(i)$ are all distinct, $T_p(S)$ is unique and $T_p$ is well-defined. Next, by construction, it is easy to see that $T_p(S+i) = -(i^*-i) = T_p(S) + i$. Finally, since $S \neq \emptyset$, $i^*$ will be such that $b(i^*)$ has a $1$ in the left-most position. Thus, $0 \in S+i^*$, which implies that $-i^* \in S$.
\end{proof}

\paragraph{Example.} Let $p = 3$ and $S = \{2\}$. Then, $S + 0 = \{2\}$, $S + 1 = \{0\}$, $S + 2 = \{1\}$, and $b(0) = 010$, $b(1) = 001$ and $b(2) = 100$ (in binary). From \cref{def:tie-breaking}, $T_3(S) = 2$.

\section{$(p,n)$-\tucker~reduces to $(p,n+1)$-\tucker}

Let $(p,n)$-\tucker\ denote the $p$-\tucker\ problem with dimension parameter $n$. We have the following lemma.

\begin{lemma}
\label{lm:n1sufficies}
For all $n \geq 1$ and prime $p \geq 2$, $(p,n)$-\tucker~reduces to $(p,n+1)$-\tucker.
\end{lemma}

\begin{proof}
The domain of $(p,n)$-\tucker~with Kuhn's triangulation can be written as
$$X_n = \{(c^1, \dots, c^p) \in U^{np}_m | \forall j \in [n], \exists i \in [p] : c^i_j = 0 \}.$$

Note that the subset of $X_{n+1}$ corresponding to $c^1_{n+1} = \dots = c^p_{n+1} = 0$ can be identified with $X_n$. Since we use Kuhn's triangulation in both cases, the triangulations ``match''.

Let $\lambda$ be an instance of $(p,n)$-\tucker. We construct an instance $\lambda'$ of $(p,n+1)$-\tucker~as follows. For any vertex $(c^1, \dots, c^p)$, we set 
\[\lambda'(c^1, \dots, c^p) = \left\{ \begin{array}{ll} 
                    \lambda(c^1, \dots, c^p), &\text{if } c^1_{n+1} = \dots = c^p_{n+1}=0\\
                    (k,n+1) \text{ where } k = T_p(\argmax_{i \in [p]} c^i_{n+1}), &\text{otherwise}
                \end{array} \right.
\]
In the second case, we have used the tie-breaking rule defined in \cref{def:tie-breaking}. Since this tie-breaking is $\mathbb{Z}_p$-equivariant, it is easy to see that $\lambda'$ also satisfies the boundary conditions.

Consider any solution to this instance, i.e., a $(p-1)$-simplex $\sigma$ that has all labels $(1,\ell), \dots, (p,\ell)$ for some $\ell$.  If $\ell = n+1$, then there exists $i \in [p]$ such that $c^i_{n+1}=0$ \emph{for all} vertices of $\sigma$. This follows from the fact that the triangulation is ``nice''. But then, $\sigma$ cannot have the label $(i,n+1)$. Hence, it must be that $\ell < n+1$ and $\sigma$ is contained in the region identified with $X_n$ where $\lambda = \lambda'$. Thus, $\sigma$ also yields a solution to the original instance.
\end{proof}

\section{\pstartucker{p} is in \ppak{p}, Full Proof}\label{sec:app:ptucker-in-ppa-p}

In this section, we provide the full proof of \cref{thm:ptucker-in-ppa-p}. Namely, we show how to reduce \pstartucker{p} to \imba{p}.
\begin{proof}

The proof is a generalization of the construction given by \citet{FT81} for Tucker's lemma. The main difficulties in generalizing their approach are:
\begin{itemize}
    \item For $p=2$, when a path hits the boundary, there is a corresponding path that also hits the boundary on the antipodal side, and we can join the two endpoints. For $p > 2$, when a path hits the boundary, there are now $p-1$ other corresponding paths that also hit the boundary. We ensure that no solution occurs there by directing all the edges. We show how to direct the edges consistently and efficiently.
    \item For $p=2$, the original construction associates a label with each axis of the domain. For $p > 2$, there are more axes than labels, and so a single label must be associated to multiple axes. This creates imbalanced nodes in the graph that are not solutions. We solve this problem by carefully assigning weights to the edges of the graph.
\end{itemize}

\paragraph*{Sub-orthants.} Recall that $d=t(p-1)$. Consider the domain $R_{p,m}^d$ with a labeling function $\lambda$ given by a Boolean circuit. Let $T$ be Kuhn's triangulation of $R_{p,m}^d$ as described earlier. The domain $R_{p,m}^d$ can be subdivided into what we call sub-orthants, which are orthants of coordinate subspaces. Formally, a sub-orthant is a space of the form $A_1 \times A_2 \times \dots \times A_d$, where $A_\ell = \{*^{i_\ell} j : 0 \leq j \leq m\}$ or $A_\ell = \{0\}$ for $\ell=1, \dots, d$.

We associate a label to every axis of $R_{p,m}^d$ as follows. The label $*^i j$ is associated to the $*^i$-axis of the $[(j-1)(p-1) + \ell]$th copy of $R_{p,m}$, for $\ell = 1,2,\dots,p-1$. Thus, every label is associated to exactly $p-1$ axes. For any sub-orthant $X$, let $S(X)$ denote the set of labels associated with the axes that are used by $X$. For any simplex $\sigma$ of $T$, we let $O(\sigma)$ denote the smallest sub-orthant that contains $\sigma$.

For any sub-orthant $O$ and $j \in [t]$, let $r_j(O)$ be the number of coordinates in the range $(j-1)(p-1)+1, \dots, j(p-1)$ that are equal to $0$ in $O$. In particular, we have $\sum_{j=1}^t r_j(O) = t(p-1) - \dim(O)$. Note that if $|S(O)| = \dim (O)$, then $r_j(O) = p-1 - |\{i \in \mathbb{Z}_p : *^i j \in S(O)\}|$. We abuse notation and denote $r_j(\sigma) := r_j(O(\sigma))$.

\paragraph*{Happy simplices.}
Let $k \in \{0,1,\dots, d\}$. A $k$-dimensional simplex $\sigma$ of $T$ is \emph{happy}, if
\begin{enumerate}
    \item $\dim(O(\sigma)) = k$ \emph{($\sigma$ is full-dimensional in its sub-orthant)}
    \item $|S(O(\sigma))| = \dim (O(\sigma))$ \emph{(the sub-orthant uses $\leq 1$ axis associated with each label)}
    \item $S(O(\sigma)) \subseteq \lambda(\sigma)$ \emph{($\sigma$ carries all the labels associated with its sub-orthant)}
\end{enumerate}
A happy simplex is called \emph{super-happy} if we actually have $S(O(\sigma)) \subsetneq \lambda(\sigma)$. In particular, the $0$-dimensional simplex $0^d$ is super-happy.

Consider a happy simplex $\sigma$ that has a facet $\tau \subset \partial R_{p,m}^d$ such that $\lambda(\tau) = S(O(\sigma))$. Such a simplex is called a boundary-happy simplex. If $\sigma$ is such a simplex, then the simplex that has $\theta \tau$ as a facet is also boundary-happy. Thus, we group $\tau, \theta \tau, \dots, \theta^{p-1} \tau$ together into an equivalence class $[\tau]$. Every such equivalence class has size exactly $p$. Formally, let $B$ be the set of all simplices $\tau \subset \partial R_{p,m}^d$ such that $\lambda(\tau) = S(O(\tau))$ and $|S(O(\tau))| = \dim (O(\tau))$. We define an equivalence relation on $B$ by $\tau_1 \equiv \tau_2$ if and only if $\exists i \in \mathbb{Z}_p$ such that $\tau_1 = \theta^i \tau_2$.

We construct a graph $G$. The vertices of $G$ are the happy simplices of $T$ and the equivalence classes $[\tau]$ (for all $\tau \in B$).

\paragraph*{Orientation.}
We now define an orientation for our simplices. Fix an ordering of the labels, e.g., $*^1 1, *^2 1, \dots, *^p 1, *^1 2, \dots$. Let $\sigma$ be any happy $k$-simplex, $k \geq 1$, and let $x_0 x_1 \dots x_k$ be any ordering of its vertices. Let $\widehat{x}_0, \widehat{x}_1, \dots, \widehat{x}_k \in \mathbb{N}^k$ denote the coordinates of $x_0, x_1, \dots, x_k$ in $O(\sigma)$, where the coordinates are ordered according to the fixed ordering of the labels. Note that there is at most one coordinate associated to each label (because $|S(O(\sigma))| = k$) and thus the coordinate vectors are uniquely determined. Furthermore, note that the coordinates are non-negative. We define the $k \times k$ matrix $M = [\widehat{x}_0-\widehat{x}_1, \widehat{x}_0 - \widehat{x}_2, \dots, \widehat{x}_0-\widehat{x}_k]$, i.e., the $i$th column is $\widehat{x}_0 - \widehat{x}_i$. Then, we define the orientation of happy simplex $\sigma$ with ordering $x_0x_1\dots x_k$ as $\textup{or}(\sigma|x_0 \dots x_k) = \det M$. Note that by construction of the triangulation $T$, we always have $\det M \in \{-1,+1\}$. Indeed, it is easy to check that $M$ can be transformed into the identity matrix by elementary operations that can only change the sign of the determinant.

\paragraph*{Edges.}
Let $\sigma$ be a happy $k$-simplex, $k \geq 1$. If $\sigma$ is super-happy, then it has a single facet $\tau$ such that $\lambda(\tau) = S(O(\sigma))$. If it is not super-happy, then it has exactly two facets $\tau_1$ and $\tau_2$ such that $\lambda(\tau_i) = S(O(\sigma))$, $i=1,2$. In any case, any such facet $\tau$ of $\sigma$ yields an edge as follows:
\begin{itemize}
    \item If $\tau$ does not lie in the boundary of the sub-orthant $O(\sigma)$, then there is exactly one other $k$-simplex $\sigma'$ in $O(\sigma)$ that also has $\tau$ as its facet. $\sigma'$ is also happy, and we put an edge between $\sigma$ and $\sigma'$.
    \item If $\tau$ lies in the boundary of the sub-orthant $O(\sigma)$, there are two cases:
    \begin{itemize}
        \item $\tau$ lies in $\partial R_{p,m}^d$. In that case, $\sigma$ is a boundary-happy simplex and we put an edge between $\sigma$ and $[\tau]$.
        \item $\tau$ does not lie in $\partial R_{p,m}^d$. Then, $\tau$ is a super-happy $(k-1)$-simplex and we put an edge between $\sigma$ and $\tau$.
    \end{itemize}
\end{itemize}
In all of these cases, the direction of the edge is determined as follows. Let $x_0x_1 \dots x_k$ be the ordering of the vertices of $\sigma$ such that $\tau = \{x_1, \dots, x_k\}$ and $x_1 \dots x_k$ are ordered according to their labels. If $\textup{or}(\sigma|x_0 \dots x_k) = 1$, then the edge is incoming into $\sigma$. Otherwise, it is outgoing out of $\sigma$. Finally, the weight of the edge is always $\prod_{j=1}^t r_j(\sigma)!$.

By the definition, it follows that there are three types of edges:
\begin{itemize}
    \item (Type 1) An edge between two happy $k$-simplices $\sigma_1, \sigma_2$ that lie in the same sub-orthant and share a facet $\tau$ with $\lambda(\tau) = S(O(\sigma_i))$, $i=1,2$.
    \item (Type 2) An edge between a happy simplex $\sigma$ and its super-happy facet $\tau$ such that $\lambda(\tau) = S(O(\sigma))$.
    \item (Type 3) An edge between a boundary-happy simplex $\sigma$ and its facet equivalence class $[\tau]$.
\end{itemize}

An edge of Type 2 or 3 is always ``created'' by exactly one of its endpoints. Thus, its direction and weight are well-defined. An edge of Type 1 however, is ``created'' by both of its endpoints and we will prove that it is well-defined, i.e., both endpoints agree on its direction and weight. For the weight this is easy to see, since $O(\sigma_1)=O(\sigma_2)$ implies that $r_j(\sigma_1) = r_j(\sigma_2)$ for all $j$. For the direction, we postpone this consistency check to the end of the proof.

We now prove that all vertices of $G$ that do not yield a solution are balanced modulo $p$, except $0^d$.

\paragraph*{The trivial solution $0^d$.}
We have $\lambda(0^d) = *^i j$. There are exactly $p-1$ sub-orthants $O$ such that $S(O) = \{*^i j\}$ and each of them contains a happy 1-simplex $\sigma$ that has $0^d$ as a facet. It follows that $0^d$ has $p-1$ edges, each with weight $((p-1)!)^t/(p-1)$. Furthermore, all the edges are outgoing, because $\textup{or}(\sigma|x_0 0^d) = 1$ (i.e., incoming into $\sigma$). It follows that the total imbalance of $0^d$ is $(p-1) ((p-1)!)^t/(p-1) = ((p-1)!)^t = (-1)^t \mod p$, where we used $(p-1)! = -1 \mod p$ since $p$ is prime (Wilson's theorem). It follows that $0^d$ is always a valid trivial solution for \textsc{Imbalance-mod-$p$}, because $(-1)^t \neq 0 \mod p$ for all $t \geq 1$.

\paragraph*{Happy, but not super-happy.}
Consider a happy $k$-simplex $\sigma$ with two facets $\tau_1, \tau_2$ that satisfy $\lambda(\tau_i) = S(O(\sigma))$, $i=1,2$. Then, $\sigma$ has two edges and they both have the same weight. Let $x_0x_1 \dots x_k$ be the ordering of $\sigma$ such that $\tau_1 = \{x_1, \dots, x_k\}$ and $x_1\dots x_k$ are ordered according to their labels. Let $i \in [k]$ be the index such that $x_i$ and $x_0$ have the same label. In particular, $\tau_2 = \{x_1, \dots, x_{i-1},x_0,x_{i+1},\dots, x_k\}$ and $x_1 \dots x_{i-1}x_0x_{i+1}\dots x_k$ are ordered according to their labels. We have
\begin{equation*}\begin{split}
    \det [\widehat{x}_0-\widehat{x}_1, \widehat{x}_0 - \widehat{x}_2, \dots, \widehat{x}_0-\widehat{x}_k] &= \det [\widehat{x}_i-\widehat{x}_1, \dots, \widehat{x}_i - \widehat{x}_{i-1}, \widehat{x}_0 - \widehat{x}_i, \widehat{x}_i - \widehat{x}_{i+1}, \dots, \widehat{x}_i-\widehat{x}_k]\\
    &= - \det [\widehat{x}_i-\widehat{x}_1, \dots, \widehat{x}_i - \widehat{x}_{i-1}, \widehat{x}_i - \widehat{x}_0, \widehat{x}_i - \widehat{x}_{i+1}, \dots, \widehat{x}_i-\widehat{x}_k]
\end{split}\end{equation*}
where we first subtracted the $i$th column from all other columns, and then multiplied the $i$th column by $-1$. It follows that $\textup{or}(\sigma|x_0 \dots x_k) = - \textup{or}(\sigma|x_i x_1 \dots x_{i-1} x_0 x_{i+1} \dots x_k)$. Thus, one edge is incoming and the other outgoing, i.e., $\sigma$ is balanced.

\paragraph*{Equivalence class.}
Consider an equivalence class $[\tau]$. Let $\sigma$ be the happy $k$-simplex that has $\tau$ as a facet. In $G$, $[\tau]$ has exactly $p$ edges: one with each of $\sigma, \theta \sigma, \theta^2 \sigma, \dots, \theta^{p-1} \sigma$. Since $S(O(\theta^i \sigma)) = \theta^i S(O(\sigma))$ for all $i$, it follows that $r_j(\theta^i \sigma) = r_j(\sigma)$ for all $i,j$. Thus, all $p$ edges have the same weight. Let $x_0 \dots x_k$ be the ordering of $\sigma$ such that $\tau = \{x_1, \dots, x_k\}$ and $x_1\dots x_k$ are ordered according to their labels. Let $y_i = \theta x_i$ for all $i$. Then, $y_1 \dots y_k$ might not be ordered according to their labels. We let $\pi$ denote the permutation that we would have to apply to order them correctly. As before, $\widehat{x}_i$ denotes the coordinates of $x_i$ restricted to $O(\sigma)$, where the coordinates are ordered according to the associated label. $\widehat{y}_i$ denotes the coordinates of $y_i$ restricted to $O(\theta \sigma)$, where the coordinates are ordered according to the associated label. Since the associated labels have changed according to $\theta$, it follows that if we re-order the coordinates of $\widehat{x}_i$ according to $\pi$ we obtain $\widehat{y}_i$ for all $i=0,1,\dots, k$. Thus, we have
\begin{equation*}\begin{split}
\textup{or}(\theta \sigma|y_0 \pi(y_1 \dots y_k)) &= \textup{sgn}(\pi) \det [\widehat{y}_0-\widehat{y}_1, \widehat{y}_0 - \widehat{y}_2, \dots, \widehat{y}_0-\widehat{y}_k]\\
&= \textup{sgn}(\pi)^2 [\widehat{x}_0-\widehat{x}_1, \widehat{x}_0 - \widehat{x}_2, \dots, \widehat{x}_0-\widehat{x}_k]\\
&= \textup{or}(\sigma|x_0 \dots x_k)
\end{split}\end{equation*}
It follows that all edges of $[\tau]$ are directed the same way, i.e., they are all incoming or all outgoing. Since there are $p$ edges and they also have the same weight, it follows that $[\tau]$ has imbalance $0$ modulo $p$.
In this argument we assumed that $\lambda$ satisfies the boundary conditions. Thus, if $[\tau]$ is not balanced modulo $p$, we obtain a counter-example, which is a solution.

\paragraph*{Super-happy.}
Consider a super-happy $k$-simplex $\sigma$, $k \geq 1$. Note that $\sigma$ has a single facet $\tau$ that satisfies $\lambda(\tau) = S(O(\sigma))$. Thus, $\sigma$ ``creates'' a single edge. Let $x_0 \dots x_k$ be the ordering of $\sigma$ such that $\tau = \{x_1, \dots, x_k\}$ and $x_1 \dots x_k$ are ordered according to their labels. The edge has weight $\prod_{j=1}^t r_j(\sigma)!$ and it is incoming if $\textup{or}(\sigma|x_0 \dots x_k) = 1$, outgoing otherwise.

Since $\sigma$ is super-happy, we have $\lambda(x_0) \notin \lambda(\tau) = S(O(\sigma))$. Let $*^i \ell = \lambda(x_0)$. If $\{*^j \ell | j \in [p] \setminus \{i\}\} \subseteq S(O(\sigma))$, or equivalently if $r_\ell(\sigma) = 0$, then $\sigma$ yields a solution. Otherwise, there are exactly $r_\ell(\sigma)$ different sub-orthants $O$ such that $O(\sigma) \subset O$ and $S(O) = S(O(\sigma)) \cup \{*^i \ell\}$. Thus, there are exactly $r_\ell(\sigma)$ happy $(k+1)$-simplices $\rho$ such that $\sigma$ is a facet of $\rho$ and $\rho$ is happy because of $\sigma$ (i.e., $S(O(\rho))=\lambda(\sigma)$). It follows that $\sigma$ has $r_\ell(\sigma)$ additional edges (apart from the one it ``created''). Each of these edges has weight $\prod_{j=1}^t r_j(\rho)! = (r_\ell(\sigma))^{-1} \prod_{j=1}^t r_j(\sigma)!$, since $r_\ell(\rho) = r_\ell(\sigma)-1$. Thus, if these $r_\ell(\sigma)$ edges have opposite direction to the edge ``created'' by $\sigma$ (from the perspective of $\sigma$), $\sigma$ will be balanced.

Consider any such $\rho$. Let $x_0\dots x_{k+1}$ be the ordering of $\rho$ such that $\sigma = \{x_1, \dots, x_{k+1}\}$ and $x_1 \dots x_{k+1}$ are ordered according to their labels. Let $t \in [k+1]$ be the index of label $*^i \ell$ if we order the labels in $S(O(\rho))$. Then, we also have that $\tau = \sigma \setminus \{x_t\}$. Furthermore, $x_1 \dots x_{t-1} x_{t+1} \dots x_{k+1}$ are also ordered according to their labels. From the perspective of $\sigma$, the edge it ``created'' is directed according to $\textup{or}(\sigma|x_t x_1 \dots x_{t-1} x_{t+1} \dots x_{k+1})$ and the edge created by $\rho$ is directed according to $- \textup{or}(\rho|x_0 x_1 \dots x_{k+1})$. As before, let $\widehat{x}_j$ denote the coordinates of $x_j$ restricted to $O(\rho)$, where the coordinates are ordered according to the associated label. Let $\bar{x}_j$ denote the coordinates of $x_j$ restricted to $O(\sigma)$, where the coordinates are ordered according to the associated label. Note that if we remove the $t$th coordinate from $\widehat{x}_j$, then we obtain $\bar{x}_j$. We now have
\begin{equation*}\begin{split}
\textup{or}(\rho|x_0 x_1 \dots x_{k+1}) &= \det [\widehat{x}_0-\widehat{x}_1, \dots, \widehat{x}_0-\widehat{x}_{k+1}]\\
&= \det [\widehat{x}_t-\widehat{x}_1, \dots, \widehat{x}_t-\widehat{x}_{t-1}, \widehat{x}_0-\widehat{x}_t, \widehat{x}_t-\widehat{x}_{t+1}, \dots, \widehat{x}_t-\widehat{x}_{k+1}]\\
&= (-1)^{t+t} \det [\bar{x}_t-\bar{x}_1, \dots, \bar{x}_t-\bar{x}_{t-1}, \bar{x}_t-\bar{x}_{t+1}, \dots, \bar{x}_t-\bar{x}_{k+1}]\\
&= \textup{or}(\sigma|x_t x_1 \dots x_{t-1} x_{t+1} \dots x_{k+1})
\end{split}\end{equation*}
where we first subtracted the $t$th column from all other columns, and then we used Laplace's determinant formula along the $t$th row. Note that the $t$th entry in $\widehat{x}_t-\widehat{x}_j$ is $0$ for all $j \in [k+1]$ and it is $1$ in $\widehat{x}_0-\widehat{x}_t$.

\paragraph*{Consistency.}
The only thing that remains to be checked is that edges of Type 1 are well-defined, in terms of the direction. Let $\sigma_1,\sigma_2$ be two happy $k$-simplices that lie in the same sub-orthant and share a facet $\tau$ with $\lambda(\tau) = S(O(\sigma_i))$, $i=1,2$. Let $\{x_1, \dots, x_k\} = \tau$ and $x_1 \dots x_k$ be the ordering according to their labels. Let $\{x_0\} = \sigma_1 \setminus \tau$ and $\{x_0'\} = \sigma_2 \setminus \tau$. We want to show that $\textup{or}(\sigma_1|x_0 x_1 \dots x_k)$ and $\textup{or}(\sigma_2|x_0' x_1 \dots x_k)$ have opposite signs. This can be proved directly combinatorially by using the way the triangulation is constructed, but we provide a proof that is more general here. Let $\phi: \mathbb{R}^k \to \mathbb{R}^k$ be the unique linear function such that $\phi(\widehat{x}_0-\widehat{x}_1)=\widehat{x}_0'-\widehat{x}_1$ and $\phi(\widehat{x}_1-\widehat{x}_i)=\widehat{x}_1-\widehat{x}_i$ for all $i \in [k] \setminus \{1\}$. $\phi$ is unique, because $\widehat{x}_0-\widehat{x}_1, \widehat{x}_1-\widehat{x}_2, \dots, \widehat{x}_1-\widehat{x}_k$ form a basis of $\mathbb{R}^k$. $\phi$ is the identity function on the hyperplane given by $\widehat{x}_1-\widehat{x}_2, \dots, \widehat{x}_1-\widehat{x}_k$ and maps $\widehat{x}_0-\widehat{x}_1$ to $\widehat{x}_0'-\widehat{x}_1$. Since $x_0$ and $x_0'$ lie on opposite sides of the hyperplane defined by $\tau$, it follows that $\widehat{x}_0-\widehat{x}_1$ and $\widehat{x}_0'-\widehat{x}_1$ lie on opposite sides of the hyperplane on which $\phi$ is the identity. It follows that $\det \phi < 0$. Thus, we can write
\begin{equation*}\begin{split}
\det [\widehat{x}_0'-\widehat{x}_1, \dots, \widehat{x}_0'-\widehat{x}_k] &= \det [\widehat{x}_0'-\widehat{x}_1, \widehat{x}_1-\widehat{x}_2 \dots, \widehat{x}_1-\widehat{x}_k]\\
&= \det \phi \det [\widehat{x}_0-\widehat{x}_1, \widehat{x}_1-\widehat{x}_2 \dots, \widehat{x}_1-\widehat{x}_k]\\
&= \det \phi \det [\widehat{x}_0-\widehat{x}_1, \dots, \widehat{x}_0-\widehat{x}_k]
\end{split}\end{equation*}
and the claim follows.
\end{proof}

\end{document}